%% file: Main.tex
\documentclass[a4paper]{article}

\input{Macros}
\input{prooftree}
\usepackage{authblk}

\title{On reduction and normalization in the computational core }
\author[1]{Claudia Faggian}
\author[2]{Giulio Guerrieri}
\author[3]{Ugo de'Liguoro}
\author[4]{Riccardo Treglia}

\affil[1]{\small Université de Paris Cité, IRIF, CNRS, F-75013 Paris, France\\
	{\tt faggian@irif.fr }}
\affil[2]{\small Aix Marseille Univ, CNRS, LIS UMR 7020, Marseille, France\\
	{\tt giulio.guerrieri@lis-lab.com}}
\affil[3]{\small Universit\`a di Torino, Department of Computer Science, Turin, Italy\\
	{\tt ugo.deliguoro@unito.it}}
\affil[4]{\small Universit\`a di Bologna, Department of Computer Science, Bologna, Italy\\
	{\tt riccardo.treglia@unibo.it}}
\date{\vspace{-5ex}}

\begin{document}
	
	\maketitle
	
	\begin{abstract}
	We study the reduction in a $\lambda$-calculus derived from Moggi's computational one, which we call the computational core.
	The reduction relation consists of rules obtained by orienting three monadic laws. Such laws, in particular associativity and identity, introduce intricacies in the operational analysis.	We investigate the central notions of returning a value versus having a normal form, and address the question of normalizing strategies. 
	Our analysis relies on factorization results.
	\end{abstract}

\input{Introduction}
\input{Background}
\input{Syntax}
\input{Translation}
\input{Operational}
\input{Factorization}

\input{Values}

\input{Normalization}

\input{Discussion_and_Related}

\bibliographystyle{alpha} 
\bibliography{references}

\newpage
\appendix
\section*{APPENDIX}

\input{Appendix}

\input{App_Normalization}

\newpage
\end{document}

%% file: Macros.tex
\usepackage{relsize}
\usepackage[utf8]{inputenc}
\usepackage{amsfonts}
\usepackage{amsmath}
    \usepackage{amssymb}
    \usepackage{bm}
\usepackage{amsthm}
\usepackage{amsbsy}
\usepackage{mathrsfs}
\usepackage{enumerate}
\usepackage{amsmath}
\usepackage{url}
\usepackage{ebproof}
\usepackage{proof}
\usepackage{color}
\usepackage{graphicx}
\usepackage{ebproof}
\usepackage{cmll}
\usepackage{stmaryrd}
\usepackage{proof}
\usepackage{wrapfig}
\usepackage{bm} 
\usepackage{diagrams}
\usepackage[hidelinks]{hyperref}
\usepackage[capitalize]{cleveref}
\Crefname{enumi}{Point}{Points}
\crefname{enumi}{Point}{Points}
\Crefname{lem}{Lemma}{Lemmas}
\crefname{lem}{Lemma}{Lemmas}
\Crefname{thmm}{Theorem}{Theorems}
\crefname{thmm}{Theorem}{Theorems}

\usepackage{thmtools,thm-restate}
\usepackage{mathtools}
\usepackage{xspace}
\usepackage{cancel}
\usepackage[normalem]{ulem}

\newcommand{\CF}[1]{{\color{violet}{**CF: #1 **}}}

\newtheorem{defn}{Definition}[section]
\newtheorem{fact}[defn]{Fact}
\newtheorem{thmm}[defn]{Theorem}
\newtheorem{coroll}[defn]{Corollary}
\newtheorem{lem}[defn]{Lemma}
\newtheorem{prop}[defn]{Proposition}
\newtheorem{exmpl}[defn]{Example}
\newtheorem{property}[defn]{Fact}
\newtheorem{rem}[defn]{Remark}
\newtheorem*{notation*}{Notation}

\newtheorem*{theorem*}{Theorem}
\newcommand{\deff}{\,\coloneqq\,}
\newcommand{\eq}{\,=\,}

\newcommand{\ssubseteq}{\, \subseteq \,}
\newcommand{\defeq}{\coloneqq}
\newcommand{\eqdef}{\eqqcolon}

\renewcommand{\l}{\lambda}
\newcommand{\la}[1]{\lambda #1.}

\newcommand{\TFlambdaComp}{\lambda_{\normalfont\scalebox{.6}{\copyright}}}
\newcommand{\FV}{\mathsf{fv}}

\newcommand{\Unit}{\mbox{\it unit}\;}

\newcommand{\Bind}{\star}

\newcommand{\Betac}{\beta_c}
\newcommand{\BindRight}{\textit{id}}
\newcommand{\BindLeft}{\sigma}
\newcommand{\Subst}[2]{[ #1 / #2]}

\newcommand{\ValTerm}{\textit{Val}}
\newcommand{\ComTerm}{\textit{Com}}
\newcommand{\Term}{\textit{Term}}

\newcommand{\Red}{\red}
\newcommand{\RedStar}{\red^*}

\newcommand{\mlsym}{\,ml^*}
\newcommand{\ml}{\lambda_{ml^*}}

\newcommand{\redlc}{\red_{\normalfont\scalebox{.6}{\copyright}}}
\newcommand{\redlcStar}{\red^*_{\normalfont\scalebox{.6}{\copyright}}}
\newcommand{\redlceta}{\red_{\normalfont\scalebox{.6}{\copyright}\eta}}
\newcommand{\redlcetaStar}{\red^*_{\normalfont\scalebox{.6}{\copyright}\eta}}

\newcommand{\eqlceta}{=_{\normalfont\scalebox{.6}{\copyright}\eta}}
\newcommand{\eqlc}{=_{\normalfont\scalebox{.6}{\copyright}}}
\newcommand{\eqlm}{=_{\textit{ml*}}}

\newcommand{\CCsym}{{\mathsmaller{\copyright}}}

\newcommand{\lc}{\lambda_{\normalfont\scalebox{.6}{\copyright}}}
\newcommand{\lcc}{\normalfont\scalebox{.6}{\copyright}}

\newcommand{\Bang}{\Lambda^!}

\newcommand{\Com}{\ComTerm}

\newcommand{\abs}{\mathtt{Abs}}

\newcommand{\II}{\mathbf{I}}

\newcommand{\bI}{\pmb{I}}
\newcommand{\bDelta}{\pmb{\Delta}}
\newcommand{\Rule}{\rho}

\newcommand{\id}{\mathsf{id}}
\newcommand{\ii}{\iota}

\newcommand{\betac}{\beta_c}
\newcommand{\redbc}{\red_{\betac}}
\newcommand{\sredbc}  {\mathrel{\sredx{\betac}}}
\newcommand{\nsredbc}{\mathrel{\nsredx{\betac}}}
\newcommand{\betab}{\beta_\oc}
\newcommand{\bbeta}{\betab}
\newcommand{\betav}{\beta_v}
\newcommand{\bs}{\betac\sigma}

\newcommand{\xrevredx}[2] {\mathrel{{\uset{#1}{\leftarrow}}{}_{\mkern-3mu#2}}}
\newcommand{\revred}{\leftarrow}

\newcommand{\rredbb}{\mapsto_{\bbeta}}
\newcommand{\rredc}{\mapsto_{\gamma}}
\newcommand{\xredx}[2]{{\mathrel{\uset{#1~}{\red}}{}_{\mkern-8mu#2}}}
\newcommand{\nxredx}[2]{{\mathrel{\uset{#1~}{\nrightarrow}}{}_{\mkern-8mu#2}}}
\newcommand{\redi}{\red_{\ii}}

\newcommand{\redid}{\redx{\id}}
\newcommand{\sredid}{\sredx{\id}}

\newcommand{\wredbc}{\wredx{\betac}}

\newcommand{\sreds}{\sredx{\sigma}}
\newcommand{\wreds}{\wredx{\sigma}}
\newcommand{\rredbc}{\Root{\betac}}

\newcommand{\cc}{\textsf {C}}  
\newcommand{\vc}{\textsf {V}}  
\newcommand{\gc}{\textsf {G}}
\newcommand{\ww}{\textsf {W}}  
\renewcommand{\ss}{\textsf {S}} 
\newcommand{\leftc}{\textsf {L}} 
\newcommand{\rightc}{\textsf {R}} 
\newcommand{\ee}{\textsf {E}}  

\newcommand{\letsf}{\mathsf{let}}
\newcommand{\insf}{\mathsf{in}}
\newcommand{\Let}[3]{\letsf \, #1 \!\defeq\! #2 \, \insf \, #3}

\newcommand{\LetSym}{\textsf{let}}
\newcommand{\tlet}{\LetSym}

\newcommand{\Var}{\textit{Var}}

\newcommand{\traa}[1]{\traap{(#1)}}
\newcommand{\trb}[1]{\trbp{(#1)}}
\newcommand{\trc}[1]{\trcp{(#1)}}
\newcommand{\trap}[1]{#1^\dagger}
\newcommand{\traap}[1]{#1^\ddagger}
\newcommand{\trbp}[1]{#1^\bullet}
\newcommand{\trcp}[1]{#1^\circ}

\makeatletter
\newcommand{\circlearrow}{}
\DeclareRobustCommand{\circlearrow}{%
	\mathrel{\vphantom{\rightarrow}\mathpalette\circle@arrow\relax}%
}
\newcommand{\circle@arrow}[2]{%
	\m@th
	\ooalign{%
		\hidewidth$#1\circ\mkern1mu$\hidewidth\cr
		$#1\longrightarrow$\cr}%
}
\makeatother

\newcommand{\LP}[2]{\mathtt{SP(#1,#2)}}
\newcommand{\F}[2]{\mathtt{Fact(#1,#2)}}
\newcommand{\PP}[2]{\mathtt{PP(#1,#2)}}

\newcommand{\ex}{\mathsf {e}}
\renewcommand{\int}{\mathsf{i}}

\newcommand{\weak}{\mathsf{w} }
\newcommand{\surf}{\mathsf{s}}
\newcommand{\lsym}{\mathsf{l}}
\newcommand{\rsym}{\mathsf{r}}

\newcommand{\lex}{\mathsf{le}}

\newcommand{\Nat}{\mathbb{N}}

\usepackage{xcolor}

\newcommand{\RED}[1]{{\color{red}{#1}}}

\newcommand{\violet}[1]{{\color{violet}{#1}}}

\newcommand{\fv}[1]{\FV(#1)}
\newcommand{\size}[1]{\mathsf{s}(#1)}
\newcommand{\sizeaux}[1]{\mathsf{s}_\sigma(#1)}

\newcommand{\hole}[1]{\langle #1\rangle}

\newcommand{\tmtwo}{u}

\newcommand{\tm}{t}
\newcommand{\tms}{s}
\newcommand{\tmu}{u}

\newcommand{\normal}{\mathsf{N}}
\newcommand{\normalweak}{\normal_\ww}
\newcommand{\normalsurface}{\normal_\ss}
\newcommand{\normalfull}{\normal}
\newcommand{\appnormal}{\mathsf{A}}
\newcommand{\appnormalweak}{\appnormal_\ww}
\newcommand{\appnormalsurface}{\appnormal_\ss}
\newcommand{\appnormalfull}{\appnormal}

\renewcommand{\to}{\xrightarrow{}}

\newcommand{\Root}[1]{\mapsto_{#1}}

\newcommand{\redx}[1]{\mathrel{\red{}_{\mkern-8mu#1}}}

\newcommand{\red}{\rightarrow}

\newcommand{\ered}{{\xredx{\ex}{}}}
\newcommand{\ired}{{\xredx{\int}{}}}

\newcommand{\nered}{\uset{\neg \ex~}{\red}}
\newcommand{\neredx}[1]{\mathrel{\nered{}_{\mkern-8mu#1}}}

\newcommand{\wred}{{\xredx{\weak}{}}}
\newcommand{\nwred}{\xredx{\neg \weak}{}}
\newcommand{\wredx}[1]  {\xredx{\weak}{#1}}
\newcommand{\nwredx}[1]  {\xredx{\neg\weak}{#1}}
\newcommand{\wredbv}{\wredx{\betav}}

\newcommand{\lredbv}{\mathrel{\xredx{\lsym}{\betav}}}
\newcommand{\nlredbv}{\mathrel{\xredx{\neg\lsym~}{\betav}}}

\newcommand{\sred}{\xredx{\surf}{}}

\newcommand{\sredx}[1]{\xredx{\surf}{#1}}
\newcommand{\nsredx}[1]{\xredx{\neg\surf}{#1}}

\newcommand{\xredxStar}[2]{{\mathrel{\uset{#1~}{\red}}{}^{{\mkern-8mu{*}}}_{\mkern-8mu#2}}}
\newcommand{\xrevredxStar}[2]{{\mathrel{\uset{#1~}{\leftarrow}}{}^{{\mkern-8mu{*}}}_{\mkern-8mu#2}}}

\newcommand{\wredxStar}[1]  {\xredxStar{\weak}{#1}\,}
\newcommand{\nwredxStar}[1]  {\xredxStar{\neg\weak}{#1}\,}

\newcommand{\sredxStar}[1]{\xredxStar{\surf}{#1}\, }
\newcommand{\nsredxStar}[1]{\xredxStar{\neg\surf}{#1}\, }

\newcommand{\sredbb}  {\mathrel{\sredx{\bbeta}}}
\newcommand{\nsredbb}{\mathrel{\nsredx{\bbeta}}}

\newcommand{\redbv}{\red_{\beta_v}}

\newcommand{\redb}{\rightarrow_{\beta}}
\newcommand{\eredx}[1]  {\mathrel{\ered{}_{\,\mkern-8mu#1}}}
\newcommand{\iredx}[1]  {\mathrel{\ired{}_{\mkern-8mu#1}}}

\newcommand{\redbb}{\rightarrow_{\bbeta}}

\newcommand{\redbcs}{\rightarrow_{\betac\sigma}}
\newcommand{\reda}{\red_{\xi}}
\newcommand{\nsreda}{\nsredx{\xi}}
\newcommand{\nsredc}{\nsredx{\gamma}}

\newcommand{\sredc}{\sredx{\gamma}}
\newcommand{\rreda}{\mapsto_{\xi}}
\newcommand{\ereda}  {\mathrel{\eredx{\xi}}}

\newcommand{\ireda}  {\mathrel{\iredx{\xi}}}

\newcommand{\redc}  {\red_{\gamma}}
\newcommand{\eredc} {\mathrel {\eredx{\gamma}}}
\newcommand{\iredc}  {\mathrel{\iredx{\gamma}}}

\newcommand{\reds} {\red_{\sigma}}

\newcommand{\lered}{\xredx{\lex}{}}
\newcommand{\leredx}[1]{\xredx{\lex}{#1}}

\newcommand{\lam}{\lambda}

\newcommand{\ie}{\textit{i.e.}\xspace}
\newcommand{\eg}{\textit{e.g.}\xspace}
\newcommand{\ih}{\textit{i.h.}\xspace}

\newcommand{\subs}[2]{ [#2/#1] }

\makeatletter
\newcommand{\uset}[3][0ex]{%
	\mathrel{\mathop{#3}\limits_{
			\vbox to#1{\kern-6\ex@
				\hbox{$\scriptstyle#2$}\vss}}}}
\makeatother

\RequirePackage{ifthen}

\newcommand{\version}{0}
\newcommand{\commenti}{0} 
\newcommand{\condinc}[2]{\ifthenelse{\equal{\commenti}{0}}{#1}{**\violet{#2}} }
\newcommand{\SLV}[2]{\ifthenelse{\equal{\version}{0}}{#1}{ \RED{#2}}}

%% file: prooftree.tex
\newdimen\proofrulebreadth \proofrulebreadth=.05em
\newdimen\proofdotseparation \proofdotseparation=1.25ex
\newdimen\proofrulebaseline \proofrulebaseline=2ex
\newcount\proofdotnumber \proofdotnumber=3
\let\then\relax
\def\hfi{\hskip0pt plus.0001fil}
\mathchardef\squigto="3A3B
\newif\ifinsideprooftree\insideprooftreefalse
\newif\ifonleftofproofrule\onleftofproofrulefalse
\newif\ifproofdots\proofdotsfalse
\newif\ifdoubleproof\doubleprooffalse
\let\wereinproofbit\relax
\newdimen\shortenproofleft
\newdimen\shortenproofright
\newdimen\proofbelowshift
\newbox\proofabove
\newbox\proofbelow
\newbox\proofrulename
\def\shiftproofbelow{\let\next\relax\afterassignment\setshiftproofbelow\dimen0 }
\def\shiftproofbelowneg{\def\next{\multiply\dimen0 by-1 }%
\afterassignment\setshiftproofbelow\dimen0 }
\def\setshiftproofbelow{\next\proofbelowshift=\dimen0 }
\def\setproofrulebreadth{\proofrulebreadth}

\def\prooftree{
\ifnum  \lastpenalty=1
\then   \unpenalty
\else   \onleftofproofrulefalse
\fi
\ifonleftofproofrule
\else   \ifinsideprooftree
        \then   \hskip.5em plus1fil
        \fi
\fi
\bgroup
\setbox\proofbelow=\hbox{}\setbox\proofrulename=\hbox{}%
\let\justifies\proofover\let\leadsto\proofoverdots\let\Justifies\proofoverdbl
\let\using\proofusing\let\[\prooftree
\ifinsideprooftree\let\]\endprooftree\fi
\proofdotsfalse\doubleprooffalse
\let\thickness\setproofrulebreadth
\let\shiftright\shiftproofbelow \let\shift\shiftproofbelow
\let\shiftleft\shiftproofbelowneg
\let\ifwasinsideprooftree\ifinsideprooftree
\insideprooftreetrue
\setbox\proofabove=\hbox\bgroup$\displaystyle 
\let\wereinproofbit\prooftree
\shortenproofleft=0pt \shortenproofright=0pt \proofbelowshift=0pt
\onleftofproofruletrue\penalty1
}

\def\eproofbit{
\ifx    \wereinproofbit\prooftree
\then   \ifcase \lastpenalty
        \then   \shortenproofright=0pt  
        \or     \unpenalty\hfil         
        \or     \unpenalty\unskip       
        \else   \shortenproofright=0pt  
        \fi
\fi
\global\dimen0=\shortenproofleft
\global\dimen1=\shortenproofright
\global\dimen2=\proofrulebreadth
\global\dimen3=\proofbelowshift
\global\dimen4=\proofdotseparation
\global\count255=\proofdotnumber
$\egroup  
\shortenproofleft=\dimen0
\shortenproofright=\dimen1
\proofrulebreadth=\dimen2
\proofbelowshift=\dimen3
\proofdotseparation=\dimen4
\proofdotnumber=\count255
}

\def\proofover{
\eproofbit 
\setbox\proofbelow=\hbox\bgroup 
\let\wereinproofbit\proofover
$\displaystyle
}%
\def\proofoverdbl{
\eproofbit 
\doubleprooftrue
\setbox\proofbelow=\hbox\bgroup 
\let\wereinproofbit\proofoverdbl
$\displaystyle
}%
\def\proofoverdots{
\eproofbit 
\proofdotstrue
\setbox\proofbelow=\hbox\bgroup 
\let\wereinproofbit\proofoverdots
$\displaystyle
}%
\def\proofusing{
\eproofbit 
\setbox\proofrulename=\hbox\bgroup 
\let\wereinproofbit\proofusing
\kern0.3em$
}

\def\endprooftree{
\eproofbit 
  \dimen5 =0pt
\dimen0=\wd\proofabove \advance\dimen0-\shortenproofleft
\advance\dimen0-\shortenproofright
\dimen1=.5\dimen0 \advance\dimen1-.5\wd\proofbelow
\dimen4=\dimen1
\advance\dimen1\proofbelowshift \advance\dimen4-\proofbelowshift
\ifdim  \dimen1<0pt
\then   \advance\shortenproofleft\dimen1
        \advance\dimen0-\dimen1
        \dimen1=0pt
        \ifdim  \shortenproofleft<0pt
        \then   \setbox\proofabove=\hbox{%
                        \kern-\shortenproofleft\unhbox\proofabove}%
                \shortenproofleft=0pt
        \fi
\fi
\ifdim  \dimen4<0pt
\then   \advance\shortenproofright\dimen4
        \advance\dimen0-\dimen4
        \dimen4=0pt
\fi
\ifdim  \shortenproofright<\wd\proofrulename
\then   \shortenproofright=\wd\proofrulename
\fi
\dimen2=\shortenproofleft \advance\dimen2 by\dimen1
\dimen3=\shortenproofright\advance\dimen3 by\dimen4
\ifproofdots
\then
        \dimen6=\shortenproofleft \advance\dimen6 .5\dimen0
        \setbox1=\vbox to\proofdotseparation{\vss\hbox{$\cdot$}\vss}%
        \setbox0=\hbox{%
                \advance\dimen6-.5\wd1
                \kern\dimen6
                $\vcenter to\proofdotnumber\proofdotseparation
                        {\leaders\box1\vfill}$%
                \unhbox\proofrulename}%
\else   \dimen6=\fontdimen22\the\textfont2 
        \dimen7=\dimen6
        \advance\dimen6by.5\proofrulebreadth
        \advance\dimen7by-.5\proofrulebreadth
        \setbox0=\hbox{%
                \kern\shortenproofleft
                \ifdoubleproof
                \then   \hbox to\dimen0{%
                        $\mathsurround0pt\mathord=\mkern-6mu%
                        \cleaders\hbox{$\mkern-2mu=\mkern-2mu$}\hfill
                        \mkern-6mu\mathord=$}%
                \else   \vrule height\dimen6 depth-\dimen7 width\dimen0
                \fi
                \unhbox\proofrulename}%
        \ht0=\dimen6 \dp0=-\dimen7
\fi
\let\doll\relax
\ifwasinsideprooftree
\then   \let\VBOX\vbox
\else   \ifmmode\else$\let\doll=$\fi
        \let\VBOX\vcenter
\fi
\VBOX   {\baselineskip\proofrulebaseline \lineskip.2ex
        \expandafter\lineskiplimit\ifproofdots0ex\else-0.6ex\fi
        \hbox   spread\dimen5   {\hfi\unhbox\proofabove\hfi}%
        \hbox{\box0}%
        \hbox   {\kern\dimen2 \box\proofbelow}}\doll%
\global\dimen2=\dimen2
\global\dimen3=\dimen3
\egroup 
\ifonleftofproofrule
\then   \shortenproofleft=\dimen2
\fi
\shortenproofright=\dimen3
\onleftofproofrulefalse
\ifinsideprooftree
\then   \hskip.5em plus 1fil \penalty2
\fi
}


%% file: Introduction.tex

\newcommand{\Kcomp}{\bullet}
\newcommand{\name}{computational core\xspace}

\section{Introduction}
\label{sect:intro}

The $\lambda$-calculus has been historically conceived as an equational theory of functions, so that reduction
had an ancillary role in Church's view, and it was a tool for studying the theory $\beta$, see \cite[Ch.~3]{Barendregt'84}. 
%
The development of functional programming languages like Lisp and ML, and of proof assistants like LCF, has brought  a new, different  interest in the
$\lambda$-calculus and its \emph{reduction} theory.

 The cornerstone of this change in perspective  is Plotkin's
\cite{Plotkin'75}, where the functional parameter passing mechanism 
is formalized by the {\em call-by-value rewrite rule $\beta_v$}, allowing reduction
only if the argument term is a {\em value}, that is a variable or an abstraction. In \cite{Plotkin'75} it is also introduced
the notion  of {\em weak evaluation}, namely no reduction in the body of a function (\ie, of an abstraction).
This  is now  the standard evaluation  implemented by functional programming languages, where  \emph{values} are the terms of interest (and the normal forms for weak evaluation in the closed case).
Full $\betav$ reduction is instead the basis of  proof assistants like Coq, where  \emph{normal forms} are the result of interest. More generally,  the computational perspective on  
$\lambda$-calculus has given a central role to reduction, whose theory  provides  
a  sound framework for reasoning about program transformations, such as compiler optimizations or parallel implementations.

The rich variety of computational effects in actual implementations of functional programming languages brings  further challenges. This dramatically affects
the theory of reduction of the calculi formalizing such features, whose proliferation makes it difficult to focus on 
suitably general issues. A major change here is the discovery by Moggi \cite{Moggi'88, Moggi'89, Moggi'91} of a whole family of calculi that are based on a few common traits,
combining call-by-value with the abstract notion of effectful computation represented by a {\em monad}, which has shown to be quite successful. 
But Moggi's {\em computational $\lambda$-calculus} is an equational theory in the broader sense; much less is known of the reduction
theory of such calculi: this is the  focus of our paper.

\paragraph{The Computational Calculus.}
Since Moggi's seminal work, computational $\lambda$-calculi have been developed as a foundation of programming languages, formalizing both functional and non-functional features, see \textit{e.g.} \cite{WadlerT03,BentonHM00}, starting a thread in the literature that is still growing. 
The basic idea of computational $\lambda$-calculi is to distinguish {\em values} and {\em computations}, so that programs, 
represented by closed terms, are thought of as functions from values to computations. 
Intuitively, computations embody a richer structure than values and do form a larger set in which values can be embedded. On the other hand,
the essence of programming is composition; to compose functions from values to computations we need a mechanism
to uniformly extend them to functions of computations, while preserving their original behavior over the (image of) values.

To model these concepts, Moggi used the categorical notion of \emph{monad}, abstractly representing
the extension of the space of values to that of computations, and the associated Kleisli category, whose morphisms 
are functions from values to computations, which are the denotations of programs. Syntactically, following \cite{Wadler95},
we can express these ideas by means of a call-by-value $\lambda$-calculus with two sorts of terms: \emph{values}, ranged over by $V,W$, 
namely variables or abstractions, and \emph{computations} denoted by $L,M,N$. Computations are formed by means of two
operators: values are embedded into computations by means of the operator
$\Unit\!\!$ written {\tt return} in Haskell programming language, whose name refers to the unit of a monad in categorical terms; 
a computation $M \Bind (\lambda x.N)$ is formed by the binary operator
$\Bind$, called {\em bind} ($\texttt{>>=}$ in Haskell), 
representing the application to $M$ of the extension to computations of the function $\lambda x.N$. 

\paragraph{The Monadic Laws.}
The operational understanding of these new operators is that evaluating
$M \Bind (\lambda x.N)$, which in Moggi's notation reads $\Let{x}{M}{N}$, amounts to first evaluating $M$ until a computation 
of the form $\Unit V$ is reached, representing the trivial computation that returns the value $V$. Then $V$ is passed to $N$ by binding $x$ to $V$,
as expressed by the identity:
\begin{equation}\label{eq:beta-c}
(\Unit V) \Bind \lambda x.N = N\Subst{V}{x}
\end{equation}
This is the first of the three \emph{monadic laws} in \cite{Wadler95}.
The remaining laws are:
\begin{align}
\label{eq:id}
M \Bind \lambda x. \Unit x &= M \\ 
\label{eq:sigma}
(L \Bind \lambda x.M) \Bind \lambda y.N &= L \Bind \lambda x.(M \Bind \lambda y.N) \quad \text{with } x \notin \fv{N}
\end{align}

 To understand these two last rules, let us define the composition (named Kleisli composition in category theory) of the functions
$\lambda x.M$ and $\lambda y.N$ as
\begin{equation*}\label{eq:Kcomp}
	(\lambda x.M) \Kcomp (\lambda y. N) \coloneqq \lambda x. (M \Bind (\lambda y.N))
\end{equation*}
where we can freely assume that $x$ is not free in $N$.

Equality \eqref{eq:id} (\emph{identity}) implies that $(\lambda z.M) \Kcomp (\lambda x. \Unit x) = \lambda z.M$, which paired with the instance of (\ref{eq:beta-c}):
$(\lambda x. \Unit x) \Kcomp (\lambda y.N) = \lambda x. N\Subst{x}{y} =_\alpha \lambda y.N$ (where $=_\alpha$ is 
the usual congruence generated by the renaming of bound variables), 
tells that $\lambda x. \Unit x$
is the  
identity of composition $\Kcomp$. 

Equality \eqref{eq:sigma} (\emph{associativity}) implies:
\[((\lambda z.L) \Kcomp (\lambda z.M)) \Kcomp (\lambda y.N) = (\lambda z.L) \Kcomp ((\lambda z.M) \Kcomp (\lambda y.N))\]
namely that composition $\Kcomp$ is associative.

\medskip
The monadic laws correspond to the three equalities in the definition of a Kleisli triple   \cite{Moggi'91}, which is an equivalent presentation of monads \cite{MacLane97}. 
Indeed,
Moggi's calculus is the internal language of a suitable category equipped with a (strong) monad $T$, and with enough structure
to internalize the morphisms of the respective Kleisli category. As such, it is a simply typed $\lambda$-calculus, where $T$ is the type constructor associating with each
type $A$ the type $TA$ of computations over $A$. Therefore, $\Unit\!$ and $\Bind$ are polymorphic operators with respective types \cite{Wadler92,Wadler95}:
\begin{equation}\label{eq:unit-bind-types}
\Unit \colon A \to TA \qquad \qquad \Bind: TA \to (A \to TB) \to TB
\end{equation}

\paragraph{The Computational Core.} The dynamics of $\lambda$-calculi is usually defined as a reduction relation on untyped terms.  Moggi's preliminary report \cite{Moggi'88} specifies both an equational and, in \S 6, a \emph{reduction} system  even if only the former 
is thoroughly investigated and appears in \cite{Moggi'89,Moggi'91}, while reduction is briefly treated for an untyped fragment of the calculus.
However, when stepping from the typed calculus to the untyped one, we need to be careful by avoiding meaningless 
terms to creep into the syntax, so jeopardizing the calculus theory.
For example: what should be the meaning of $M \Bind N$ where both $M$ and $N$ are computations?
What about $(\lambda x.N) \Bind V$ for any $V$? Shall we have functional applications of any kind?

To answer these questions, in  \cite{deLiguoroTreglia20}   typability  is taken as syntactic counterpart of being meaningful:  inspired by  ideas in \cite{Scott80},  the untyped computational $\lambda$-calculus is a special case of the typed one, where there are just two types $D$ and $TD$,
related by the type equation $D = D \to TD$, that is Moggi's isomorphism of the call-by-value reflexive object (see \cite{Moggi'88}, \S 5).
With such a proviso, we get the following syntax:
\begin{align*}
	V,W & \Coloneqq x\mid \lam x.M &\qquad {(\ValTerm)}\\
	M,N, L & \Coloneqq \Unit V\mid M \Bind V& \qquad{(\ComTerm)}
\end{align*}

If we assume that all variables have type $D$, then it is easy to see that all terms in $\ValTerm$ have type $D = D \to TD$, which is consistent with
the substitution of variables with values in (\ref{eq:beta-c}). On the other hand, considering the typing of $\Unit$ and $\Bind$ in (\ref{eq:unit-bind-types}), 
terms in $\ComTerm$ have type $TD$. As we have touched above, there is some variety in notation among computational $\lambda$-calculi; we choose the
above syntax because it explicitly embodies the essential constructs of a $\lambda$-calculus with monads, but for functional application, which is
definable: see \Cref{sec:compcal} for further explanations. We dub the calculus {\em computational core}, noted $\lc$.

\paragraph{From Equalities to Reduction.}

Similarly to \cite{Moggi'88} and \cite{SabryWadler97}, the 
reduction  rules in the computational core $\lc$ are  the relation obtained by orienting the monadic laws 
from left to right. We indicate by $\betac$, $\id$, and $\sigma$ the rules corresponding to \eqref{eq:beta-c}, \eqref{eq:id} and \eqref{eq:sigma}, respectively.
 The contextual closure of these  rules, noted $\redlc$, has been  proved  confluent in \cite{deLiguoroTreglia20}, which implies 
that equal terms have a common reduct and the uniqueness of normal forms.

In \cite{Plotkin'75} call-by-value
reduction $\redbv$ is an intermediate
concept between the equational theory and the evaluation relation $\wredx{\betav}$, that  models an abstract machine. 
Evaluation consists of persistently choosing the leftmost $\betav$-redex that is not in the scope of an abstraction, 
\ie evaluation  is \emph{weak}.  
The following crucial result bridges reduction (hence, the foundational calculus) with  evaluation (implemented by 
an ideal programming language):
\begin{equation}\label{eq:cor1}
	M\xredxStar{}{\beta_v}\! V \mbox{ (for some value  $V$)} \mbox{ if and only if } M\wredxStar{\betav\!} V' \mbox{ (for some value $V'$)}
\end{equation}
Such a result (Corollary 1 in \cite{Plotkin'75}) comes from an analysis  of the \emph{reduction} properties of $\redbv$, 
namely standardization.

As we will see, the rules induced by associativity and identity  make the behavior of the reduction in $\lc$---and the 
study of its operational properties---\emph{non-trivial} in the setting of any monadic $\lambda$-calculus.
The issues are inherent to the  rules coming from the  monadic laws  \eqref{eq:id} and \eqref{eq:sigma}, independently of 
the  syntactic representation of the calculus that internalizes them.
The difficulty appears clearly  if we want to   follow  a similar route to  \cite{Plotkin'75}, as we discuss next.

\paragraph{Reduction vs. Evaluation.} 
Following \cite{Felleisen'88},  reduction $\redlc$ and evaluation $\wredx{\lcc}$ of $\lc$ can be defined 
as the closure of  the reduction rules under arbitrary and evaluation contexts, respectively. Consider the following grammars: 
%
\begin{align*}
\cc & \Coloneqq \hole{\,} \mid   \Unit (\lambda x.\cc)  \mid \cc \Bind V \mid M \Bind (\lambda x.\cc) &\qquad \text{(arbitrary) contexts}\\
\ee &\Coloneqq  \hole{\,} \mid \ee \Bind V &\qquad \text{evaluation contexts} 
\end{align*}
where the hole $\hole{\,}$ can be filled by terms in $\ComTerm$, only.
Observe that  the closure under evaluation context $\ee$ is precisely   weak reduction.

Weak reduction of $\lc$, however, turns out to be non-deterministic, non-confluent, and its normal forms are \emph{not unique}. The following is a counterexample to all  such properties---see \Cref{sec:operational} for further examples.

\begin{diagram}[width=2em,height=2em]
	((\Unit z \Bind z) \Bind \lambda x.\,M) \Bind \lambda y.\, \Unit y  & &  \rTo_{\weak} &  &	(\Unit z \Bind z) \Bind \lambda x.\,(M \Bind \lambda y.\, \Unit y)  \\
	\dTo_{\weak} & & & & \\
	(\Unit z \Bind z) \Bind \lambda x.\,M &  & & &
\end{diagram}

Such an  issue is not specific to the syntax of  the \name. The same phenomena show up with the {\em let}-notation, more commonly used in computational calculi. 
Here, evaluation, usually   called  \emph{sequencing},  is the reduction defined by the following  contexts \cite{Filinski:phd1996,JonesSLT98,LevyPT03}:
\[ \ee_{\textit{let}} \Coloneqq \hole{\,} \mid \Let{x}{\ee_{\textit{let}}}{N}.\]
Examples similar to the one  above can be reproduced. We give the details  in \Cref{ex:sequencing}.

\subsection{Content and  Contributions}
The focus of this paper   is an \emph{operational}  analysis of  two crucial  properties of a term $M$: 
	\begin{enumerate}[(i)]
		\item 	$M$ returns a \emph{value}  (\ie $M\redlc^* \Unit V$, for some  $V$ value).
	
\item 	  $M$ has a \emph{normal form} (\ie $M\redlc^* N$, for some  $N$ $\lcc$-normal).

	\end{enumerate}
As in \cite{Acc19}, the cornerstone of our analysis are  \emph{factorization} results (also called \emph{semi-standardization} in the literature): any reduction sequence can be re-organized  so as to  first performing  specific steps  and then everything else.

Via factorization, we show the key result \eqref{eq:weak},  analogous  to \eqref{eq:cor1},  relating reduction and \emph{evaluation}:
\begin{equation}\label{eq:weak}
\!\!	M\RedStar_{\lcc} \Unit V \mbox{ (for some value  $V$)} \iff 
M\wredxStar{\betac} \Unit V' \mbox{ (for some value $V'$)}
\end{equation}

We then analyze the property of having a normal form  (\emph{normalization}), and define  a family of \emph{normalizing strategies}, \ie subreductions that are guaranteed to reach a normal form, if any exists.

\paragraph{On the Rewrite Theory of Computational Calculi.} In this paper we study the rewrite theory of a specific computational calculus, namely $\lc$. 
	We expose a number of issues, which we argue to be  intrinsic to the monadic rules of computational calculi, namely associativity and identity. Indeed, the same issues which we expose   in $\lc$, also appear  in other computational calculi, as we discuss in  \Cref{sec:operational}, where we take as reference the calculus in \cite{SabryWadler97}, which we recall in \Cref{subsect:computational-let}. 
	We expect that  the solutions we propose for $\lc$ could be adapted also there.

\paragraph{Surface Reduction.} The form of weak reduction which we defined in the previous section (\emph{sequencing}) is standard in the literature. 
In this paper we study also a less strict form of weak reduction,  namely \emph{surface reduction}, which is less constrained and better behaved then sequencing.   Surface reduction disallows reduction under the $\Unit$ operator only, and not under abstractions. Intuitively, weak reduction  does not act in the body of a function, while surface reduction does not act in the scope of   {\tt return}.   As we  discuss in \Cref{subsect:computational-let}, it can also be seen as a more natural extension of call-by-value weak reduction to a computational calculus. 
	
Surface  reduction is well studied in the literature because it naturally arises when interpreting $\lambda$-calculus into linear logic, and indeed the name surface (which we take from \cite{Simpson05}) is reminiscent of a similar notion in calculi based on linear logic \cite{Simpson05, EhrhardGuerrieri16}. In \Cref{sec:bang} we will make explicit  the 
correspondence with such calculi, showing that the $\Unit\!$ operator (from the computational core) behaves exactly like a bang $\oc$ (from linear logic).


\paragraph{Identity and Associativity.} 
Our analysis exposes the operational role of the rules associated to the monadic laws of
identity and associativity.
 \begin{enumerate}[(i)]
 	\item To compute a \emph{value},   only $\betac$ steps are necessary.

	\item  To compute a   \emph{normal form},  $\betac$ steps do not suffice: \emph{associativity}  (\ie $\sigma$ steps)  is  \emph{necessary}.

 \end{enumerate}

Hence, the rule associated to the identity law turns out to be operationally \emph{irrelevant}.

\paragraph{Normalization.}
The study of {normalization} is more complex than that of evaluation, and requires some sophisticated techniques. We highlight some   specific  contributions.
\begin{itemize}
	\item  We define two families of \emph{normalizing strategies} in $\lc$. 
	The first one, quite constrained,  relies on an\emph{ iteration of weak reduction} $\wredx{\lc}$.
	The second one, more liberal, is based on an \emph{iteration of surface reduction} $\sredx{\lc}$.
	The  definition and proof of normalization is \emph{parametric} on either. 
	
	\item The technical \emph{difficulty} in the proofs for normalization comes from the fact that neither weak nor surface reduction is deterministic. To deal with that  
	we  rely on a fine  \emph{quantitative analysis} of the {number of $\betac$ steps,} which we carry-on when we study   factorization in  \Cref{sec:factorizations}. 
\end{itemize} 
The most  challenging  proofs in the paper are those related to normalization via surface reduction. The effort is justified by 
the interest in a larger and \emph{more versatile} strategy, which then does not induce a single abstract machine but \emph{subsumes} several ones, each  following a different  reduction policy. It thus facilitates reasoning about 
optimization techniques and    parallel  implementation. 


%

\paragraph{A Roadmap.}
Let us  summarize the structure of the paper.

\Cref{sec:background} contains the    background notions which are relevant to our paper.

\Cref{sec:compcal} gives the \textbf{formal definition of the \name}  $\lc$ and its reduction.

In  \Cref{sec:betac} and \Cref{sec:operational}, we  analyze the \textbf{properties of weak and surface reduction}.  We first  study $\redbc$,  and then  we move to the  whole  $\lc$, where associativity and identity  also come to play, and issues appear.

In \Cref{sec:factorizations} we study several  \textbf{factorization} results. The   cornerstone of our construction is 
surface factorization (\Cref{thm:sfactorization}).  We then further refine this result, first by postponing the $\id$ steps which are not $\betac$ steps, and then with a form of weak factorization. 

In \Cref{sec:values} we study  \textbf{evaluation},  and analyze some relevant consequences of this result. 
 We actually provide two different ways to deterministically compute a value. 
 The first way is the one given by \eqref{eq:weak}, via an   \emph{evaluation context}.
The second way requires no contextual closure at all: simply applying $\betac$- and $\sigma$-\emph{rules} will return a value, if possible.

In \Cref{sec:normalization}  we study \textbf{normalization and normalizing strategies}. 

	
\Cref{sec:related} concludes with final discussions and related work.

%% file: Background.tex
\section{Preliminaries}\label{sec:background}

\subsection{Basics on  Rewriting}\label{sec:rewriting}
We recall here some standard definitions and notations in rewriting  that we shall use in this paper
(see for instance Terese \cite{Terese03} or Baader and Nipkow \cite{BaaderN98} for details).

\paragraph{Rewriting System.} 
\renewcommand{\AA}{A}

An \emph{abstract rewriting system (ARS)} is a pair $(\AA, \red)$ 
consisting of a set $A$ and a binary relation $\red$ on $A$ whose pairs are written  $t \to s$ and 
called \emph{steps}. 
A \emph{$\red$-sequence} from $\tm \in \AA$ is a sequence $(t_i \red t_{i+1})_{i \in I}$ of $\red$ steps, where $I = \Nat$ or $I = \{0, 1, \dots, n-1\}$ for some $n \in \Nat$, $t_i \in \AA$ for all $i \in I$ and $t_0 = t$ (in particular, the sequence is empty for $ I = \emptyset$, \ie $n = 0$). 
We denote by  $\red^*$ (resp. $ \red^= $; $\red^+$) the  transitive-reflexive  (resp. reflexive; transitive) closure of 
$\red$, and $\leftarrow$ stands for the \emph{transpose} of $\red$, that is,  $\tmtwo \leftarrow \tm$ if  $\tm\red\tmtwo$.  
We write $t \red^k s$ for a $\red$-sequence $t \red t_1 \red \dots \red t_k \allowbreak= s$ of $k \in \Nat$ steps.
If $\red_1,\red_2$ are binary relations on $\AA$ then $\red_1\cdot\red_2$ denotes their composition,
\ie  $\tm \red_1\cdot\red_2 \tms$ if there exists  $\tmu\in \AA$ such that $\tm \red_1  \tmu \red_2 \tms$.
We often set $\red_{12} \, \defeq \, \red_1 \cup \red_2$.

 A relation $\red$ is \emph{deterministic} if for each $t\in \AA$ there is at most one $s\in \AA$ such that $t\red s$.
 It 
 is \emph{confluent} if ${\revred}^{*}\cdot {\red}^{*}
 {~\subseteq~}{\red} ^{*}\cdot\, {\revred}^{*}$. 
 \SLV{}{
Two relations $\red_{1}$ and $\red_{2}$  on $\AA$
\emph{commute} if  
${\revred_1}^{*}\cdot {\red_2}^{*}
{~\subseteq~}{\red_2} ^{*}\cdot\, {\revred_1}^{*}$.
A relation $\red$ on $\AA$ is \emph{confluent} if it commutes with itself. 
}

We  say that $u \in \AA$ is $\red$-\emph{normal} (or a \emph{$\red$-normal form}, noted $u \not \red$) if $u \not\red t$ for all $t \in \AA$, that is, there is no $t \in \AA$ such that 
$u\red t$; and we say that $t \in \AA$ \emph{has a normal form} $u$ if $t \red^* u$ with $u$ $\red$-normal. 
Confluence implies that  each $\tm\in \AA$ has \emph{unique normal form}, if any exists.

%
%


\paragraph{Normalization.}
Let  $(\AA,\red)$ be an ARS.
In general, a term may or may not reduce to a normal form. 
And if it does, not all reduction sequences necessarily lead to a normal form. 
A term is 
\emph{weakly} or \emph{strongly normalizing}, depending on if it  may or must reduce to normal form.
If a term $\tm$ is strongly normalizing, any choice of steps will eventually lead to a normal form. 
However, if $\tm$ is weakly normalizing, how do we compute a normal form? This is the problem tackled by \emph{normalization} and \emph{normalizing strategies}:  by repeatedly performing \emph{only specific  steps},  a normal form will be computed, provided that $\tm $ can reduce to~any.
We recall two important notions of normalization.

\begin{defn}[Normalizing]\label{def:normalizing_terms} Let  $(\AA,\red)$ be an ARS and $t \in A$.
\begin{enumerate}
	\item $\tm$ is  \emph {strongly $\red$-normalizing} (or {\textbf{terminating}}) if every maximal $\red$-sequence from $\tm$ ends in a normal form (\ie, $\tm$ has no infinite $\red$-sequence).
	\item 	 $\tm$ is   \emph{weakly $\red$-normalizing} (or just \textbf{normalizing}) if there exists a $\red$-sequence from $\tm$ that ends in a $\red$-normal form (\ie, $t$ has a $\red$-normal form).
\end{enumerate}
Reduction $\red$ is \emph{strongly} (resp.~\emph{weakly}) \emph{normalizing} if so is every $\tm\in \AA$.
Reduction $\red$ is \textbf{uniformly normalizing} if every weakly $\red$-normalizing $t \in \AA$ is also strongly $\red$-normalizing.
\end{defn}
 
Clearly, strong normalization implies weak normalization, and any deterministic reduction is uniformly normalizing.

A  \emph{normalizing strategy}  for $\red$ is a reduction strategy which, given a term $\tm$, is guaranteed to reach its $\red$-normal form, if any exists.

\begin{defn}[Normalizing  strategies]
	\label{def:strategy}
A subreduction $\ered\subseteq \red$ is a  \textbf{normalizing strategy} for  $\red$ 
 if $\ered$ has the same normal forms as $\red$, and for all $t \in A$, if $t$ has a $\red$-normal form, then
 \emph{every} maximal $\ered$-sequence from $t$ ends in a $\red$-normal form.
\end{defn}

Note that in \Cref{def:strategy}, $\ered$ need not be deterministic, and $\ered$ and $\red$ need not be confluent.

\paragraph{Factorization.} In this paper, we will extensively use factorization results.
%
%
%
%
%
%

\begin{defn}[Factorization, postponement]\label{def:factorization}
	Let $(A,\to)$ be an ARS with $\red \ \eq \ered \cup  \ired $. 
	
	Relation   $\red$  satisfies  \emph{$\ex$-factorization}, written $\F{\ered}{\ired}$, if
	\begin{equation}\tag{\textbf{Factorization}}
		\F{\ered}{\ired}: \quad (\ered \cup  \ired)^*~ \subseteq ~\ered^* \cdot \ired^*  
	\end{equation}	

	Relation $\ired$ \textbf{postpones} after $\ered$,  written $\PP{\ered}{\ired}$, if
\begin{equation}\tag{\textbf{Postponement}}
	\PP{\ered}{\ired}: \quad \ired^*\cdot \ered^*  ~\subseteq~  \ered^* \cdot \ired^*	
\end{equation}
\end{defn}	

It is  an easy result that
$\ex$-factorization is  equivalent to postponement, which is a more convenient way to express it. 

\begin{lem}		\label{lem:postponement_eq}
	The following  are equivalent (for any two relations $\ered,\ired$):
	\begin{enumerate}
		\item \emph{Postponement}:  $\PP{\ered}{\ired}$.
		\item \emph{Factorization}: $ \F{\ered}{\ired} $.
	\end{enumerate}
	
\end{lem}

Hindley \cite{HindleyPhD} first noted  that a local 
property implies 
factorization. Let $\red = \ered \cup \ired$.
We say that  $\ired$ \emph{strongly postpones} after $\ered$,
if
\begin{equation}\label{eq:SP}\tag{\textbf{Strong Postponement}}
	\LP{\ered}{\ired}:	\qquad	
\ired \cdot \ered ~\subseteq~\ered^*\cdot  \ired^=
\end{equation}


\begin{lem}[Hindley \cite{HindleyPhD}]
\label{lem:SP}\label{l:SP} 
$\LP{\ered}{\ired}$ implies  $\F{\ered}{\ired}$.
\end{lem}


Observe that the following  are special cases of strong postponement. 
The first one  is  \emph{linear} in $\ered$; we refer to it   as \emph{linear postponement}.
In the second one, recall that $\red \ \eq \ered \cup  \ired $.
\begin{enumerate}
	\item 
	$  \ired \cdot \ered \ssubseteq  \ered \cdot \ired^=   $\,.

	\item 
$  \ired \cdot \ered \ssubseteq  \ered \cdot \red  $\,.
\end{enumerate}	

Linear variants of postponement can easily be adapted to  \emph{quantitative} variants, which  allow us to ``count the steps'' and are useful to establish termination properties. We do this in \Cref{sec:w_factorization}.

\paragraph{Diamonds.}
We recall also another \emph{quantitative} result, which we will use.
	\begin{fact}[Newman \cite{Newman}]\label{fact:diamond} In an ARS $(\AA,\red)$, 	if $\red$ is quasi-diamond, then it  has random descent, where quasi-diamond and random descent are defined below.
		\begin{enumerate}	
			\item  \textbf{Quasi-Diamond}: For all $\tm\in \AA$,   if $\tm_1\leftarrow \tm \rightarrow \tm_2$, then $\tm_1=\tm_2$ or  $\tm_1\rightarrow \tmu \leftarrow \tm_2$ for some $\tmu$.	
			\item  \textbf{Random Descent}:  For all $\tm\in \AA$, all maximal $\red$-sequences from $\tm$ have the \emph{same number of steps}, and all end in the same normal form, if  any exists.
		\end{enumerate}
	\end{fact}

Clearly, if $\red$ is quasi-diamond then it is confluent and uniformly normalizing. 

\paragraph{Postponement, Confluence and Commutation.}
\newcommand{\decr}[1]{\langle #1 \rangle}

\newcommand{\redL}[1]{\triangleright_{#1}}
\newcommand{\revredL}[1]{\triangleleft_{#1}}
\newcommand{\redM}[1]{\blacktriangleright_{#1}}
\newcommand{\revredM}[1]{\blacktriangleleft_{#1}}

Both postponement and confluence are   commutation properties.
	Two relations $\redL{}$ and $\redM{}$  on $\AA$
\emph{commute} if 
\begin{equation}\tag{\textbf{Commutation}}
 \revredL{}^{*} \cdot \redM{}^{*} {\subseteq }
\redM{}^{*} \cdot\revredL{}^{*}.
\end{equation}


So, a relation $\red$ is confluent if and only if it commutes with itself. 
Postponement and commutation can be defined in terms of each other, simply taking  
$\ired$ for $ \revredL{} $ and $\ered $ for $\redM{}$  ($\ired$ postpones after $\ered$  if and only  if
$\xrevredx {\int}{}$ commutes with  $\ered$).   
As propounded  in \cite{vO20}, this fact  allows for proving postponement by means of  \emph{decreasing diagrams} \cite{Oostrom94,DD}. 
This is a powerful and general technique to prove commutation properties: it reduces the problem of showing commutation to a \emph{local} test; in exchange for localization, diagrams need to be decreasing with respect to a labelling.

\begin{defn}[Decreasing] Let $\redL{} \deff \bigcup_{k\in K} \redL{k}$ and $\redM{}  \deff \bigcup_{j\in J} \redM{j}$. 
	The  pair of relations $\redL{}, \redM{}$ is \emph{decreasing} if for some well-founded strict order $<$ on the set of labels $K\cup J$ the following holds: 
\begin{align*}
		\revredL{k}\cdot \redM{j} ~\subseteq  ~ (
	\redM{ \decr{k}}^* \cdot   \redM{j}^= \cdot  \redM{\decr  { k,j}}^* )
	~	\cdot ~ ( \revredL{\decr {j}}^* \cdot   \revredL{k}^= \cdot  \revredL{\decr  {k,j}}^*  ) 
	\qquad \text{for every } k\in K, j\in J 
\end{align*}
	where $\decr L =\{ i\in K\cup J\mid \exists l\in L.~l>i \}$ for any $L \subseteq K \cup J$, and $\decr{i_1, \dots, i_n} = \decr{\{i_1, \dots, i_n\}}$.
\end{defn}

\begin{thmm}[Decreasing diagram \cite{Oostrom94}]\label{thm:DD}  A pair of relations $\redL{},\redM{}$ commutes if it is decreasing.  
\end{thmm}

%
%

\paragraph{Modularizing Confluence.}
A classic tool to modularize a proof of confluence is Hindley--Rosen lemma: the union of confluent reductions is itself confluent if they all commute with each other.  
\begin{lem}[Hindley--Rosen]\label{lem:HR} Let  $\red_1$ and $\red_2$ be relations on a set $A$.  If $\red_1$ and 
	$\red_2$ are confluent and commute with each other, then $\red_1\cup \red_2$ is confluent.
\end{lem}
Like for postponement, strong commutation implies commutation.
\begin{lem}[Strong commutation \cite{HindleyPhD}]\label{lem:SC}
	Strong commutation 	($\leftarrow_1 \cdot \red_2 ~\subseteq ~   {\red_2}^* \cdot {\leftarrow_1}^=  $)
	implies commutation.
\end{lem}

\subsection{Basics on the $\lam$-Calculus }
\label{sec:CbV}

We recall the syntax  and  some relevant notions of the $\lam$-calculus, taking 
Plotkin's call-by-value (\emph{CbV}, for short) $\l$-calculus \cite{Plotkin'75}  as a concrete  example.

 Terms and values  are mutually generated by the grammars below.
\begin{align*}		
%
V ~ \Coloneqq  ~ x   \mid \lam x. T   		&\quad (\mathit{values}\textup{; set: } \ValTerm)&
		T, S, R \Coloneqq  V \mid TS   & \quad (\mathit{terms}\textup{; set: } \Lambda) 
\end{align*}
\noindent
where $x$ ranges over a countably infinite set $\Var$ of \emph{variables}. Terms of shape 
$TS$ and  $\lam x. T$ are called \emph{applications} and   \emph{abstractions}, respectively.
In $\lam x.T$, $\lam x$ binds the occurrences of $x$ in $T$.
The set of \emph{free} (\ie non-bound) variables of a term $T$ is denoted by $\fv{T}$.
Terms are identified up to (clash-avoiding) renaming of their bound variables ($\alpha$-\emph{congruence}).



\paragraph{Reduction.}
\begin{itemize}
	\item \textbf{Contexts} (with exactly one \emph{hole} $\hole{\ }$) are generated  by  the grammar
	\begin{equation*}	
			\cc  \Coloneqq  \hole{~}   \mid T\cc\mid \cc T \mid \lambda x.\cc 
			 \qquad (\textit{Contexts})		
	\end{equation*}
	$\cc\hole{T}$ 
	stands for the term obtained from $\cc$ by replacing the  hole with the term $T$ (possibly capturing some free variables of $T$).
	
	\item A \textbf{rule}  $ \Rule$ is  a binary relation on $\Lambda$, also noted   $\Root{\Rule}$, writing  $R \Root{\Rule} R'$;
	$R$ is called a $\Rule$-\emph{redex}.
	
	%
	\item A \textbf{reduction step}  $\red_{\Rule}$ is 
	the  closure of $\Rule$ under contexts $\cc$.
	Explicitly, 
	if $T, S \in \Lambda $ then $T \red_{\Rule} S$ if $T = \cc\hole R$ and 
	$S = \cc\hole{R'}$,  for some context $\cc$ and some $ R\Root{\Rule} R'$.
\end{itemize}

\paragraph{The Call-by-Value $\lam$-calculus.}
\newcommand{\lv}{\lam_v}

The CbV $\l$-calculus is  the rewrite system $(\Lambda, \redbv)$, 
the set of terms $\Lambda$ equipped with 
\emph{$\betav$-reduction} $\redbv$, that is, the 
contextual closure  of the  rule  $\mapsto_{\betav}$:  
\[(\lam x. {T})V \Root{\betav} T\Subst{V}{x}    \quad (V\in \ValTerm)\] 
\noindent where  $T\Subst{V}{x}$ is the term obtained from $T$ by capture-avoiding substitution of $V$ for the free occurrences of $x$ in $T$.
Notice that  here
$\beta$-redexes can be fired only when the argument is a \emph{value}.

%
%
%

\textbf{Weak evaluation} (which does not
reduce in the body of a function) evaluates closed terms to values. In the literature of CbV, there are three main weak
schemes: reducing from left to right, as defined by Plotkin \cite{Plotkin'75}, from right to left \cite{Leroy-ZINC}, or in an arbitrary order
\cite{LagoM08}.
\emph{Left} contexts  $\leftc$, \emph{right} contexts $\rightc$, and (arbitrary order)   \emph{weak} contexts $\ww$ 
are respectively  defined by
\begin{equation*}
	\leftc \Coloneqq \hole{~}  \mid  \leftc\, T \mid  V \leftc   \quad\quad
	\rightc \Coloneqq \hole{~}  \mid   T \rightc \mid   \rightc V  \quad
	\quad  \ww \Coloneqq \hole{~}  \mid  \ww T \mid   T \ww
\end{equation*}
Given a rule  $\Root{\Rule}$  on $\Lambda$,
\emph{weak  reduction} 
$\wredx{\Rule}$  is the closure of  $\Root{\Rule}$ under weak contexts $\ww$;
\emph{non-weak reduction} $\nwredx{\Rule}$ is the closure of $\Root{\Rule}$ under contexts $\cc$ that are not weak.
\emph{Left} and \emph{non-left} reductions ($\xredx{ \lsym}{\Rule}$  and $\xredx{\neg \lsym}{\Rule}$), \emph{right} and \emph{non-right} reductions ($\xredx{ \rsym}{\Rule}$  and $\xredx{\neg \rsym}{\Rule}$)
are defined analogously.

Note that $\lredbv$ and $\xredx{\rsym}{\betav}$ are deterministic, whereas $\wredbv$ is not.


\paragraph{CbV Weak Factorization.} Factorization of $\redbv$ allows for a  characterization of the terms which reduce to a value. 
Convergence below is a   remarkable consequence of factorization.
\begin{thmm}[Weak left  factorization \cite{Plotkin'75}] \label{thm:factorization_CbV}\hfill
	\begin{enumerate}
		\item\label{p:factorization_CbV-left} \emph{Left  Factorization of $\redbv$:}  ~	 $\redbv^* \ssubseteq  \lredbv^*  \cdot \nlredbv^*$.
		\item\label{p:factorization_CbV-converge} \emph{Value Convergence:}~   $T \redbv^* V$ for some value $V$ if and only if $T \lredbv^*V'$ for some value $V'$.
	\end{enumerate}
The same results hold for $\wredbv$ and $\xredx{\rsym}{\betav}$ in place of $\lredbv$.
\end{thmm}

Since the $\lredbv$-normal forms of closed terms are exactly closed values, \Cref{thm:factorization_CbV}.\ref{p:factorization_CbV-converge} means that every closed term $T$ $\betav$-reduces to a value if and only if $\lredbv$-reduction from $T$ terminates.
%
%
%

%% file: Syntax.tex

\section{The Computational Core $\TFlambdaComp$}
\label{sec:compcal}

We recall the syntax and the reduction  of the {\em computational core}, shortly $\TFlambdaComp$, introduced in \cite{deLiguoroTreglia20}.
 We use a notation slightly different from the one used in \cite{deLiguoroTreglia20} (and recalled in \Cref{sect:intro}).
Such a syntactical change is convenient both  to present the calculus in a more familiar fashion, and to establish useful connections between $\TFlambdaComp$ and two well known calculi, namely Simpson's calculus \cite{Simpson05} and Plotkin's call-by-value $\lambda$-calculus \cite{Plotkin'75}.
The equivalence between the current presentation of $\lc$ and \cite{deLiguoroTreglia20} is detailed in \Cref{app:crbiso}.

\begin{defn}[Terms of $\TFlambdaComp$]\label{def:compgram}	
	Terms of the computational core
	consist of two sorts of expressions:
	\begin{align*}
		\ValTerm: && V, W & \Coloneqq  x \mid \lambda x.M && \mbox{(values)} \\
		\ComTerm: && M,N & \Coloneqq  \oc V \mid VM && \mbox{(computations)}
	\end{align*}
	where $x$ ranges over a countably infinite set $\Var$ of variables. We set $\Term \defeq \ValTerm\, \cup \ComTerm$; 
	$\FV(V)$ and $\FV(M)$ are the sets of free variables occurring in $V$ and $M$ respectively, and are defined as usual.
	Terms are identified up to clash-avoiding renaming of bound variables (\emph{$\alpha$-congruence}).
\end{defn}

%


The unary operator $\oc$ is just another notation for $\Unit$ as presented in \Cref{sect:intro}: it coerces a value $V$ into a computation $\oc V$, sometimes called \emph{returned value}. 

\begin{rem}[Application]
	\label{rmk:application}
A computation $VM$ is a \emph{restricted} form of application, corresponding to the term $M \Bind V$ in \cite{Wadler95} (see \Cref{sect:intro}) where there is no functional application. 
The reason is that the bind $\Bind$ represents an effectful form of 
application, such that by redefining the unit and bind one obtains an actual evaluator for the desired computational effects \cite{Wadler95}.
This restriction may seem a strong limitation because we apparently cannot express iterated applications: $(VM)N$ is not well formed in $\lc$.
However, application among computations is definable in $\lc$: 
\[
MN \deff (\lambda z. zN) M \qquad\quad \mbox{where $z\not\in\FV(N)$.}
\]

\end{rem}

\SLV{}{\RED{\paragraph{\emph{To Apply or not to Apply.}}
	
	As remarked in  \cite{deLiguoroTreglia20}, there is no application in the syntax. 
	Indeed the only form of application which is admissible is $VW$, that is a computation with type $TD$. By adding these new terms to $\Comp$,
	we have to consider the\RED{ axiom of $\beta_v$-conversion}, namely $(\lambda x.M)W = \Subst{x}{W}{M}$. However, this introduces redundancies in the calculus; 
	let $VW$ be defined as abbreviation of $\Unit W \Bind V$, then by (\ref{eq:beta-c}) and taking $V = \lambda x.M$:
	\begin{equation}\label{eq:betav}
		(\lambda x.M)W \equiv \Unit W \Bind (\lambda x.M) \Red M\Subst{W}{x} 
	\end{equation}
	Further we observe that, even the ill-formed application $MN$ of a computation to a computation can be simulated in our syntax by taking:
	\begin{equation}\label{eq:monadicApp}
		MN \equiv M \Bind( \lambda z. N \Bind z) \qquad \mbox{for $\quad z \not\in \FV(N)$}
	\end{equation}
	Last but not least, there is a deeper reason for not having functional application as primitive, which is that the $\Bind$ operator is the application
	of the computational $\lambda$-calculus, as explained at the beginning of this Introduction. Indeed one should keep in mind that both
	$\Unit\!\!$ and $\Bind$ are abstractions of concrete realizations of such operators, depending on the effects, and hence on the monad we want to model.
	Adding application as primitive, and consequently the $\beta_v$ rule, extends the calculus by a kind of a pure application that is in general different than
	the impure one, that is effectful and represented by the bind: this is an advantage when designing a programming language, but not with a foundational calculus.
}}

\renewcommand{\BindRight}{\id}
\paragraph{Reduction.}
The operational semantics of $\lc$ puts 
together rules corresponding to the monad laws.

\begin{defn}[Reduction]\label{def:reduction}
	Relation $\Root{\lcc} \ = \ \Root{\Betac} \cup \Root{\BindRight} \cup \Root{\BindLeft}$ is the union of the following rules:
	\begin{align*}
		\Betac) && (\lambda x.M)(\oc V) & \mapsto_{\betac}  M\Subst{V}{x} \\ 
		\BindRight \,) && (\lambda x. \oc x) M & \mapsto_{\id}  M \\ 
		\BindLeft\,) &&  (\lambda y.N)((\lambda x.M)L)  & \mapsto_{\sigma} (\lambda x.(\lambda y.N)M)L  & \mbox{for $x\not\in \FV(N)$}
	\end{align*}

For every $\Rule\in\{\betac,\sigma,\id, \copyright\}$, \emph{reduction} $\redx{\Rule}$ is the  contextual closure of  $\Root{\Rule}$, where \emph{contexts} are defined as follows: 
\begin{align*}
	\cc & \Coloneqq \hole{\,} \mid   \oc (\lambda x.\cc)  \mid V\cc \mid  (\lambda x.\cc)M & \text{Contexts}
\end{align*}
\end{defn}

All reductions in \Cref{def:reduction} are binary relations on $\ComTerm$, thanks to the proposition below.

\begin{prop}\label{property:conservativity}\label{prop:closure_red}
	The set of computations $\ComTerm$ is closed under substitution and reduction:
	\begin{enumerate}
		\item\label{property:conservativity-1}  If $M\in \ComTerm$ and $V\in \ValTerm$, then $M\Subst{V}{x} \in \ComTerm$.
		\item\label{property:conservativity-2} For every $\Rule\in\{\betac,\sigma,\id, \copyright\}$, if $N \redx{\Rule} N'$, then: $N \in \ComTerm$ if and only if $N'\in \ComTerm$. 
	\end{enumerate}
	
\end{prop}

\begin{proof}
	\Cref{property:conservativity-1} (formally proved by induction on $M$) holds because $M\Subst{V\!}{x}$ just replaces a value, $x$, with another value, $V\!$.
	\Cref{property:conservativity-2} is proved by induction on the context for $N\redx{\Rule} N'\!$, using \cref{property:conservativity-1}.
\end{proof}
%
%


%

The \emph{computational core} $\lc$ is the rewriting system $(\ComTerm, \redlc)$.

\begin{prop}[Confluence, \cite{deLiguoroTreglia20}]
	\label{prop:confluence}
	Reduction $\redlc$ is confluent.
\end{prop}

\begin{rem}[$\betac$ and $\betav$]The  relation between $\redbc$ of the computational core and 
		$\redbv$ of Plotkin's CbV $\lam$-calculus is investigated  in \Cref{subsect:translations-comp-cbv}.
To give a taste of it, we show with an example how  $\beta_v$-reduction is simulated by $\Betac$-reduction, possibly with more steps.
Since $\redlc$ is a relation on $\ComTerm$ (\Cref{prop:closure_red}),
no computation $N$ will ever reduce to any value $V$; however, reduction to values is represented by a reduction $N \redlcStar \oc V$, where $ \oc V$ is the coercion of value $V$ into a computation. 
Let us assume that $M \redbc^* \oc \lambda x.M'$ and $N \redbc^* \oc V$.
We have:
\[\begin{array}{lcl@{\hspace{0.5 cm}}l}
	MN & = & (\lambda z.zN)M &\mbox{(by the encoding in \Cref{rmk:application})}\\
	& \redbc^* & (\lambda z.z \oc V) (\oc \lambda x.M') &\mbox{(where $z \not\in \FV(V)$ since $z \not\in \FV(N)$)} \\
	& \redbc & (\lambda x.M')  \oc V \\
	&  \redbc & M'\Subst{V}{x} 
\end{array}\]
Similarly, in Plotkin's CbV $\lam$-calculus, if $M \redbv^* \lambda x.M'$ and $N \redbv^* V$, then $MN \redbv^* M'\Subst{V}{x}$.

\end{rem}

\paragraph{Surface and Weak Reduction.}
As we shall see in the next sections, 
 there are two natural restrictions of $\redlc$:  \emph{weak reduction} $\wredx{\lcc}$
which does not fire in the scope of  $\lambda$, and  \emph{surface reduction} $\sredx{\lcc}$, which does not fire in the scope of  $\oc$. 
The former is the evaluation  usually studied in CbV $\lambda $-calculus (\Cref{thm:factorization_CbV}). The latter is the natural evaluation  in linear logic, and in  Simpson's calculus, whose relation with $\lc$ we discuss  in \Cref{sec:bang}.

\emph{Surface} and \emph{weak contexts} are respectively defined by the grammars
\begin{align*}
	\ss &\Coloneqq \hole{\,} \mid V\ss \mid (\lambda x.\ss)M &\qquad \text{Surface Contexts}
\\	
	\ww &\Coloneqq \hole{\,} \mid V\ww &\qquad \text{Weak Contexts}
\end{align*}
 For $\Rule\in \{\betac,\id,\sigma,\copyright\}$, 
\emph{weak} reduction  $\wredx{\Rule}$
 is the closure  of $\Rule$ under weak contexts $\ww$, \emph{surface} reduction $\sredx{\Rule}$  is its  closure  under surface contexts  $\ss$.
\emph{Non-surface} reduction $\nsredx{\Rule}$ is the closure of $\Rule$ under contexts $\cc$ that are not surface.
Similarly, \emph{non-weak} reduction $\nwredx{\Rule}$ is the closure of $\Rule$ under contexts $\cc$ that are not weak.

Clearly, $\wredx{\Rule} \subsetneq \sredx{\Rule} \subsetneq \, \redx{\Rule}$.
Note that $\wredbc$ is a \emph{deterministic} relation, while $\sredbc$~is~not.
\begin{exmpl} To clarify the difference between surface and weak, let us  consider the  term $(\lambda x. \bI\oc x) \oc \lambda y. \bI\oc y$, where  $\bI=\lambda z.\oc z$, and two different $\redlc$ steps from it.
We underline the fired redex.
	\begin{align*}
	 (\lambda x.\underline{ \bI\oc x}) \oc \lambda y. \bI\oc y \, &\sredx{\lcc}\,  (\lambda x. \oc x)\oc \lambda y. \bI\oc y
	 &
	 (\lambda x. \bI\oc x) \oc \lambda y.\underline{  \bI\oc y} \, &\nxredx{\surf}{\lcc} \,(\lambda x. \bI\oc x)\oc  \lambda y. \oc y
	\\
	(\lambda x. \underline{\bI\oc x}) \oc \lambda y. \bI\oc y \,  &\nxredx{\weak}{\lcc}\, (\lambda x. \oc x)\oc \lambda y. \bI\oc y
	&
	(\lambda x. \bI\oc x) \oc \lambda y.\underline{  \bI\oc y} \, &\nxredx{\weak}{\lcc} \,(\lambda x. \bI\oc x)\oc  \lambda y. \oc y	.
\end{align*}
\end{exmpl}

Surface reduction can be seen as the natural counterpart of weak reduction in calculi with \emph{let}-constructors or explicit substitutions, as we show in \Cref{subsect:computational-let}.

\begin{rem}[Weak contexts]
In the CbV $\lam$-calculus (see \Cref{sec:CbV},) weak contexts can be given in three forms, according to the order in which  redexes that are not in the scope of abstractions are fired: $\leftc,\rightc, \ww$. When the grammar of terms is restricted to computations, the three coincide. 
So, in $\lc$  there is \emph{only one} definition of  {weak context}, and weak, left and right reductions coincide. 
\end{rem}

In  \Cref{sec:betac} and \Cref{sec:operational}, we  analyze the \emph{properties} of weak and surface reduction.  We first  study $\redbc$,  and then  we move to the  whole  $\lc$, where $\sigma$ and $\id$ also come to play.

%
%
%
\paragraph*{Notation.}In the rest of the paper, we adopt the following  notation:  
\begin{center}
	$\II:=\lam x.!x$  \quad and \quad  $\bDelta := \lam x. x!x$.
\end{center}

\subsection{The Computational Core vs. Computational Calculi with $\tlet$-Notation}
\label{subsect:computational-let}

	It is natural to compare the computational core $\lc$  with  other untyped computational calculi, and wonder if the analysis of the rewriting theory of $\lc$ we present in this paper applies to them. 
	There is indeed a rich literature on computational calculi refining Moggi's $\lambda_c$ \cite{Moggi'88, Moggi'89,  Moggi'91},
	 most of them use the  {\em let}\,-constructor. 
	A standard reference is Sabry and Wadler's  $\lambda_{\mlsym}$ \cite[Sect. 5]{SabryWadler97}, which we display in \Cref{Figure-lm*}.

	 $\lambda_{\mlsym}$ has a two sorted syntax that separates \emph{values} (\ie variables and abstractions) and \emph{computations}. 
	 The latter are either $let$-expressions (aka explicit substitutions, capturing monadic binding), or applications (of values to values), or coercions $[V]$ of values $V$ into computations ($[V]$ is the notation  for $\Unit V$ in \cite{SabryWadler97}, so it corresponds to $\oc V$ in $\lc$).
\begin{itemize}
	\item 	The  \emph{reduction rules} in $\lambda_{\mlsym}$ are the usual $\beta$ and $\eta$  from Plotkin's \emph{call-by-value} $\lambda$-calculus \cite{Plotkin'75}, plus the oriented version of three \emph{monad laws}: $\LetSym.\beta$, $\LetSym.\eta$, $\LetSym.\textit{ass}$ (see \Cref{Figure-lm*}).
	\item 	\emph{Reduction} $\redx{\textit{ml}^*}$ is the contextual  closure of the union of these rules. 
\end{itemize}

\begin{figure}[!t]\label{fig:ml-syntax}
	{\small
\begin{align*}
	\textit{Values:} \quad V, W & \Coloneqq  x \mid \lambda x. M 
	& &&
	\textit{Computations:} \quad M, N & \Coloneqq  [V] \mid  \Let{x}{M}{N} \mid VW
\end{align*}
		
		\noindent
		\begin{align*}
			(c.\beta) && (\lambda x.M)V  & \Red  M\Subst{V}{x} \\
			(c.\eta)  && \lambda x. Vx & \Red  V \qquad \qquad x \not\in \fv{V} \\ 
			(c.\LetSym.\beta) && \Let{x}{[V]}{N} & \Red  {N}\Subst{V}{x} \\ 
			(c.\LetSym.\eta) && \Let{x}{M}{[x]} & \Red  M \qquad \qquad x \not\in \fv{M} \\ 
			(c.\LetSym.\textit{ass}) &&
			\Let{y}{(\Let{x}{L}{M})}{N} & \Red  \Let{x}{L}{(\Let{y}{M}{N})} \quad x \notin \fv{N} 
		\end{align*}
\\[-7pt]	
	}

	\caption{$\lambda_{ml^*}$: Syntax and Reduction}
	\label{Figure-lm*}
\end{figure}


\begin{figure}[!t]
			{\small 
			\[\qquad \trb{\cdot}: \lambda_{ml^*} \to \lc \]
			\\[-1em]
			\[\begin{array}{rcl}
				\trap{(x)} & \defeq & x \\ [1mm]
				\trap{(\lambda x.M)} & \defeq & \lambda x. \trb{M}\\ [3mm]
								\trb{[V]} & \defeq & \oc(\trap{V})\\ [1mm]
				\trb{VW} & \defeq & \trap{V} \, \oc (\trap{W}) \\ [1mm]
				\trb{\Let{x}{M}{N}} & \defeq &  (\lambda x. \trbp{N})\trbp{M} 
			\end{array}
			\]}	
		\hspace{8mm}%
		\newline
			{\small
			\[\trc{\cdot}:  \lc\to \lambda_{ml^*}\]
			\\[-1em]
			\[\begin{array}{rcl}

				\traa{x} & \defeq & x \\ [1mm]
				\traa{\lambda x.M} & \defeq & \lambda x. \trc{M} \\ [3mm]
				\trc{\oc V} & \defeq & [\traap{V}]\\ [1mm]
				\trc{V \,\oc W} & \defeq & \traap{V}\,\traap{W}\\ [1mm]
				\trc{xM} & \defeq &  \Let{y}{ \trcp{M} }{xy} \quad \text{if } y\not\in \fv{M}, \ M \neq \oc V \text{ for any value } V \\ [1mm]
				\trc{(\lambda x.N)M} & \defeq & \Let{x}{\trcp{M} }{\trcp{N} } \quad\! \text{if } M \neq \oc V \text{ for any value } V
			\end{array}
			\]}


	\caption{Translations between $\lambda_{ml^*}$ and $\lc$}
	\label{Figure-transl}
\end{figure}

To state a correspondence between $\lc$  and $\lambda_{ml^*}$, consider the translations in Figure \ref{Figure-transl}: translation $\trb{\cdot}$ from  $\ml$ to $\lc$ (resp. $\trc{\cdot}$ from  $\lc$  to $\ml$) is defined via the auxiliary encoding $\trap{(\cdot)}$ (resp. $\traa{\cdot}$) for values.
The translations induce an equational correspondence by adding $\eta$-equality to $\lc$.
More precisely, let $\redx{\eta}$ be the closure of the rule $\Root{\eta}$ (below left) under contexts $\gc$ (below right).
\[\lambda x. (V\oc x) \Root{\eta} V \qquad \qquad\qquad \vc \Coloneqq \hole{\,} \mid \lam x.\gc \qquad \gc \Coloneqq \oc \vc \mid \vc M \mid V \gc \]
Let $\eqlceta$ be the reflexive-transitive and symmetric closure of the reduction $\redlceta \eq \redlc \cup \Red_\eta$, and similarly for $\eqlm$ with respect to $\redx{ml^*}$.



\begin{restatable}{prop}{eqCorr}\label{prop:eq-correspondence}
%
The following hold:
\begin{enumerate}
\item \label{prop:eq-correspondence-1} $M \eqlceta\trb{\trcp{M}} $ for every computation $M$ in $\lc$;
\item \label{prop:eq-correspondence-2} $ \trc{\trbp{P}}\eqlm P$ for every computation $P$ in $\ml$;
\item \label{prop:eq-correspondence-3}$M \eqlceta N$ implies $\trcp{M}\eqlm \trcp{N}$, for every computations $M, N$ in $\lc$;
\item \label{prop:eq-correspondence-4} $P \eqlm Q$ implies $\trbp{P} \eqlceta \trbp{Q}$, for every computations $P, Q$ in $\ml$.

\end{enumerate}
\end{restatable}

\begin{proof}
	\begin{enumerate}
		\item By induction on $M$ in $\lc$.
		\item By induction on $P$ in  $\ml$.
		\item We prove that $M \redlceta N$ implies $\trcp{M}\eqlm \trcp{N}$, by induction on the definition of $M \redlceta N$.
		\item We prove that $P \redx{ml^*} Q$ implies $\trbp{P} \redlceta \trbp{Q}$, by induction on the definition of $P \redx{ml^*} Q$.
		\qedhere
	\end{enumerate}
\end{proof}
\Cref{prop:eq-correspondence} establishes a precise correspondence between the \textit{equational} theories of $\lc$ (including $\eta$-conversion) and $\lambda_{ml^*}$. 
We had to consider $\eqlceta$ since 
\begin{align}\label{eq:missing-eta}
	xM \not\eqlc\trb{\trc{xM}}\eq (\lambda y. x\oc y ) \trb{\trcp{M}} \qquad \text{when } M \neq \oc V \text{ for any value } V
\end{align}
(where $\eqlc$ is the  reflexive-transitive and symmetric closure of $\redlc$) and so condition (\ref{prop:eq-correspondence-1}) in \Cref{prop:eq-correspondence} would not hold if we replace $\eqlceta$ with $\eqlc$.

	
	\begin{rem}[Some intricacies]\label{rem:intricacies}
The correspondence between the \emph{reduction} theories of $\lc$ (possibly including $\eta$) and $\lambda_{ml^*}$ is not  immediate and demands further investigations since, according to the terminology in \cite{SabryWadler97}, there is no Galois connection:
\Cref{prop:eq-correspondence} where we replace $\eqlceta$ with $\redlcetaStar$, and $\eqlm$ with $\RedStar_{\textit{ml*}}$ does not hold. 
More precisely:
\begin{itemize}
	\item the condition corresponding to \Cref{prop:eq-correspondence-1}, namely $M \RedStar_{\lcc\eta}\trb{\trcp{M}} $, fails since 
	$xM\not\RedStar_{\lcc\eta} \trb{\trc{xM}}$ when $M \neq \oc V$ for any value $V$, see \eqref{eq:missing-eta} above;
	\item the condition corresponding to \Cref{prop:eq-correspondence-3}, namely $M \redlceta^* N$ implies $\trcp{M} \RedStar_{\textit{ml}^*} \trcp{N}$, fails because $(\lambda y.N)((\lambda x.\oc V) \oc W) \reds (\lambda x. (\lambda y.N) \oc V) \oc W$\\ but $\trc{(\lambda y.N)((\lambda x.\oc V) \oc W)} \not\RedStar_{\textit{ml}^*} \trc{(\lambda x. (\lambda y.N) \oc V) \oc W}$.
\end{itemize}
	\end{rem}

\paragraph{Surface vs.~Weak Reduction.}
%
In this paper we study not only weak but also  surface reduction, as the latter has better rewriting properties than the former. 
	\emph{Surface reduction in $ \lc $} can be seen as a natural  counterpart to Plotkin's  \emph{weak reduction in calculi with the \emph{let}-constructor}, such as $\lambda_{ml^*}$.  Intuitively, a term of the form $\Let{x}{M}{N}$ can be interpreted as syntactic sugar for $(\lambda x.N)M$, however, in  the expression  $\Let{x}{M}{N}$, it is not  obvious that weak  reduction should avoid firing redexes in  $N$.  The distinction between $\lambda$ and \emph{let} allows for a clean interpretation of weak reduction: it forbids reduction under $\lambda$, but not under \emph{let}, which  is compatible with  surface reduction in $\lc$.

Technically, when we embed $\lambda_{ml^*}$ in the computational core  via the translation $\trb{\cdot}$ in \Cref{Figure-transl}, weak reduction $\wredx{\textit{ml}^*}$ in $\ml$ (defined as the restriction of $\redx{\textit{ml}^*}$ that does not fire under $\lambda$) corresponds to surface reduction $\sredx{\lcc}$ in $\lc$: if $P \,\wredx{\textit{ml}^*}\, P'$ then $\trbp{P} \,\sredx{\lcc}\, \trb{P'}$ but not necessarily $\trbp{P} \allowbreak\wredx{\lcc}\, \trb{P'}$. 
Indeed, consider ${P} = \Let{x}{R}{Q} \,\wredx{\textit{ml}^*}\, \Let{x}{R}{Q'} = {P'}$ in $\ml$, with $Q \,\wredx{\textit{ml}^*}\, Q'$;
then $\trbp{P} = (\lambda x.Q)R \,\sredx{\lcc}\, (\lambda x.Q')R = \trb{P'}$, which is not a weak step, in $\lc$.

%% file: Translation.tex
\section{The Operational Properties of $\betac$}\label{sec:betac}\label{sec:bang} \label{subsect:translations-comp-bang}
\label{sect:translations}
%
%
 
 Since  $\beta$-reduction is the engine of any $\lambda$-calculus, we start our analysis of the rewriting theory of $\lc$ by studying  the properties of $\redbc$ and its surface restriction  $\sredx{\betac}$. 
As we show in this section, $\redbc$ and  $\sredx{\betac}$ have already been studied in the literature: there is an exact  correspondence with  Simpson's calculus \cite{Simpson05}, which stems from Girard's linear logic \cite{LL}. 
 Indeed, the operator $\oc$ in \cite{Simpson05} (modelling the \emph{bang} operator  from linear logic) behaves exactly as the the operator $\oc$ in $\lc$ (modelling 
 $\Unit$ in  computational calculi).
It is easily seen that  $(\ComTerm,  \redbc)$, that is, $\lc$ when considering only   $\Betac$ reduction,  is nothing but the restriction of the bang calculus to 
 computations.  Thus, $\redbc$  has the same operational properties as $\redbb$.
 In particular,  surface factorization and confluence for $\redbc$ are inherited from the corresponding  properties of $\redbb$ in Simpson's calculus.




\paragraph{The Bang Calculus.}
We call  \emph{bang calculus}  the fragment of Simpson's  linear $\lambda$-calculus \cite{Simpson05} without linear abstraction. 
It has also been studied in \cite{EhrhardGuerrieri16,GuerrieriManzonetto18,FaggianGuerrieri21,GuerrieriO21} (with the name bang calculus, which we adopt),
and it is closely related to Levy's Call-by-Push-Value \cite{Levy99}. 


We briefly recall the   bang calculus $ (\Bang,\redbb) $.
\emph{Terms}  $\Lambda^!$ are defined by
\begin{align*}
	T,S,R & \Coloneqq x  \mid 	TS  \mid 	\lambda x.T    \mid \oc T & (\textbf{terms}\textup{, set: } \Lambda^!)
\end{align*}

\emph{Contexts} ($\cc$) and \emph{surface contexts} ($\ss$) are generated  by the grammars:
\begin{align*}
	\cc & \Coloneqq \hole{\,}  \mid T\cc \mid \cc T   \mid \lambda x.\cc    \mid \oc\cc  & \qquad(\textbf{contexts})\\
	\ss & \Coloneqq \hole{\,} \mid T\ss \mid \ss T   \mid \lambda x.\ss &\qquad (\textbf{surface contexts})
\end{align*}

The reduction $\redbb$   is the closure under context $\cc$ of the rule
\[	(\lambda x.R)\oc T \mapsto_{\betab} R \Subst{T}{x}\]

\emph{Surface} reduction 
$\sredbb$   is the closure of the rule $\mapsto_{\betab}$ under surface contexts $\ss$.
\emph{Non-surface} reduction $\sredbb$   is the closure of the rule $\mapsto_{\betab}$ under contexts $\cc$ that are not surface.
Surface reduction \emph{factorizes} $\redbb$.

\begin{thmm}[Surface factorization \cite{Simpson05}] \label{thm:factorization_simpson}
	\label{prop:simpson-factorization}
	In $\Lambda^\oc$:
	\begin{enumerate}
		\item \emph{Surface factorization of $\redbb$:}	~ $\redbb^* \ssubseteq  \sredxStar{\betab}  \cdot \nsredxStar{\betab}$.
		\item \emph{Bang convergence:}  ~ $T\redbb^* \oc R$ for some term $R$ if and only if $T\sredxStar{\betab} \oc S$ for some term $S$.
	\end{enumerate}
\end{thmm}	

Surface reduction is non-deterministic, but satisfies the diamond property  of \Cref{fact:diamond}.	
	\begin{thmm}[Confluence and diamond {\cite{Simpson05}}]
		\label{thm:confluence_bang} 
		In $\Lambda^!$:
		\begin{itemize}
			\item reduction $\redbb$ is confluent;
			\item reduction $\sredbb$ is quasi-diamond (and hence confluent).
		\end{itemize}
	\end{thmm}

\paragraph{Restriction to Computations.} The restriction of the bang calculus  to computations, \ie $(\ComTerm, \redbb)$ is \emph{exactly  the same} as the fragment of $\lc$ with  $\betac$-rule as unique reduction rule, \ie  $(\ComTerm,\redbc)$.

First,  observe that the set of computations $\ComTerm$ defined in \Cref{sec:compcal} is a subset of the terms $\Lambda^!$, and moreover it  is closed under $\redbc$ reduction (exactly as in \Cref{prop:closure_red}).
Second, observe  that the restriction of contexts and surface contexts to computations, gives exactly the grammar defined in \Cref{sec:compcal}.
Then  
$(\ComTerm, \redbb)$ and $(\ComTerm,\redbc)$ are in fact \emph{the same}, and for every $M, N \in \ComTerm$:
	\begin{itemize}
		\item $M\redbb N$ if and only if   $M\redbc N$.
		\item $M\sredbb N$ if and only if   $M\sredbc N$.
	\end{itemize}


Hence, $\redbc$ in $\lc$ \emph{inherits the operational properties of $\redbb$}, in  particular
 \emph{surface factorization} and the \emph{quasi-diamond} property of $\sredbb$ (\Cref{prop:simpson-factorization,thm:confluence_bang}). 
 We will use both extensively. 
 
\begin{fact}[Properties of $\betac$ and of its \emph{surface} restriction]
	\label{fact:betac-surface-factorization}
	In $\lc$: 
	\begin{itemize}
		\item reduction $\redbc$ is non-deterministic and  confluent;
		
		\item reduction $\sredbc$ is quasi-diamond (and hence confluent);
		
		\item reduction $\redbc$ satisfies surface factorization: $\redbc^* \ssubseteq  \sredbc^*  \cdot \nsredbc^*$.
		
		\item $M\redbc^* \oc V$ for some value $V$ if and only if $T\sredxStar{\betac} \oc W$ for some value $W$.
	\end{itemize}
		 
\end{fact}

In \Cref{sec:factorizations,sec:values} we shall generalize and refine the last two points, respectively, to reduction $\redlc$ instead of $\redbc$.

%% file: Operational.tex
\section{Operational Properties of $\lc$, Weak and Surface Reduction}\label{sec:operational}

%
%
%

We study evaluation and normalization in $\lc$ via    \emph{factorization}  theorems (\Cref{sec:factorizations}), which are based on both  weak and surface reductions.
%
The construction we develop  in the next sections demands more work  than one may expect. This  is due to the fact that the  
rules induced by the monadic laws of  associativity and identity make the analysis of the reduction properties non-trivial.
In particular---as anticipated in the introduction---\emph{weak reduction} does not factorize $\redlc$, and has severe drawbacks, which we explain  next. \emph{Surface reduction} behaves better, but still present difficulties. 
In the rest of this section, we examine their respective  properties.

\subsection{Weak Reduction: the Impact of Associativity and Identity}
  
\emph{Weak (left)} reduction   (\Cref{sec:CbV}) is one of the most common and studied way to implement evaluation in CbV, and more generally in 
calculi with effects.

Weak $\betac$ reduction $\wredbc$, that is, the closure of $\betac$ under weak contexts is a \emph{deterministic} relation.
However, when including the rules induced by the monadic equation of  associativity and identity,   the   reduction is  \emph{non-deterministic},  
 \emph{non-confluent}, and  normal forms are  \emph{not unique}. 
 
This is somehow surprising, given the prominent role of such a reduction in the literature  of  calculi with effects. Notice  that the issues only come from $\sigma$ and $\id$, not from $\betac$.
To resume:
\begin{enumerate}
	\item Reductions $\wredx{\id}$ and  $\wredx{\betac\id}$ are non-deterministic, but are both confluent.
	\item\label{p:non-confluence} $\wreds, \ \wredx{\bs}, \ \wredx{\sigma\id} \text{, and } \wredx{\lcc}$  are non-deterministic, non-confluent and their normal forms are not unique,
	\ie, adding $\sigma$, weak reductions lose confluence and uniqueness of normal forms. 
\end{enumerate}


\begin{exmpl}[Non-confluence]\label{ex:weak}
		An example of the non-determinism of $\wredx{\id}$ is the following:
		$$V((\lambda y. \oc y)N) \,\xrevredx\, {\weak}{\id} (\lambda x. \oc x)(V((\lambda y. \oc y)N)) \,\wredx{\id}\, (\lambda x. \oc x)(VN).$$
			
		Because of the $\sigma$ rule, 
		weak reductions $\wreds, \ \wredx{\bs}, \ \wredx{\sigma\id}$ and $\wred{\lcc}$ are not confluent and their normal forms are not unique (\Cref{p:non-confluence} above). 
		Indeed,	consider 
		$T \eq V((\lam x.P)   ((\lam y.Q)L))$ where $V = \lambda z. z \oc z$ and $P = Q = L = z \oc z$.
		Then,
		\begin{align*}
			M_1= ( \lam x.VP) ((\lam y.Q)L)  \, \xrevredx{\weak }{\sigma}   \  T \ \wreds \,  V  ((\lam y. (\lam x.P) Q)L )= N_1
			\\
			M_1 \, \wreds \,  M_2 \defeq (\lam y.( \lam x. VP) Q)L  \neq  (\lam y. V((\lam x.P) Q)) L \,\eqdef\, N_2 \,\xrevredx{\weak}{\sigma} \, N_1
		\end{align*}
		where $M_2$ and $N_2$ are different $\wredx{\lcc}$-normal forms (clearly, $N_2 \, \nwredx{\sigma} \, M_2$).
%
\end{exmpl}
Reduction $\redbc$  (like $\redbv$, \Cref{thm:factorization_CbV}.\ref{p:factorization_CbV-left}) admits weak factorization (\Cref{fact:factorization-betac} in \Cref{subsect:translations-comp-cbv})
\[\redbc^* \ \subseteq \  \wredxStar{\betac} \cdot \nwredxStar{\betac} \]
This is not the case for $\redlc$.
The following   counterexample is due to van  Oostrom \cite{ex_vO}.
\begin{exmpl}[Non-factorization \cite{ex_vO}]\label{ex:vO} Reduction $\redlc$ does not admit weak factorization.
	\label{ex:weak-no-factorize}
	Consider the reduction sequence
	\[
	M \defeq (\lam y. \II \oc y)(z\oc z) \,\nwredx{\betac}\, (\lam y. \oc y) (z\oc z )\, \wredx{\id}\, z\oc z
	\]
	$M$ is $\wredx{\lcc}$-normal and cannot reduce to $z \oc z$ by only performing $\nwredx{\lcc}$ steps (note that $(\lam y. \oc y) (z\oc z )$ is $\nwredx{\lcc}$-normal), hence it is \emph{impossible}  to factorize the sequence from $M$ to $ z\oc z$ as $M\wredxStar{\lcc}\cdot  \nwredxStar{\lcc}  z\oc z$.
\end{exmpl}

\paragraph{Let: Different Notation, Same Issues.}
We stress that the 
  issues are inherent to the associativity and identity rules, not to the specific syntax of $\lc$.
Exactly the same issues appear  in   Sabry and Wadler's $\ml$ \cite{SabryWadler97} (see our \Cref{fig:ml-syntax}), as we show in the example below.

 \begin{exmpl}[Evaluation  context in \tlet-notation]\label{ex:sequencing} 
 	In \tlet-notation,  the standard evaluation  is  \emph{sequencing} \cite{Filinski:phd1996,JonesSLT98,LevyPT03}, which exactly corresponds to weak reduction in $\lc$.  The  evaluation   context for sequencing is  
 \[\ee_\tlet  \, \Coloneqq \, \hole{\,} ~ \mid~ \Let x {\ee_\tlet}  M.\]	
 We write $\eredx{\mlsym}$ for the closure of the $\ml$ rules  (in \Cref{fig:ml-syntax}) under contexts $\ee_\tlet$. 
 We observe two problems, the first one due to the rule $c.\LetSym.\textit{ass} $, the second one to the rule $c.\LetSym.\eta$.
 
 \begin{enumerate}
 	\item \textbf{Non-confluence.} Because of the associative rule $c.\LetSym.\textit{ass}$,  reduction $\eredx{\mlsym}$ is \emph{non-deter\-ministic}, \emph{non-confluent}, and \emph{normal forms are not unique}.
 	Consider the following term
 	\[T \deff   \overline{\Let z {\, \underline{(\Let   x {\,( \Let y L Q)} P )}} R}
 	\qquad\text{with } R= P = Q = L = z  z.
 	\]
 	There are two weak redexes in $T$, the overlined and the underlined one. Therefore,
 	\begin{align*}
 		T ~& \eredx{\mlsym} ~ \Let x {\,(\Let y L Q)}{(\Let z P R)}\\
 		& \eredx{\mlsym} ~ \Let y L {(\Let x Q {(\Let z P R)})} \eqdef T'\\[3pt]
 		T ~  & \eredx{\mlsym}~ \Let z {\,(\Let y L {(\Let x Q P)})} R \\
 		& \eredx{\mlsym} ~	\Let y L {(\Let z {(\Let x Q P)} R)} \eqdef T''
 	\end{align*} 	
 	where $T'$ are  $T''$ different $\eredx{\mlsym}$-normal forms.
 	
 	\item \textbf{Non-factorization.} Because of the $c.\LetSym.\eta$-rule,
 	\emph{factorization w.r.t. sequencing} does not hold. That is, a reduction sequence  $M \RedStar_{\mlsym}N$  cannot be reorganized as  weak steps followed by non-weak steps. 
 	Consider the following variation on van Oostrom's \Cref{ex:vO}:
 	\begin{align*}
 		M ~&\defeq~ \Let{y}{zz}{(\Let{x}{[y]}{[x]})} 
 		\ \nwredx{c.\LetSym.\eta}~ \Let{y}{ z z}{[y]} ~~\wredx{\,c.\LetSym.\eta}~~ zz 
 	\end{align*}
 	$M$ is $\wredx{\mlsym}$-normal and cannot reduce to $z z$ by only performing $\nwredx{\mlsym}$ steps (note that $\Let{y}{z z}{[y]}$ is $\nwredx{\mlsym}$-normal), so it is \emph{impossible}  to factorize the sequence form $M$ to $zz$ as $M\wredxStar{\mlsym} \!\cdot  \nwredxStar{\mlsym} zz$.
 \end{enumerate}

 	 \end{exmpl}	
  

\subsection{Surface Reduction}
In $\lc$, surface reduction  is non-deterministic, but  confluent, and well-behaving.  
\begin{fact}[Non-determinism]
	For $\Rule \in \{\copyright, \betac, \sigma,\id, \bs, \beta\id, \sigma\id\}$,   $\sredx{\Rule}$ is non-deterministic (because in general more than one surface redex can be fired).

\end{fact}

We now analyze confluence of surface reduction. 
We will use  confluence of $\sredx{\bs}$  (\Cref{p:surface-properties-betasigma} below) in \Cref{sec:normalization}
(\Cref{prop:uniform}).
\begin{prop}[Confluence of surface reductions]
	\label{prop:surface-properties}
\hfill
	\begin{enumerate}
	\item\label{p:surface-properties-separate} Each of the reductions  $\sredbc$, $\sredid$,  $\sreds$ is  confluent.
	\item\label{p:surface-properties-betasigma} Reductions $\sredx{\betac\id} =(\sredbc \cup \sredid)$  and 
	$\sredx{\bs}  = (\sredbc \cup \sreds)$ are confluent.
	\item\label{p:surface-properties-betasigmaid} Reduction $\sredx{\lcc} = (\sredbc \cup \   \sreds \cup \sredid)$  is confluent.
	\item\label{p:surface-properties-sigmaid} Reduction $\sredx{\sigma\id} = (\sreds \cup \sredid)$ is \emph{not} confluent.
\end{enumerate}

\end{prop}

\begin{proof} 
	We rely on confluence of $\sredbc$ (by \Cref{thm:confluence_bang}), and on   Hindley--Rosen Lemma (\Cref{lem:HR}). We prove commutation via strong commutation (\Cref{lem:SC}).
	The only delicate point  is the  commutation of $\sreds$ with $\sredid$ (\Cref{p:surface-properties-betasigmaid,p:surface-properties-sigmaid}).
	
	\begin{enumerate}
		\item   $\sreds$ is locally confluent and terminating, and so confluent; 
		$\sredid$ is quasi-diamond (in the sense of \Cref{fact:diamond}), and hence also confluent.
		\item It is easily verified that  $\sredbc $ and $  \sredid$ strongly commute, and similarly  for  
		$\sredbc $ and $  \sreds$. The claim then  follows by Hindley--Rosen Lemma.
		\item $\sredbc \cup \sredid$ strongly commutes with $\sreds$. This point is delicate because to close a  diagram  of the shape $\xrevredx{\surf}{\sigma} \cdot \sredid$ may require a $\sredbc$ step, see \Cref{ex:confluence} below.
		 The claim then follows by Hindley--Rosen Lemma.
		\item  A counterexample is provided by the same diagram  mentioned in the previous point  (\Cref{ex:confluence}), requiring a $\sredbc$ step to close.
		\qedhere
	\end{enumerate}	
\end{proof}

%
%

%
%

\begin{exmpl} 
	\label{ex:surface}
	Let us  give an example for  non-determinism and confluence of surface reduction.
	\begin{enumerate}
		\item \emph{Non-determinism:}	Consider the term $(\lambda x. R)((\lambda y. R')N)$ where $R$ and $R'$ are any redexes.
		\item \emph{Confluence:} 	Consider the same  term  as in \Cref{ex:weak}:
		$T=V((\lam x.P)   ((\lam y.Q)L ))$.
		Then,
		\[ M_2	\xrevredx{\surf }{\sigma} M_1  \xrevredx{\surf }{\sigma}     T\, \sreds  \,  N_1 \sreds \,   N_2 \]
		Now we can close the diagram:
		\[M_2 \eq  (\lam y.( \lam x. VP) Q)L \xrevredx{\surf}{\sigma} (\lam y. V((\lam x.P) Q)) L =  N_2.\]
	\end{enumerate}
\end{exmpl}

{\begin{exmpl}\label{ex:confluence}
		In the  following counterexample to confluence of $\sreds \cup \sredid$, where $M= N=z!z$, the  $\BindLeft$-redex overlaps with the $\BindRight$-redex.  
		The corresponding  steps are surface (and even weak) and the only way to  close the diagram is  by means of a $\Betac$ step, which is also surface (but not weak).
		\begin{diagram}[width=4em,height=2.5em]
			(\lambda y.\, N)((\lambda x.\, \oc x)M) & &	 \rTo^{\BindLeft}	& & 	(\lambda x.(\lambda y.\, N)\, \oc x)M  \\
			\dTo^{\BindRight}		&   	&							&	&     \dDashto_{\Betac} 		\\
			(\lambda y.\, N)M	&	&        \rDashes^=		&	& (\lambda x.\, (N\Subst{x}{y}))M	  \\
		\end{diagram}
		Note that  $x\not \in \FV(N)$, and so $ \lambda x. (N\Subst{x}{y}) = \lam y. N$, since $ \lambda x. (N\Subst{x}{y})$ is the term obtained from $\lam y.N$ by renaming its bound variable $y$ to $x$.
		
		This  is also a counterexample to confluence of $(\reds \cup \redid)$  and of  $(\wreds \cup \wredx{\id})$.  
\end{exmpl}}

In \Cref{sec:factorizations}, we 
 prove that {surface reduction} \emph{does  factorize} $\redlc$, similarly to what happens for Simpson's calculus (\Cref{prop:simpson-factorization}).
Surface reduction  also has a drawback:  it 
   does not allows us to separate $\redbc$ and $\reds$ steps. This fact  makes it  difficult to reason about  returning a value.
\begin{exmpl}[An issue with surface reduction] 
	Consider the term $\Delta ((\lambda  x. \oc\Delta) (xx))$, which  is normal for $\redx{\betac}$ and in particular for $\sredx{\betac}$, but	
\[	\Delta ((\lambda  x. \oc \Delta) (x\oc x)) \,\sreds\, (\lambda  x. \Delta \oc\Delta) (x\oc x) \,\sredx{\betac}\, (\lambda  x. \Delta \oc\Delta) (x\oc x)\]
Here it is not possible to postpone a step $\sreds$  after a step $\sredx{\betac}$.
\end{exmpl}

\subsection{Confluence Properties of $\betac$, $\sigma$ and $\id$}
Finally, we briefly revisit the  confluence of $\lc$, already established in \cite{deLiguoroTreglia20}, in order to analyze the confluence properties of the different subsystems too. This completes  the analysis given in \Cref{prop:surface-properties}.
In \Cref{sec:normalization} we will use  confluence of $\redx{\bs}$ (\Cref{thm:normalizing_strategy}).
\begin{prop}[Confluence of  $\betac$,  $\sigma$ and $\id$]
	\label{prop:conf}
	\hfill
	\begin{enumerate}
		\item Each of the reductions  $\redbc$, $\redid$,  $\reds$ is  confluent.
		\item\label{p:conf-betasigma} Reductions $\redx{\betac\id} \, = (\redbc \cup \redid)$  and 
		$\redx{\bs} \, = (\redbc \cup \reds)$ are confluent.
		\item\label{p:conf-betasigmaid} Reduction $\redlc \, = (\redbc \cup \reds \cup \redid)$  is confluent.
		\item\label{p:conf-sigmaid} Reduction $\redx{\sigma\id} \, = (\reds \cup \redid)$ is \emph{not} confluent.
	\end{enumerate}
\end{prop}
\begin{proof} We rely on confluence of $\redbc$ (by \Cref{thm:confluence_bang}), and on   Hindley--Rosen Lemma (\Cref{lem:HR}). We prove commutation via strong commutation (\Cref{lem:SC}).
	The only delicate point  is again  the  commutation of $\reds$ with $\redid$ (\Cref{p:conf-betasigmaid,p:conf-sigmaid}).
	
	\begin{enumerate}
		\item   $\reds$ is locally confluent and terminating, and so confluent.  $\redid$ is quasi-diamond in the sense of \Cref{fact:diamond}, and hence confluent.
		\item It is easily verified that  $\redbc $ and $  \redid$ strongly commute, and  
		$\redbc $ and $  \reds$ do as well. 
		The claim then  follows by Hindley--Rosen Lemma.
		\item $\redbc \cup \redid$ strongly commutes with $\reds$. 
		This point is delicate because to close a  diagram  of the shape $\xrevredx{}{\sigma} \cdot \redid$ may require a $\redbc$ step (see \Cref{ex:confluence}).
		The claim then  follows by Hindley--Rosen Lemma.
		\item  A counterexample is provided by the  same diagram mentioned in the previous point (\Cref{ex:confluence}), requiring a $\redbc$ step to close.
		\qedhere
	\end{enumerate}	
\end{proof}

\SLV{}{
	\subsection{Roadmap: from factorization to evaluation and normalization}
	Because of their respective properties, we use both surface and weak reduction.
	
	In \Cref{sec:factorizations} we study  \textbf{factorization} results. We
	first prove surface factorization (\Cref{thm:sfactorization}), which is the cornerstone of our construction.   Then, we refine this  result first by postponing $\id$ steps and then by means  of  weak factorization (on the surface steps). 
	This further  phase allows us: 
	\textbf{(1.)}  to separate $\redbc$ and $\reds$ steps,  by postponing  $\sigma$  steps after  the $\wredx{\betac}$ steps 
	and 
	\textbf{(2.)} to perform 
	a fine analysis of quantitative properties, namely the number of beta steps. We need (1.) to define the evaluation relation, and (2.) to define the normalizing strategies. 
	

By  relying  on the factorization properties from \Cref{sec:factorizations}, we then obtain  two major results.
\begin{enumerate}
	\item \textbf{Evaluation} (\Cref{sec:values}).  If $M$ returns a value,   only $\betac$ steps are necessary to compute a value:
	\[ M \redlc^* \oc V  \mbox{ iff }    M \wredbc^* \oc W  \]
	We  study evaluation  (see \Cref{thm:values}) in  \Cref{sec:values}, where we also analyze some relevant consequences of this result. 
	
	\item \textbf{Normalization} (\Cref{sec:normalization}).
	If $M$ has a normal form, only $\betac$ and $\sigma$ steps are necessary to normalization. Indeed,  $\id$ steps that are not $\betac$ do not play any role (see \Cref{thm:normals}). 
	From here, and again building on the factorization results in \Cref{sec:factorizations}, we  define two families of normalizing strategies,  by iterating weak reduction and surface reduction, respectively. 
	
\end{enumerate}

A point which is worth stressing is that 
the proofs related to  normalization and normalizing strategies rely on a fine analysis of the \emph{number of $\betac$-steps,} which we carry-on while studying factorization in  \Cref{sec:factorizations}. 

}

%% file: Factorization.tex

\section{Surface and Weak  Factorization }
\label{sec:factorizations}

	In this section, we prove several  factorization results for $\lc$. 
	Surface factorization is  the cornerstone for the subsequent development.
	It is proved in \Cref{subsect:surface-factorization}.
	
	\begin{restatable}[Surface factorization in $\lam_{\lcc}$]{thmm}{SurfaceFactorization}\label{thm:sfactorization}
		Reduction $\redlc$ admits surface factorization:
		\begin{align*}
			M\RedStar_{\lcc}N \text{ implies }  M \sredxStar{\lcc} \cdot  \nsredxStar{\lcc} N.
		\end{align*}
	\end{restatable}

 We   then refine this  result first by \emph{postponing} $\id$ steps  which are not also $\betac$ steps, and then by means  of  \emph{weak factorization} (on surface steps). 
	This further  phases serve two purposes: 
	\begin{enumerate}
		\item \label{p:weak} to postpone non-weak $\betac\sigma$ steps after weak $\betac\sigma$ steps, and
		\item\label{p:separate} to separate weak $\betac$ and $\sigma$ steps,  by postponing  $\wreds$  steps after $\wredx{\betac}$ steps, and
		\item\label{p:quantitative} to perform 
		a fine analysis of quantitative properties, namely the number of $\betac$ steps.
	\end{enumerate}
	We will  need \Cref{p:weak,p:separate} to define \emph{evaluation} relations (\Cref{sec:values}), and \Cref{p:quantitative} to define  \emph{normalizing strategies} (\Cref{sec:normalization}). 
	

	

\paragraph{Technical Lemmas.}
We shall often  exploit  some basic properties of  contextual closure, which we collect here.
 In any variant of the $\lambda$-calculus, if a step $ T \redx{\Rule} T'$ is obtained by the closure of a rule $\Root{\Rule}$ under a \emph{non-empty context} (\ie, a context other than the hole), then  $T$ and $T'$ have the \emph{same shape}, that is, they are both
	applications or both abstractions or both variables or both !-terms.

\begin{restatable}[Shape preservation]{fact}{shape}
	\label{fact:shape} Let $\Root{\Rule}$ be a rule and $\redx{\Rule}$ be its contextual closure.
	Assume $T=\cc\hole{R}\redx{\Rule} \cc\hole {R'}=T'$ where $R \Root{\Rule} R'$ and    $\cc \neq \hole{\,}$. 
	Then $T$ and $T'$ have the same shape.
	\end{restatable}

An easy-to-verify consequence of \Cref{fact:shape} in $\lc$ is the following.

\begin{restatable}[Redexes preservation]{lem}{lemredexbsi}
	\label{lem:redex_bsi}
	Let  $M \,\nsredx{\lcc}\, N$ and $\gamma \in \{\betac,\sigma,\id\}$:  
	$M$ is a $\gamma$-redex if and only if $N$ is a $\gamma$-redex.
\end{restatable}
\begin{proof}
	See the Appendix, namely \Cref{cor:redex} for $\gamma \in \{\sigma,\betac\}$, 
	and \Cref{lemma:id-redex} for $\gamma = \id $.
\end{proof}

\Cref{lem:redex_bsi} is false if we replace the hypothesis $M\nsredx{\lcc} \, N$ with $M\nwredx{\lcc} \, N$.
Indeed, consider $M = (\lam x. (\lam y. \oc y)\oc x)L \nwredx{\lcc} (\lam x. \oc x) L = N$: $N$ is a $\id$-redex but $M$ is not.

Notice the  following inclusions, which we will use freely.
\begin{fact}
	$ \wredx{\lcc} \subsetneq \sredx{\lcc}$ and  $\nsredx{\lcc}\subsetneq \nwredx{\lcc}$, because a  weak context is necessarily a surface context (but a surface context need not be a weak context, \eg the surface context $\ss = (\lam x. \hole{\,})M$ is not weak).
\end{fact}

\subsection{Surface Factorizations in $\lam_{\lcc}$, Modularly}
\label{subsect:surface-factorization}

We prove surface factorization in $\lc$.  
We already know that surface factorization holds for $\redbc$ (\Cref{fact:betac-surface-factorization}), so we can rely on it, and 
work modularly, following the approach proposed in \cite{AccattoliFaggianGuerrieri21}.   The tests for call-by-name head factorization  and call-by-value weak factorization in  \cite{AccattoliFaggianGuerrieri21}  easily adapt to surface factorization in $\lc$, yielding    the following convenient test.  It modularly establishes surface factorization of a reduction $\redbc \cup \redc $, where $\redc$ is a new reduction added to $\redbc$. 
Details of the proof are in \Cref{app:factorization,subsect:restriction}.


\begin{restatable}[A modular test for surface factorization with $\betac$]{prop}{testss}\label{prop:test} 	
	Let  $\redc$ be the contextual closure of a rule $\rredc $.  	
	Reduction	$\redbc \cup \redc$  satisfies surface factorization (that is, $(\redbc \cup \redc)^* \subseteq (\sredbc \!\cup\, \sredc)^* \cdot (\nsredbc \!\cup\, \nsredc)^*$) if:
	\begin{enumerate}
		\item \emph{Surface factorization of $\redc$}: ~~
		$\redc^* \, \, \subseteq ~\sredxStar{\gamma} \cdot \nsredxStar{\gamma}$\,.
		\item  $\rredc$ is  \emph{substitutive}:  ~~ 	$R \rredc R' 
		\text{ implies } R \subs x M \rredc R'\subs x M.$
		\item \emph{Root linear swap}: ~~ $\nsredbc \cdot \rredc \, \subseteq \ \rredc \cdot\redbc^* $.
	\end{enumerate}
\end{restatable}

We will  use the following easy property (an instance of \Cref{l:surf_swaps} in the Appendix).
%
%
%
\begin{lem}\label{fact:roots} Let  $\reda, \redc$ be the contextual closure of rules $\rreda,\rredc$.
	Then,  $\nsreda \cdot \rredc \subseteq  {\sredc} \cdot \reda^= $ implies 
			$\nsreda \cdot \sredc \subseteq  {\sredc} \cdot \reda^= $. 
\end{lem}

\paragraph{Root Lemmas.}
\Cref{l:sigma_basic,l:id_basic} below  provide  \emph{everything   we need} to verify the conditions of the modular test in \Cref{prop:test}, and so to establish surface factorization in $\lc$.

\begin{lem}[$\sigma$-Roots]\label{l:sigma_basic} 
	Let $\gamma \in \{\sigma, \id, \betac\}$. The following holds:
	\begin{center}
		{ $M\nsredx{\ \gamma}\ L \mapsto_{\sigma} N$ implies  $M \mapsto_{\sigma} \cdot  \red_{\gamma} N$.}
	\end{center}
\end{lem}
\begin{proof}We have $L = (\lam x.L_1)((\lam y.L_2)L_3)  \mapsto_{\sigma} (\lam y. (\lam x.L_1)L_2)L_3  = N$. Since  $L $ is a $\sigma$-redex,  $M$  also is a $\sigma$-redex (\Cref{lem:redex_bsi}). So   
		$M =(\lam x.M_1)((\lam y.M_2)M_3) \nsredc\allowbreak (\lam x.L_1)((\lam y.L_2)L_3) =L $, where for only one $i\in\{1,2,3\}$  $M_i \nsredc L_i$  and otherwise $M_j=L_j$, for  $j\neq i$  (by \Cref{fact:isteps} in the Appendix). 
		Therefore,  $M = (\lam x.M_1)((\lam y.M_2)M_3)  \mapsto_{\sigma} 
		 (\lam y. (\lam x.M_1) M_2)M_3 \allowbreak\nsredc  (\lam y. (\lam x.L_1)L_2)L_3 = N$.
\end{proof}

\begin{lem}[$\id$-Roots]\label{l:id_basic} 
	Let $\gamma \in \{\sigma, \id, \betac\}$. The following holds:
	\begin{center}
		{ $M\nsredx{\ \gamma}\ L \mapsto_{\id} N$ implies  $M \mapsto_{\id} \cdot  \red_{\gamma} N$.}
	\end{center}
\end{lem}
\begin{proof} We have $ L\eq  \II N \mapsto_{\id} N$.  Since  $L$ is an $\id$-redex,  $M$  also is  (\Cref{lem:redex_bsi}). So,
	$M\eq \II P \nsredx{\gamma} \, \II N $ for some $P \nsredx{\gamma} \, N$. Therefore, $M\eq \II P\mapsto_{\id}  P \nsredx{\gamma} \, N$.
\end{proof}

Let us make explicit the content of the two lemmas above. 
By  instantiating $\gamma \in \{\sigma, \id, \betac\}$ in \Cref{l:sigma_basic,l:id_basic}, and combining them with \Cref{fact:roots},  we obtain the following facts:

\begin{fact}\label{cor:sigma_basic} 
	\begin{enumerate}
	
		\item\label{p:sigma_basic-sigma-sigma}
		$M \nsredx{\sigma}	\cdot \mapsto_{\sigma} N$ implies  $M \mapsto_{\sigma} \cdot  \red_{\sigma}^= N$  and so  (\Cref{fact:roots})
		$M \nsredx{\sigma} \cdot \sredx{\sigma} N$ implies  $M \sredx{\sigma} \cdot  \red_{\sigma}^= N$ (\ie strong postponement holds).
		
		\item \label{p:sigma_basic-id-sigma}
		$M \nsredx{\id} \cdot \mapsto_{\sigma} N$ implies  $M \mapsto_{\sigma} \!\cdot\!  \red_{\id}^=  N$ and so  (\Cref{fact:roots}) 	$M \nsredx{\id} \!\cdot\! \sredx{\sigma} N$ implies  $M \sredx{\sigma} \cdot\!  \red_{\id}^=  N$.

		\item\label{p:sigma_basic-bc-sigma}
		 $M \nsredbc \cdot \mapsto_{\sigma} N$ implies  $M \mapsto_{\sigma} \cdot  \red_{\betac}^= N$.
		
	\end{enumerate}
\end{fact}


\begin{fact}\label{cor:id_basic} 

	\begin{enumerate}
		
		\item\label{p:id_basic-id-id}
		$M \nsredx{\id} \cdot \mapsto_{\id} N$ implies  $M \mapsto_{\id} \cdot  \red_{\id}^= N$, and so (\Cref{fact:roots}) $M \nsredx{\id} \cdot \sredx{\id} N$ implies  $M \sredx{\id} \cdot  \red_{\id}^= N$ (\ie strong postponement holds).
		\item \label{p:id_basic-sigma-id}
		$M \nsredx{\sigma} \cdot \mapsto_{\id} N$ implies  $M \mapsto_{\id} \!\cdot\!  \red_{\sigma}^=  N$, and so (\Cref{fact:roots})
		$M \nsredx{\sigma} \!\cdot\! \sredx{\id} N$ implies  $M \sredx{\id} \cdot\!  \red_{\sigma}^=  N$.
		
		\item \label{p:id_basic-bc-id}
		$M \nsredbc \cdot \mapsto_{\id} N$ implies  $M \mapsto_{\id} \cdot  \red_{\betac}^= N$. 
	\end{enumerate}
\end{fact}

\paragraph{Surface Factorization of  $\redx{\id} \cup \redx{\sigma}$.}
We can now combine the facts above concerning $\sigma$ and $\id$ steps, using the modular approach proposed in \cite{AccattoliFaggianGuerrieri21} (see \Cref{thm:modular} in the Appendix), to prove surface factorization of  $\redx{\id  \sigma} \, = \, \redx{\id} \cup \redx{\sigma} $.
\
\begin{lem}[Surface factorization of $\id  \sigma$]\label{cor:factorization_is}
	\emph{Surface factorization} of $\redx{\id} \cup \redx{\sigma}$ holds, because: 
	\begin{enumerate}
		\item\label{p:factorization-is-sigma} {Surface factorization} of $\redx{\sigma}$ holds (that is, $\redx{\sigma}^* \,\subseteq \sredx{\sigma}^* \cdot \nsredx{\sigma}^*$).   
		\item\label{p:factorization-is-id} 	{Surface factorization} of $\redx{\id} $ holds (that is, $\redx{\id}^* \,\subseteq \sredx{\id}^* \cdot \nsredx{\id}^*$). 
		\item\label{p:factorization-is-swap-id-sigma} Linear swap: $\nsredx{\id} \cdot \sredx{\sigma } \subseteq  \sredx{\sigma} \cdot \RedStar_{\id} $.
		\item\label{p:factorization-is-swap-sigma-id} Linear swap: $\nsredx{\sigma} \cdot \sredx{\id } \subseteq  \sredx{\id } \cdot \RedStar_{\sigma} $.
	\end{enumerate}
\end{lem}
\begin{proof}

	\Cref{p:factorization-is-sigma,p:factorization-is-id} follow from  \Cref{cor:sigma_basic}.\ref{p:sigma_basic-sigma-sigma} and \Cref{cor:id_basic}.\ref{p:id_basic-id-id}, respectively, by (linear) strong postponement (\Cref{lem:SP}).
	\Cref{p:factorization-is-swap-id-sigma}  is   \Cref{cor:sigma_basic}.\ref{p:sigma_basic-id-sigma}.
	\Cref{p:factorization-is-swap-sigma-id}  is   \Cref{cor:id_basic}.\ref{p:id_basic-sigma-id}.
\end{proof}

\paragraph{Surface Factorization of $\lc$, Modularly.}
We are now ready to use the modular test for surface factorization with $\betac$ (\Cref{prop:test}) to prove \Cref{thm:sfactorization}.

\SurfaceFactorization*

\begin{proof}
		All conditions in \Cref{prop:test} hold, namely
	\begin{enumerate}
		\item \emph{Surface factorization} of $\redx{\id} \cup \redx{\sigma}$ holds by \Cref{cor:factorization_is}.
		
		\item \emph{Substitutivity}: ~$\mapsto_{\id}$ and  $\mapsto_{\sigma}$ are substitutive (the proof is immediate).
		
		\item \emph{Root linear swap}: for $\xi \in \{\id,\sigma\}$, 
		$\nsredbc \cdot \mapsto_{\xi}\ \subseteq\ \mapsto_{\xi} \cdot  \red_{\betac} ^=$ by \Cref{cor:sigma_basic}.\ref{p:sigma_basic-bc-sigma} and \Cref{cor:id_basic}.\ref{p:id_basic-bc-id}.
		\qedhere
	\end{enumerate}	
\end{proof}

	
Interestingly, the same machinery can also be used to prove another surface factorization result, which says that surface factorization, when applied to $\RedStar_{\bs}$ only, does not create $\redid$ steps.

\begin{prop}[Surface factorization of $\bs$]
	\label{prop:sfactorization-bs}
	Reduction $\redx{\bs}$ admits surface factorization:
	\begin{align*}
		M\RedStar_{\bs}N \text{ implies }  M \sredxStar{\bs} \cdot  \nsredxStar{\bs} N.
	\end{align*}
\end{prop}

\begin{proof}
	All conditions in \Cref{prop:test} hold, namely
	\begin{enumerate}
		\item \emph{Surface factorization} of $\redx{\sigma}$ holds by \Cref{cor:sigma_basic}.\ref{p:sigma_basic-sigma-sigma} and  
		 strong postponement (\Cref{lem:SP}).
		
		\item \emph{Substitutivity}: ~$\mapsto_{\sigma}$ is substitutive (the proof is immediate).
		
		\item \emph{Root linear swap}: 
		$\nsredbc \cdot \mapsto_{\sigma}\ \subseteq\ \mapsto_{\sigma} \cdot  \red_{\betac}^=$ by \Cref{cor:sigma_basic}.\ref{p:sigma_basic-bc-sigma}.
		\qedhere
	\end{enumerate}	
\end{proof}

The fact that two similar factorization results (\Cref{thm:sfactorization,prop:sfactorization-bs}) can be proven by means of the \emph{same} modular test (\Cref{prop:test}) fed on \emph{similar} lemmas shows one of the benefits of our modular approach: passing from one result to the other is smooth and effortless.

\subsection{A Closer Look at $\id$ Steps, via Postponement}\label{subsec:iotapostponement}
We show a   postponement result for $\id$ steps on which  evaluation (\Cref{sec:values}) and normalization (\Cref{sec:normalization}) rely.
\condinc{}{
	Observe that the rule $\redid$ overlaps with $\redbc$.
	We define  as $\redi$ a $\redid$ step that is not a $\redbc$. In terms of rules: 
	\begin{equation}\tag{$\iota$ rule}
		\mapsto_{\ii} \deff \mapsto_{\id} \smallsetminus \mapsto_{\betac}
	\end{equation}
	Clearly, $\redlc \eq \redbc \cup \reds \cup \redi$.
}
Note that $\redid$ overlaps with $\redbc$.
We define $\redi$ as a $\redid$ step that is not $\redbc$.
 
\begin{equation}\tag{$\iota$-rule}
	\mapsto_{\ii} \deff \mapsto_{\id} \smallsetminus \mapsto_{\betac}
\end{equation}
Clearly, $\redlc \eq \redbc \cup \reds \cup \redi$.
In the proofs, it is convenient to split  $\redbc$ into steps that are also $\redid$ steps, and those that are not.
 \begin{align}\tag{$\beta1$-rule}
 	\mapsto_{\beta1} \deff \mapsto_{\id} \cap  \mapsto_{\betac}
	\\
	\tag{$\beta2$-rule}
	\mapsto_{\beta2} \deff \mapsto_{\betac} \smallsetminus \mapsto_{\id}
\end{align}
 Notice  that  $\redbc  \eq \red_{\beta1}  \cup \red_{\beta2}$.

We prove that $\redi$ steps can be postponed after both $\redbc$ and $\reds$ steps, by using van Oostrom's decreasing diagrams technique \cite{Oostrom94,DD} (\Cref{thm:DD}). 
The proof closely follows van Oostrom's proof of the  postponement of $\eta$ after $\beta$ \cite{vO20}.

\begin{restatable}[Postponement of $\ii$]{thmm}{DDiotapostponement}\label{thm:i_postponement}
If	$M\redx{\lcc}^*N$ then $M  \redx{\bs}^* \cdot   \redx{\ii}^* N$.
\end{restatable}
\begin{proof}
Let $\revredL{} \eq \redi$ and let $\redM{} \eq \redx{\beta1} \cup \redx{\beta2} \cup \reds$. 
We equip the labels $\{\beta1, \beta2, \sigma, \iota\}$ with the following (well-founded) order
\[\beta2<\iota  \quad \quad \iota<\beta1   \quad\quad \iota < \sigma\]
We prove that the pair of relations 
$\redL{},\redM{}$ is decreasing by checking the following three local commutations  hold (the technical details are in the Appendix):
\begin{enumerate}
	\item  $\redi \cdot \redx{\beta1} ~{\subseteq}~  \redx{\beta1}^* \cdot \redi^= $  (see \Cref{lem:iotaVSbeta1});
	\item $\redi \cdot \redx{\beta2} ~{\subseteq}~ \redx{\beta1}^* \cdot \redx{\beta2}^= \cdot\redx{\beta1}^* \cdot\redi^* $ (see \Cref{lem:iotaVSbeta2});
	\item $\redi \cdot \reds ~{\subseteq}~   (\reds \cup \redx{\beta1 })^* \cdot \redi^= $  (see \Cref{lem:iotaVSsigma}).
\end{enumerate}
Hence, 
by  \Cref{thm:DD}, the relations  $\redL{}$ and $\redM{}$ commute. That is,  $\redi$ postpones after 
$\redbc \cup \reds $:
\[\redi^* \cdot \redx{\bs}^*   ~\subseteq~ \redx{\bs}^* \cdot \redi^*  \]
Or equivalently (\Cref{lem:postponement_eq}),  $\redlc^*  ~\subseteq~ \redx{\bs}^* \cdot \redi^*$.
\end{proof}

\begin{restatable}[Surface factorization + $\ii$ postponement]{coroll}{ipostponement}\label{cor:i_postponement} 
	\label{cor:surface_i}
 If $M\RedStar_{\lcc} N$ then  $M \sredxStar{\bs } \!\cdot \nsredxStar{\bs } \!\cdot  \redx{\ii}^*   N$.
\end{restatable}

\begin{proof}
	Immediate consequence of $\ii$-postponement (\Cref{thm:i_postponement}) and of surface factorization of the resulting initial $\redx{\bs}$-sequence (\Cref{prop:sfactorization-bs}).
\end{proof}

%
%

%% file: Values.tex

\subsection{Weak Factorization}\label{sec:w_factorization}

\newcommand{\bsize}[1]{|#1|}

Thanks to surface factorization plus $\ii$-postponement (\Cref{cor:i_postponement}), every $\redlc$-sequence can be rearranged so that it starts with an $\sredx{\bs}$-sequence. 
We show now that such an initial $\sredx{\bs}$-sequence can in turn be factorized into weak steps followed by non-weak steps.
This \emph{weak factorization} result will be used in \Cref{sec:values} to obtain evaluation via weak $\betac$ steps (\Cref{thm:values}).

Remarkably, weak factorization of an $\sredx{\bs}$-sequence \emph{preserves}  the number of $\betac$ steps. 
This property has no role with respect to evaluation, but it will be  crucial when  we investigate   normalizing strategies in \Cref{sec:normalization}. 
For this reason we include it in the statement of \Cref{prop:weak_fact}.

\SLV{}{ It will allows us to establish that all maximal $\wredx{\bs}$-sequences have the same number of $\beta$-steps, from which it follows that such (non deterministic and non confluent ) reduction is uniformly normalizing. We obtain a similar result follows also for the  $\sredx{\bs}$-sequences.}

\paragraph{Quantitative Linear Postponement.}
Let us take an abstract point of view.
The condition in \Cref{lem:SP}---Hindley's strong postponement---can  be refined into quantitative (linear) variants, which  allow us to ``count the steps'' and are useful to establish termination properties.

\begin{lem}[Linear  postponement]\label{lem:quantitative_postponement} Let $(A, \red)$ be an ARS with $\red \eq \ered \cup \ired$.
	\begin{itemize} 
		\item  	If  $\ired \cdot \ered \ssubseteq  \ered \cdot \ired^=$, then 
		$M\red^*N$ implies $M \ered ^*\cdot \ired^* N$ and the two sequences have the same number of $\ered$ steps.
		
%
		\item For all $l\in L$ with $L$ a set of indices, let $\redx l \eq  \eredx{\, l} \cup \iredx{\,\, l}$ and $\ered =\bigcup\limits_{l} \eredx{\, l}$ and $\ired = \bigcup\limits_{l} \iredx{\,\, l}$. 
		Assume
			\begin{align} \label{eq:linear-postponement}
				\iredx{\,\, j} \cdot \eredx{\,\, k} \ssubseteq  \eredx{\,\, k} \cdot \redx j~ \text{ for all } j,k\in L
			\end{align} 
			Then,
			$M\red^* N$ implies $M \ered ^*\cdot \ired^* N$ and the two sequences have the same number of $\redx l$ steps, for each $l\in L$.
	\end{itemize}
	
\end{lem}

			Observe  that in \eqref{eq:linear-postponement}, the last step is $ \redx j $, not necessarily $ \iredx{\,\, j} $.

\paragraph{Weak Factorization.}

We show two kinds of \emph{weak factorization}: a surface $\betac\sigma$ sequence can be reorganized, so that non-weak $\betac\sigma$ steps are postponed after the weak ones; 
and a weak $\betac\sigma$ sequence can in turn be rearranged so that weak $\betac$ steps are before weak $\sigma$ steps. 

\begin{thmm}[Weak factorization]\label{prop:weak_fact}\hfill
	\begin{enumerate}
		\item\label{p:weak_fact-from-surface} If	$M\sredxStar{\bs} N$ then $M\wredxStar{\bs} \cdot \nwredxStar{\bs} N$ where all steps are surface, and the two sequences have the same number of $\betac$ steps.
		\item\label{p:weak_fact-from-weak} If	$M\wredxStar{\bs} N$ then $ M\wredxStar{\betac} \!\cdot \wredxStar{\sigma} N $, and the two sequences have the same number of $\betac$ steps.
	\end{enumerate}
\end{thmm}

\begin{proof}In both claims, we use \Cref{lem:quantitative_postponement}. 
	Its linearity allows us to count the $\betac$ steps.
	In the proof, we write $\wred$ (resp.  $\sred $) for $\wredx{\bs}$ (resp.  $\sredx{\bs} $).

	\begin{enumerate}
		\item Let 	  $\ired \eq \sred \smallsetminus \wred$ (\ie $\ired$ is a surface step whose redex is in the scope of a $\lambda$). We   prove linear postponement: 
		\begin{align}\label{eq:linear-postponement-bis}
			 \ired \cdot \wred  \ssubseteq  \wred \cdot \sred
		\end{align}
		
		Assume $M\ired L \wred N$:  $M$ and $L$ have the same shape, which is not 
		$\oc U$, otherwise no weak or surface step from $M$ is possible. 
		We examine the cases.
		\begin{itemize}
			\item The step $L\wred N$ has  empty context:
			\begin{itemize}
				\item $L\mapsto_{\betac}N.$
				Then  $L= (\lam x. P') !V \rredbc   P'\subs x V = N$, and $M= (\lam x. P) !V \ired  (\lam x. P') !V   $ with $P \sred P'$. \\Therefore   $ (\lam x. P) !V \rredbc P\subs x V \sred P'\subs x V = N$.
				
				\item $L \mapsto_{\sigma}N$. 
				Then $L = V ( (\lam x. P) Q) \mapsto_{\sigma}  (\lam x. V P) Q = N $,
				and    $M=  V_0 ( (\lam x. P_0) Q_0) \ired\allowbreak  V ( (\lam x. P) Q)$ where exactly one among $V_0, P_0, Q_0$ reduces to $V, P, Q$, respectively, the other two are unchanged.
				So, $M=  V_0 ( (\lam x. P_0) Q_0)   \mapsto_{\sigma}  (\lam x. V_0 P_0) Q_0 \ired (\lam x. V P) Q = N $.
			\end{itemize}
			\item The step $L \wred N$ has  non-empty context. Necessarily, we have  $L = VQ \wred VQ'$, with $Q \wred Q'$:
			
			\begin{itemize}
				\item Case   $M=M_VQ\ired VQ  \wred VQ' = N$ with $M_V \ired V$ and $Q \wred Q'$, then $M_VQ \wred M_VQ'  \ired VQ'$.

				\item Case $M=VM_Q\ired VQ \wred VQ' = N$ with $M_Q \ired Q \wred Q$. We conclude by \ih.
				
			\end{itemize}
		\end{itemize}	
		
		Observe that we have proved more than \eqref{eq:linear-postponement-bis}, namely we proved
		\[\iredx{\;j} \cdot \wredx k  \ssubseteq  \wredx k \cdot \sredx j  ~~~ (\text{for all } j,k\in\{\betac,\sigma\})\]
		So, we conclude that the two sequences have the same number of $\betac$ steps, by  \Cref{lem:quantitative_postponement}.
		
		\item We prove $\wreds \cdot \wredbc \ssubseteq  \wredbc \cdot \wreds^=$ similarly to \Cref{p:weak_fact-from-surface}, and conclude by \Cref{lem:quantitative_postponement}.
		\qedhere
	\end{enumerate}

\end{proof}

Combining \Cref{p:weak_fact-from-surface,p:weak_fact-from-weak} in \Cref{prop:weak_fact}, we deduce that 
\begin{center}
	$M\sredxStar{\bs} !V$ implies $ M\wredxStar{\betac}\cdot \wredxStar{\sigma} \cdot \nwredxStar{\bs}\oc V $
\end{center}
and the two sequences from $M$ to $\oc V$ have the same number of $\betac$ steps.

%

\section{Returning a Value}\label{sec:values}

In this section we focus on \emph{values}. 
They are the terms of interest in the CbV $\lam$-calculus. 
Also, for weak reduction there,  closed values
are exactly the normal forms of closed terms, \ie  of \emph{programs}.

 In a computational setting such as $\lc$, we are interested in knowing if a term $M$
\emph{returns} a value, \ie if $M\redlc^* \oc V$ for some value $V$, noted $M \Downarrow$ (the computation $\oc V$ is sometimes called a \emph{returned value}, in that it is the coercion of a value $V$ to the computational level).
Since a term may be reduced in several ways and so its reduction graph can become quite complicated, it is natural to search for deterministic reductions to return a value.
Hence, the question is: if $M$ returns a value, is there a \emph{deterministic} reduction (called \emph{evaluation}) from $M$ that is guaranteed to return a value? The answer is positive. 
In fact, there are two such reductions:
$\wredbc$ and $ \Root{\betac\sigma}  $ (\Cref{thm:return-value} below).
%
%
Recall that   $\mapsto_{\bs} \eq (\mapsto_{\betac}  \cup \mapsto_{\sigma})$ is the union of two rules without any contextual closure. 

Thanks to their simple reduction graph, deterministic reductions are quite useful in particular for proving negative results such as showing that a computation cannot return a value.

\begin{fact}\label{fact:determinism}
	In $\lc$, reductions $\wredx{\betac}$, $\Root{\betac}$, $ \Root{\sigma}$, and $ \Root{\betac\sigma} \,=\, \Root{\betac} \!\cup \Root {\sigma} $ are deterministic.
\end{fact}

In $\lc$ one of the reasons for the interest in values is that, akin to the CbV $\lam$-calculus, \emph{closed} (\ie, without free variables) returned values are exactly the closed normal forms for weak reductions $\wredx{\lcc}$ and $\wredbc$.
This is a consequence of the following syntactic characterizations of normal forms.

\begin{prop}
	\label{prop:normal}
	A computation is normal for reduction $\wredx{\lcc}$ (resp. $\sredx{\lcc}$; $\redlc$) if and only if it is of the form $\normalweak$ (resp. $\normalsurface$; $\normalfull$) defined below, where $\hat{M}$ denotes a computation $M \neq \oc x$ for any $x \in \Var$.
	
	\begin{align*}
		\normalweak &\Coloneqq \oc V \mid \appnormalweak \mid (\lambda x. \hat{M})\appnormalweak & \appnormalweak &\Coloneqq x\normalweak
		\\
		\normalsurface &\Coloneqq \oc V \mid \appnormalsurface \mid (\lambda x. \hat{\normalsurface})\appnormalsurface & \appnormalsurface &\Coloneqq x\normalsurface
		\\
		\normalfull &\Coloneqq \oc x \mid \oc \lam. \normalfull \mid \appnormalfull \mid (\lambda x. \hat{\normalfull})\appnormalfull & \appnormalfull &\Coloneqq x\normalfull
	\end{align*}
\end{prop}

\begin{proof}
	The right-to-left part is proved by induction on $\normalweak$ (resp. $\normalsurface$; $\normalfull$).
	The left-to-right part follows easily from the observation that every computation can be written in a unique way as $V_1(\dots (V_n \oc V_0) \dots )$ for some $n \geq 0$ and some values $V_0, \dots, V_n$.
\end{proof}

\begin{coroll}[Closed normal forms]
	\label{cor:closed-normal}
	Let $\to \in \{\wredbc, \wredx{\lcc}, \sredbc, \sredx{\lcc}\}$. 
	A closed computation is $\to$-normal if and only if it is a returned value.
\end{coroll}

\begin{proof}
	 Computations of shape $\appnormalweak$ and $\appnormalsurface$ have a free variable.
	 So,  according to \Cref{prop:normal}, closed returned values are all and only the closed normal forms for $\wredx{\lcc}$ and $\sredx{\lcc}$.
	 
	 Moreover, since every closed computation $M$ can be written in a unique way as $V_1(\dots (V_n \oc V_0) \dots )$ for some $n \geq 0$ and some closed values $V_0, \dots, V_n$, if $M$ is $\wredbc$-normal or $\sredbc$-normal then $n=0$ (otherwise $V_n \oc V_0$ would be a $\betac$-redex), hence $M$ is a returned value.
\end{proof}

\Cref{cor:closed-normal} means that reductions $\wredbc, \wredx{\lcc}, \sredbc, \sredx{\lcc}$ behave differently only on \emph{open} computations (that is, with at least one free variable).

We can now state the main result in this section. 
\Cref{sec:values_via_betab,sec:values_via_roots}  are devoted to prove it.

\begin{thmm}[Returning a value] 
	\label{thm:return-value}
	The following are equivalent:
	\begin{enumerate}
		\item\label{p:return-value-full} $M$ returns a value, \ie $M\redlc^* \oc V$.
		
		\item\label{p:return-value-bs}   The maximal $\wredx{\betac}$-sequence from $M$ is \emph{finite} and ends in a returned value $!W$.
		
		\item\label{p:return-value-root}   The maximal $\Root{\betac \sigma}$-sequence from $M$ is \emph{finite} and ends in a returned value   $!W$.
	\end{enumerate}
\end{thmm}

\begin{proof} $(\labelcref{p:return-value-full}) \implies (\labelcref{p:return-value-bs})$ is \Cref{thm:values} below, which we prove in forthcoming \Cref{sec:values_via_betab}.
	
	$(\labelcref{p:return-value-bs}) \implies (\labelcref{p:return-value-root})$ is \Cref{lem:two_returns} below, which we prove in forthcoming \Cref{sec:values_via_roots}.
	
	$(\labelcref{p:return-value-root}) \implies (\labelcref{p:return-value-full})$ is trivial.
\end{proof}

Note that \Cref{thm:return-value} (and hence the analysis that will follow) is not restricted to closed terms.	
Indeed, an open term may well return a value. For example, $!(\lam x. \oc z)$ or $\oc x$ or $(\lam x. \oc x) \oc z$.

\subsection{Values via  Weak $\betac$ Steps}\label{sec:values_via_betab}
Thanks to factorization, we can prove  that  $\wredbc$ steps suffice to return a value.
This is an immediate consequence of 
 surface factorization plus $\ii$ postponement (\Cref{cor:i_postponement}), and weak factorization (\Cref{prop:weak_fact}), and  the fact that    non-weak steps, $\ii$ steps, and $\sigma$ steps cannot produce $!$-terms.

 \begin{lem}\label{lem:!preservation}
		If $M \redx{\lcc} \oc V$ with a  step that is \emph{not}  $M \wredx{\betac} \oc V$, then $M=\oc W$ for some value $W$.
  \end{lem}
 \begin{proof}

Indeed, one can easily check the following (recall that $\redi \,=\, \redid \smallsetminus \redbc$).
 	\begin{itemize}	
 		
 		\item 	If $M \redx{\sigma} \oc V$, then $M = \oc W$ for some value $W$ (proof by induction on $M$). 
 		
 		\item 	If $M \redx{\ii} \oc V$, then $M = \oc W$ for some value $W$ (proof by induction on $M$). 
 		 		
 		
 		 \item 	If $M \nwredx{\betac} \oc V$, then $M = \oc W$ for some value $W$ (by shape preservation,  \Cref{fact:shape}).
 		 \qedhere
 	\end{itemize}

  \end{proof}


 \begin{restatable}[Values via weak $\betac$ steps]{thmm}{StateValues}\label{thm:values}	The following are equivalent:
 	\begin{enumerate}
 		\item\label{p:values-lc} $M\redlcStar \oc V$ for some $V\in \ValTerm$;
 		
 		\item\label{p:values-betaweak}  $M\wredxStar{\betac} \oc W$  for some $W \in \ValTerm$.
 	\end{enumerate}
 \end{restatable}
 

\begin{proof}
	\Cref{p:values-betaweak} trivially implies \Cref{p:values-lc}, as $\wredx{\betac} \subseteq \, \redlc$.
	Let us show that Point \ref{p:values-lc} entails \Cref{p:values-betaweak}.
	
	If $M\redlcStar \oc V$ then
	$M \sredxStar{\bs } \cdot \nsredxStar{\bs } \cdot  \redx{\ii}^*   \oc V$ by surface factorization plus $\ii$ postponement (\Cref{cor:surface_i}).
	By weak factorization (\Cref{prop:weak_fact}.\ref{p:weak_fact-from-surface}-\ref{p:weak_fact-from-weak}), we have 
\[	M \wredxStar{\betac}  M'  \wredxStar{\sigma} \cdot  \nwredxStar{\bs} \cdot  \nsredxStar{\bs} \cdot \redi^*  \oc V.\]

	
	By iterating \Cref{lem:!preservation} from $!V$ backwards (and since $\nsredx{\Rule} \subseteq \nwredx{\Rule}$), we have that all terms  in the sequence from $M'$ to $!V$ are $!$-terms.
	So in particular, $M'$ has shape $\oc W$ for some value $W$.
\end{proof}
%
%

\begin{rem}
\Cref{thm:values} was  already claimed in \cite{deLiguoroTreglia20}, for closed terms.
However, the inductive  argument there (which does not use any factorization) is fallacious, it does not suffice to produce  a complete proof in the case where $M \nwredx{\lcc} \cdot \wredx{\lcc} \oc V$. 

\end{rem}

\subsection{Values  via $\bs$ Root Steps}\label{sec:values_via_roots}
We also show   an alternative way   to evaluate a term in $\lc$. 
Let us call \emph{root steps} the rules $\mapsto_{\betac}, \mapsto_{\sigma}$ and $\mapsto_{\id}$. 
The first two suffice to return a value, without the need for any  contextual closure.

Note that this property  holds only because  terms are restricted  to computations (for example, in Plotkin's CbV $ \lam $-calculus, 
$(II)(II)$ can be reduced, but it is not itself a redex, so $(II)(II) \not\mapsto_{\betav}$).

%
%
Looking closer at the proof of \Cref{cor:closed-normal}, we observe that any closed (\ie without free variables) computation has the following property: it \emph{is} either a returned value (when $n = 0$), or a $\betac$-redex (when $n = 1$) or a $\sigma$-redex (when $n > 1$).
More generally, the same holds for any (possibly open) computation that returns a value (\Cref{cor:root_redex} below).

\begin{lem}\label{lem:weakbeta} Assume $M \wredxStar{\betac} \oc W$ for some value $W$. Then, 
	\begin{itemize}
		\item either $M = \oc W$,
		\item or $M=(\lam x.P)M'$  and $M'\wredxStar{\betac} \oc U$, for some value $U$.
	\end{itemize}

Thus, $M = V_1(\dots (V_n!U) \dots)$, where $n \geq 0$ and the $V_i$'s are abstractions, and  if $n > 0$ then
	 $M = V_1 \dots (V_{n-1}(\lam x_n.P_n)!U) \dots) \,\wredbc\, V_1(\dots (V_{n-1} P_n \subs{x_n}U ) \dots ) $. 
\end{lem}

\begin{coroll}[Progression via root steps]
	\label{cor:root_redex}
	If $M$ returns a value (\ie $M \redlc^* \oc W$ for some value $W$), then   $M$ 
is either a $\betac$-redex, or  a $\sigma$-redex, or it  has shape $\oc V$ for some value $V$. 
\end{coroll}

\begin{proof}
	By \Cref{thm:values}, $M \wredbc^* \oc W'$ for some value $W'$.
	By \Cref{lem:weakbeta}, we conclude. 
\end{proof}

\Cref{cor:root_redex} states a \emph{progression} results: a $\Root{\bs}$-sequence from $M$ may only end in  a $!$-term.
We still need to verify that such a sequence terminates.

\begin{prop}[Weak steps and root steps]\label{lem:two_returns}
If $M \wredxStar{\betac} !W$ then $M \Root{\betac \sigma}^* \oc W$.
 Moreover, the two sequences have the same number of $\betac$ steps.
\end{prop}
\begin{proof}
	By induction on the number $k$ of $\wredbc$ steps. If $k=0$ the claim holds trivially.	
	Otherwise, $M\wredbc M_1\wredxStar{\betac} !W$ and by \ih 
	\begin{align}\label{eq:root}
		M_1 \Root{\betac \sigma}^* \oc W.
	\end{align}

	\begin{itemize}
		\item 	If $M$ is $\betac$-redex, then $M\Root{\betac} M_1$ by determinism of $\wredbc$ (\Cref{fact:determinism}), and the claim is proved.

\item 	If  $M$ is a $\sigma$-redex, observe that  by \Cref{lem:weakbeta}, 
\begin{itemize}
	\item $M= (\lam x_0.P_0)(\dots (\lam x_{n-1}.P_{n-1})((\lam x_n.P_n)\oc U) \dots )$, and 
	\item $M_1\eq (\lam x_0.P_0)(\dots (\lam x_{n-1}.P_{n-1})(P_n \subs{x_n}U) \dots )$.
\end{itemize}

 	We apply all possible $\Root{\sigma}$ steps starting from $M$, obtaining 
	\begin{align*}
	M &\Root{\sigma}^*  ( \lam x_{n-1}.(\dots (\lam x_0.P_0)\dots ) P_{n-1}) ((\lam x_n.P_n)\oc U)  
	\\
	&\Root{\sigma}  (\lam x_n.(\lam x_{n-1}.(\dots (\lam x_0. P_0) \dots ) P_{n-1}) P_n)\oc U = M'
	\end{align*}
	which is a $\betac$-redex, so $M'\Root{\betac} (\lam x_{n-1}.(\dots (\lam x_0. P_0) \dots ) P_{n-1}) (P_n\subs{x_n}{U}) \eqdef N$
	(note that we used the hypothesis on free variables of $\Root{\sigma}$).
	We observe that  $M_1 \Root{\sigma}^* N$. 
	We conclude, by using $\eqref{eq:root}$ and the fact that $\Root{\bs}$ is deterministic (\Cref{fact:determinism}).
	\qedhere
	\end{itemize}	
\end{proof}
The converse of \Cref{lem:two_returns} is also true and immediate.
We can finally prove that \emph{root} steps $\mapsto_{\betac}$ and $\mapsto_{\sigma}$ suffice to return a value, without the need for any contextual closure.

\begin{restatable}[Values via root $\betac\sigma$ steps]{thmm}{StateValuesRoot}\label{thm:values-root}	The following are equivalent:
	\begin{enumerate}
		\item\label{p:values-root-lc} $M\redlcStar \oc V$ for some $V\in \ValTerm$;
		
		\item\label{p:values-root-betas}  $M\Root{\betac\sigma}^* \oc W$  for some $W \in \ValTerm$.
	\end{enumerate}
\end{restatable}

\begin{proof}
Trivially $(\labelcref{p:values-root-betas}) \implies (\labelcref{p:values-root-lc})$. 
Conversely, $(\labelcref{p:values-root-lc}) \implies (\labelcref{p:values-root-betas})$ by \Cref{lem:two_returns} and \Cref{thm:values}.
\end{proof}

\subsection{Observational Equivalence}\label{sec:adequacy}
We now adapt the notion of observational equivalence, introduced in \cite{Plotkin'75} for the CbV $\lam$-calculus, to $\lc$. 
Informally, two terms  are observationally equivalent if they can be substituted for each other in all contexts without observing any difference in their behavior. For a computation $M$ in $\lc$, the ``behavior'' of interest   is \emph{returning a value}: $M \redlc^* \oc V$ for some value $V$, also noted $M \Downarrow$. 

\begin{defn}[Observational equivalence]
	\label{def:observational-equivalence}
	Let $M, N \in \ComTerm$.
	We say that $M$ and $N$ are \emph{observationally equivalent}, noted $M\cong N$, if for every context $ \cc $, $\cc\hole{M} \Downarrow$ if and only if $\cc\hole{N} \Downarrow$.
\end{defn}

A consequence of \Cref{thm:values} is that the behavior of interest in \Cref{def:observational-equivalence} can be equivalently defined as  $\cc\hole{M} \wredbc^* \oc V$ for some value $V$, instead of  $\cc\hole{M} \Downarrow$: the resulting notion of observational equivalence would be exactly the same.
The definition using $\wredbc$ instead of $\redlc$ is more in the spirit of Plotkin's original one for  the CbV $\lam$-calculus \cite{Plotkin'75}.
Reduction $\wredbc$ is deterministic, and for closed terms it terminates if and only if it ends in a returned value (\Cref{cor:closed-normal}). 
Hence, for \emph{closed} terms, returning a value amounts to say that their evaluation $\wredbc$ \emph{halts}.

The advantage of our \Cref{def:observational-equivalence} is that it allows us to prove an important property of observational equivalence---the fact that it contains the equational theory of $\lc$ (\Cref{cor:observational-equivalence})---in a very easy way, thanks to the following obvious lemma and adequacy (\Cref{thm:adequacy}).

\begin{lem}[Value persistence]
	\label{lem:value-persistence}
	For every value $V$, if $\oc V \redlc M$ then $M = \oc W$ for some value $W$.
\end{lem}

\begin{proof}
	In $\lc$, no redex has shape $\oc V$ for any value $V$, hence the step $\oc V \redlc M$ is obtained via a non-empty contextual closure.
	By shape preservation (\Cref{fact:shape}), $M = \oc W$ for some value $W$.
\end{proof}

An easy argument,  similar   to that  in \cite{Crary09} (which in turn simplifies the one in \cite{Plotkin'75}) gives:

\begin{thmm}[Adequacy] 
	\label{thm:adequacy}
	If $M\RedStar_{\lcc} N$ then $M \Downarrow$   if and only if $N \Downarrow$.
\end{thmm}

\begin{proof} 
	Suppose $N \Downarrow$, that is, $N \RedStar_{\lcc} \oc V$ for some value $\oc V$. 
	Therefore, $M \RedStar_{\lcc} N \RedStar_{\lcc} \oc V$ and so $M \Downarrow$.
	
	Conversely, suppose $M \Downarrow$, that is, $M \RedStar_{\lcc} \oc V$ for some value $V$. 
	By confluence of $\redlc$ (\Cref{prop:confluence}), since $M \redlcStar N$, there is $L \in \ComTerm$ such that $N \redlcStar L$ and $\oc V \redlcStar L$.	
	Since $\oc V$ is a returned value, so is $L$ by value persistence (\Cref{lem:value-persistence}). Therefore, $N \Downarrow$.
\end{proof}

\begin{coroll}[Observational equivalence contains equational theory]
	\label{cor:observational-equivalence}
	If $M \eqlc  N$ then $M\cong N$.
\end{coroll}

\begin{proof}
 As $M \eqlc N$, there are $L_0, \dots, L_n \in \ComTerm$ ($n \geq 0$) such that $M = L_0 \leftrightarrow_{\lcc} L_1 \leftrightarrow_{\lcc} \dots \leftrightarrow_{\lcc} L_n = N$, where $\leftrightarrow_{\lcc} \, \defeq \, \redlc \cup \revred_{\lcc}$.
 Hence, for every context $\cc$, $\cc\hole{L_0} \leftrightarrow_{\lcc} \cc\hole{L_1} \leftrightarrow_{\lcc} \dots \leftrightarrow_{\lcc} \cc\hole{L_n}$.
 By adequacy (\Cref{thm:adequacy}), $\cc\hole{L_i} \Downarrow$ if and only if $\cc\hole{L_{i+1}} \Downarrow$ for all $1 \leq i < n$. 
 Thus, $M \cong N$.
\end{proof}

The converse of \Cref{cor:observational-equivalence} fails.
Indeed, $\oc \lambda x.\oc x \cong \oc \lambda x. \oc\lambda y. x \oc y$ but $\oc \lambda x.\oc x \not\eqlc \oc \lambda x. \oc\lambda y. x \oc y$.

%% file: Normalization.tex

\newcommand{\s}{\mathfrak{s}}
\section{Normalization and Normalizing Strategies}\label{sec:normalization}
In this section we study normalization and normalizing strategies in $\lc$.


Reduction $\redlc$ is obtained by adding $\redi$ and $\reds$ to $\redbc$.  
What is the role of $\ii$ steps and $\sigma$ steps with respect to normalization in $\lc$?
Perhaps surprisingly, despite the fact that both $\redi$ and $\reds$ are strongly normalizing (\Cref{lemma:strong-normalizing-idsigma} below), their role is quite different.

\begin{enumerate}
		\item\label{p:betasigma} Unlike the case of terms returning a  value we studied  in \Cref{sec:values}, $\betac$ steps do not suffice to capture   $\lcc$-normalization, in that $\sigma$ steps may turn a $\betac$-normalizing term into one that is not $\lcc$-normalizing. That is, $\sigma$ steps are \emph{essential} to   normalization in $\lc$ (see \Cref{subsect:betasigma}).
	
		\item\label{p:irrelevance} $\ii$ steps  instead are \emph{irrelevant} for normalization in $\lc$, in the sense that they  play no role. 
		Indeed, a term 
		has a $\lcc$-normal form if and only if it has a $\bs$-normal form
		(see \Cref{subsect:irrelevance}).

\end{enumerate}

Taking  into account both  \Cref{p:irrelevance,p:betasigma}, in \Cref{subsect:normalizing-strategies} we    define two families of normalizing strategies in $\lc$.
The first one, quite constrained,  relies on an\emph{ iteration of weak reduction} $\wredx{\lcc}$.
The second one, more liberal, is based on an \emph{iteration of surface reduction} $\sredx{\lcc}$.
%
 The interest of a rather liberal strategy is that it provides \emph{a more versatile framework} to reason about program transformations, or optimization techniques such as parallel  implementation.

\paragraph{Technical Lemmas: Preservation of Normal Forms.}
We collect here some properties of preservation of (full, weak and surface) normal forms, which we will use along the section.
{The easy proofs are in Appendix \ref{app:normalization}.}

\begin{restatable}[]{lem}{lemredi}
	\label{lem:redi}
	Assume $M\redi N$.
	\begin{enumerate}
		\item\label{p:redi-beta-normal} $M$ is $\betac$-normal  if and only if  $N$ is $\betac$-normal.
		\item\label{p:redi-sigma-normal} If $M$ is $\sigma$-normal,  so is $N$. 
	\end{enumerate}
\end{restatable}

\begin{lem}
	\label{lem:wbeta_nf} \label{lem:wsbeta_nf}
	If 	$M \reds N$, then: 	$M$ is $\wredbc$-normal  if and only if  so is $N$. 
\end{lem}

\Cref{lem:wbeta_nf} fails if we replace $\wredbc$ with $\sredbc$. Indeed, $M\reds N$ for some $M$ $\sredbc$-normal does not imply that $N$ is  $\sredbc$-normal, as we will see in~\Cref{ex:blocked_beta}.

\begin{restatable}{lem}{lemsurfacenf}
	\label{lem:surface_nf} 
	Let $\ex \in \{ \weak, \surf\}$. 
	If $M \neredx{\bs} N$ then: $M$ is $\eredx{\bs}$-normal  if and only if $N$ is $\eredx{\bs}$-normal.

\end{restatable}

\SLV{}{
\RED{TO CHECK:\\
	($M\reds N$ and $N$ is $\sredbb$-normal) implies  $M$ is $\sredbb$-normal}
}

\subsection{Irrelevance of $\ii$ Steps for Normalization}
\label{subsect:irrelevance}

We show that postponement of $\ii$ steps (\Cref{thm:i_postponement}) implies that $\redi$ steps have no impact on  normalization, \ie whether a term $M$ has or not a $\lcc$-normal form. 
Indeed, saying that $M$ has a $\lcc$-normal form is equivalent to say that $M$ has a $\betac\sigma$-normal form.


On the one hand, if $M\RedStar_{\bs}N$ and $N$ is $\bs$-normal, to reach a $\lcc$-normal form it suffices to extend the reduction with $\ii$ steps to a $\ii$-normal form (since $\redi$ is terminating, \Cref{lemma:strong-normalizing-idsigma}). Notice that here  we use \Cref{lem:redi}.
On the other hand, the proof that  $\lcc$-normalization implies  $\betac\sigma$-normalization is trickier, because $\sigma$-normal forms are not preserved by performing a $\ii$ step backward (the converse of \Cref{lem:redi}.\ref{p:redi-sigma-normal} is false). Here is a counterexample.
\begin{exmpl}
	Consider $(\lambda x. x \oc x)(I(z\oc z)) \redi (\lambda x. x \oc x) (z \oc z)$, where $(\lambda x. x \oc x) (z\oc z)$ is $\sigma$-normal (actually $\lcc$-normal) but $(\lambda x. x \oc x)(I(z\oc z))$ is not $\sigma$-normal.
\end{exmpl} 

Consequently, the fact that $M$ has a $\lcc$-normal form $N$ means (by postponement of $\redi$) that $M \redbcs^* P \redi^* N$ for some $P$ that \Cref{lem:redi} guarantees to be $\betac$-normal only, not $\sigma$-normal.
To prove that $M$ has a $\betac\sigma$-normal form is not even enough to take the $\sigma$-normal form of $P$, because a $\sigma$ step can create a $\betac$-redex. 
To solve the problem, we need the following technical lemma.
\begin{restatable}{lem}{lemNorm}
	\label{lem:id_tail} 
	Assume $M\redi^k N$, where $k> 0$, and $N$ is $\sigma\ii$-normal.
	If $M$ is not $\sigma$-normal, then there exist $M'$ and $N'$ such that either $M\reds M' \redi N'\redi^{k-1} N$  or $M\reds M' \redbc N'\redi^{k-1} N$.
\end{restatable}
We also use  that $\reds$ and $\redi$ are strongly normalizing (\Cref{lemma:strong-normalizing-idsigma}). 
Instead of proving that $\reds$ and $\redi$ are---separately---so, we state a more general result (its proof is in \Cref{app:normalization}).
\begin{restatable}[Termination of $\sigma\id$]{prop}{lemmastrongnormalizingidsigma}
	\label{lemma:strong-normalizing-idsigma}
	Reduction $\redx{\sigma\id} \, = (\reds \cup \redid)$ is strongly normalizing.
\end{restatable}

Now we have all the elements to prove the following.

\begin{restatable}[Irrelevance of $\ii$ for normalization]{thmm}{StateNormals}
	\label{thm:normals}
	The following are equivalent:
	\begin{enumerate}
		\item\label{p:normals-lc} $M$ is  $\lcc$-normalizing;
		\item\label{p:normals-bc} $M$ is  $\bs$-normalizing. 
	\end{enumerate}
\end{restatable}

\begin{proof}
    $ $
	\begin{description}
		\item[(1)$\Rightarrow$(2):] If $M$ is $\lcc$-normalizing, then $M \redlc^* N$ for some $\lcc$-normal $N$. 
		By postponement of $\ii$ steps (\Cref{thm:i_postponement}), for some $P$ we have
		\begin{align}\label{eq:postponement-iota}
		M \redbcs^*  P \redi^* N
		\end{align} 
		By   \Cref{lem:redi}.\ref{p:redi-beta-normal}, $P$ is $\betac$-normal in \eqref{eq:postponement-iota}. 
		
		For any sequence of the form \eqref{eq:postponement-iota}, let $w(P)=(w_{\ii}(P),w_{\sigma}(P))$, where $w_{\ii}(P)$ and $w_{\sigma}(P)$ are the lengths of the maximal $\ii$-sequence and of the maximal $\sigma$-sequence from $P$, respectively; 
		they are well-defined because $\redi$ and $\reds$ are strongly normalizing (\Cref{lemma:strong-normalizing-idsigma}).
		
		We proceed by induction on $w(P)$ ordered lexicographically to prove that $ M \redbcs^*  P' \redi^* N$ for some  $P'$ $\betac\sigma$-normal (and so $M$ is $\betac\sigma$-normalizing).
		
		\begin{itemize}
			\item If $w(P) = (0,h)$ then $P = N$, so $P$ is $\sigma$-normal and hence $\betac\sigma$-normal.
			\item If $w(P) = (k,0)$, then $P$ is $\sigma$-normal and hence $\betac\sigma$-normal.  
			\item Otherwise $w(P) = (k,h)$ with $k, h > 0$.
			By \Cref{lem:id_tail}, $M \redbcs^*  P' \redi^{*} N$ for some $P'$ with $w(P')<w(P)$: indeed,  $w(P')=(k,h-1)$ or $w(P')=(k-1,h)$. 
			By \ih, we can conclude.
		\end{itemize}
		
		\item[(2)$\Rightarrow$(1):] As $M$ is $\betac\sigma$-normalizing, $M \redbcs^* N$ for some $\betac\sigma$-normal $N$.
		As $\redi$ is strongly normalizing (\Cref{lemma:strong-normalizing-idsigma}), $N \redi^* P$ for some $P$ $\ii$-normal.
		By \Cref{lem:redi}.\ref{p:redi-beta-normal}-\ref{p:redi-sigma-normal}, $P$ is also $\betac$-normal and $\sigma$-normal.
		Summing up, $M \redlc^* P$ with $P$ $\lcc$-normal, \ie, $M$ is $\lcc$-normalizing.	
		\qedhere 
	\end{description}

\end{proof}

\subsection{The Essential Role of $\sigma$ Steps for Normalization}
\label{subsect:betasigma}

In  $\lc$, for normalization, $\sigma$ steps play a crucial role, unlike $\ii$ steps.
Indeed, $\sigma$ steps can unveil ``hidden'' $\betac$-redexes in a term.
Let us see this with an example, where we consider a term that is $\betac$-normal 
but  diverging in $\lc$ and this divergence is ``unblocked'' by a $\sigma$ step.

\begin{exmpl}[Normalization in $\lc$] 
	\label{ex:blocked_beta}
	Let $\Delta = \lambda x . x \oc x$.
	Consider the $\sigma$ step 	
	\[M_z=\Delta((\lam y. !\Delta)(z!z )) \reds (\lam y.\Delta !\Delta)(z!z) =N_z\]
	$M_z$ is $\betac$-normal, but not $\lcc$-normal. 
	In fact, $M_z$ is diverging in $\lc$ (that is, it is not $\lcc$-normalizing): 
	\[M_z\reds N_z \redbc N_z\redbc \dots \]
	Note that the $\sigma$ step is weak and that $N_z$ is normal for $\wredbc$ but not for $\sredbc$.
	
	The fact that a $\sigma$ step can unblock a hidden $\betac$-redex is not limited to open terms.
	Indeed, $\oc \lambda z. M_z$ is closed and $\betac$-normal, but divergent in $\lc$:
	\[\oc \lambda z. M_z\reds \oc \lambda z.N_z \redbc \oc \lambda z.N_z\redbc \dots \]
%
%
\end{exmpl}

\Cref{ex:blocked_beta} shows that, contrary to $\ii$ steps, $\sigma$ steps are essential to determine whether a term has or not a normal form in $\lc$.
This fact is in accordance with the semantics. 
First, it can be shown that the term $M_z$ above and $\Delta \oc  \Delta$ are observational equivalent. 
Second, the denotational models and type systems studied in \cite{Ehrhard12,deLiguoroTreglia20} (which are compatible with $\lc$) interpret $M_z$ in the same way as $\Delta \oc \Delta$, which is a $\betac$-divergent term.
It is then reasonable to expect that the two terms have the same operational behavior in $\lc$. 
Adding $\sigma$ steps to $\betac$-reduction is a way to obtain this: both $M_z$ and $\Delta \oc \Delta$ are divergent~in~$\lc$. 
Said differently, $\sigma$-reduction restricts the set of $\lcc$-normal forms, so as to exclude some $\betac$-normal (but not $\bs$-normal) forms that are semantically meaningless.

Actually, $\sigma$-reduction \emph{can only restrict} the set of terms having a normal form: it may turn a $\betac$-normal form into a term that diverges in $\lc$, but it cannot turn a $\betac$-diverging term into a  $\lc$-normalizing one.
To prove this  (\Cref{prop:from-betasigma-to-beta}), we rely on 
the following lemma.

\begin{lem}
	\label{lemma:postpone-sigma}
	If $M$ is not $\betac$-normal and $M \reds L$, then $L$ is not $\betac$-normal
	and $L  \redbc N$ implies $M \redbc \cdot \reds^= N$.
\end{lem}

Roughly, \Cref{lemma:postpone-sigma} says that a $\sigma$ step on a term that is not $\betac$-normal cannot erase a $\betac$-redex, and hence it can be postponed.
\Cref{lemma:postpone-sigma} does not contradict \Cref{ex:blocked_beta}: the former talks about a $\sigma$ step on a term that is not $\betac$-normal, whereas the start terms in \Cref{ex:blocked_beta} are $\betac$-normal.

{\begin{prop}
	\label{prop:from-betasigma-to-beta}
	If a term is $\bs$-normalizing (resp. strongly $\bs$-normalizing), then it is $\betac$-normalizing (resp. strongly $\betac$-normalizing).

\end{prop}}

\begin{proof}
	As $\redbc \,\subseteq\, \redbcs$, any infinite $\betac$-sequence is an infinite $\bs$-sequence. 
	So, if $M$ is not strongly $\betac$-normalizing, it is not strongly $\bs$-normalizing.
	
	We prove now the part of the statement about normalization. 
	If $M$ is $\bs$-normalizing, there exists a reduction sequence $\s : M \redbcs^* N$ with $N$ $\bs$-normal. 
	Let $|\s|_\sigma$ be the number of steps in $\s$, and let $|\s|_{\betac}$ be the number of $\betac$ steps after the last $\sigma$ step in $\s$ (when $|\s|_\sigma = 0$, $|\s|_{\betac}$ is just the length of $\s$). We prove by induction on $(|\s|_\sigma, |\s|_{\betac})$ ordered lexicographically that $M$ is $\betac$-normalizing.
	There are three cases. 
	\begin{enumerate}
		\item If $\s$ contains only $\betac$ steps ($|\s|_\sigma = 0$), then $M \redbc^* N$ and we are done.
		\item If $\s : M \redbcs^* L \reds^+ N$ ($\s$ ends with a non-empty sequence of $\sigma$ steps), then $L$ is $\betac$-normal by \Cref{lemma:postpone-sigma}, as $N$ is $\betac$-normal; by \ih applied to the sequence $\s' : M \redbcs^* L$ (as $|\s'|_\sigma < |\s|_\sigma$), $M$ is $\betac$-normalizing.
		
		
		\item Otherwise 
		$\s : M \redbcs^* L \reds P \redbc Q \redbc^* N$ ($L \reds P$ is the last $\sigma$ step in $\s$, followed by a $\betac$ step). 
		By \Cref{lemma:postpone-sigma}, either there is a sequence $\s' : M \redbcs^* L \redbc R \reds Q \redbc^* N$, then $|\s'|_\sigma = |\s|_\sigma$ and $|\s'|_{\betac} < |\s|_{\betac}$;
		or $s' : M \redbcs^* L \redbc Q \redbc^* N$ and then $|\s'|_\sigma < |\s|_\sigma$.
		In both cases $(|\s'|_\sigma, |\s'|_{\betac}) < (|\s|_\sigma, |\s|_{\betac})$, so by \ih~$M$ is $\betac$-normalizing.
		\qedhere
	\end{enumerate}	
\end{proof}

%
%

\subsection{Normalizing Strategies}
\label{subsect:normalizing-strategies}

Irrelevance of $\ii$ steps (\Cref{thm:normals}) implies that 
to  define a normalizing strategy for $\lc$, it suffices to define a normalizing strategy for  $\bs$.
We do so  by iterating either \emph{surface} or \emph{weak} reduction. Our definition of $\bs$-normalizing strategy and the proof  of normalization (\Cref{thm:normalizing_strategy}) is \emph{parametric} on either. 

The difficulty here is that both   weak and   surface reduction are \emph{non-deterministic.}
The key property we need in  the  proof   is that the reduction  we iterate  is \emph{uniformly normalizing} (see \Cref{def:normalizing_terms}). 
We first establish that this holds for weak and surface reduction.  While uniform normalization  is easy to prove for  the former, it is  \emph{non-trivial} for  the latter, its proof is rather sophisticated.
Here we reap the fruits of   the careful analysis of the number of $\betac$ steps
 in \Cref{sec:w_factorization}.
Finally, we formalize the strategies and tackle normalization.

\begin{notation*}Since we are now only concerned with $\bs$ steps,
for  the sake of readability 
in the rest of the section, we often  write $\red$, $\sred$ and $\wred$ for $\redx{\bs}$,  $\xredx{\surf}{\bs}$ and  $\xredx{\weak}{\bs}$, respectively.
\end{notation*}

\paragraph{Understanding Uniform Normalization.} The fact that $\sred$ and $\wred$ are uniformly normalizing  is key in the definition of normalizing strategy and deserves some discussion. 
	
	The heart of the normalization proof is that if $M$ has a $\red$-normal form $N$, we can perform surface steps and reach  a \emph{surface normal form}.  
	Note that {surface factorization} only guarantees that \emph{there exists}  a $\sredx{}$-sequence such that if 
	$M\redx{}^*N$ then $  M\sredx{}^*  L \nsredx{}^*  N $, where $L$ is $\sredx{}$-normal.
	The \emph{existential} quantification is crucial here because
	$\sredx{}$ is \emph{not} a  deterministic reduction.
	\emph{Uniform normalization} of $\sredx{}$ transforms the existential into a \emph{universal} quantification:
	if $M$ has a $\red$-normal form (and so a fortiori a surface normal form), then \emph{every} sequence of $\sredx{}$ steps will terminate.
	The normalizing strategy  then  iterates this process, performing surface reduction on the subterms of a surface normal form, until we obtain a $\red$-normal form.

\subsubsection{Uniform Normalization of Weak and Surface Reduction}\label{sec:uniform}
\newcommand{\bsteps}[1]{\mathtt{b}(#1)}

We prove that both weak and surface reduction are uniformly normalizing, \ie for $\ex\in \{\weak,\surf\}$,  if a term $M$ is $\ered$-normalizing, then it is
strongly  $\ered$-normalizing.
In both cases, 
the proof relies on the fact that all maximal $\ered$-sequences from a given term $M$ have the same number of $\betac$ steps.
%

\begin{fact}[Number of $\betac$ steps]
	\label{fact:b_steps}
	Given a  $\redx{\bs}$-sequence $\s$, the number    of its $\betac$ steps
	 is finite if and only if $\s$ is finite. 
\end{fact}
\begin{proof}
	The right-to-left implication is obvious. 
	The left-to-right is an immediate consequence of the fact that $\reds $ is strongly normalizing (\Cref{lemma:strong-normalizing-idsigma}).
\end{proof}

A  \emph{maximal} $\ered$-sequence from $M$ is either infinite, or ends in a  $\ex$-normal form. 
\Cref{prop:uniform} states that  for $\ex\in \{\weak,\surf\}$, all   \emph{maximal} $\ered$-sequences from the same term $M$ have the \emph{same behavior}, also quantitatively (with respect to the number of $\betac$ steps).
The proof relies on the following lemma.
Recall that weak reduction is not confluent (\Cref{ex:weak}); however,   $\wredbc$ is deterministic.

%

\begin{lem}[Invariant]\label{lem:weak}
Given $M \in \ComTerm$,  \emph{every sequence}  $ M\wredxStar{\bs} N$ where    $N$  is  $\wredbc$-normal
has   \emph{the same number $k$} of $\betac$ steps.
Moreover,
	\begin{enumerate}
	\item  the unique maximal $\wredbc$-sequence from $M$ has length $k$, and 
	\item  there exists a sequence   $M\wredbc^{\!k} \, L \, \wredxStar{\sigma} N$ for some $L \in \ComTerm$.   
\end{enumerate}
\end{lem}

\begin{proof}The argument is illustrated in \Cref{fig:weak}. 
	Let $k$ be the number of $\betac$ steps in a sequence  $\s \colon M\wredxStar{\bs} N$ where $N$ is  $\wredbc$-normal.
	By weak factorization (\Cref{prop:weak_fact}.\ref{p:weak_fact-from-weak}) there is a sequence   $M\wredbc^{\!k} L \wredxStar{\sigma} N$ with the same number  $k$ of $\betac$ steps. 
	As $N$ is  $\wredbc$-normal, so is $L$ 
	(\Cref{lem:wbeta_nf}). 
	Thus, $M\wredbc^k L$  is a maximal  $ \wredbc $-sequence from $M$, and it is unique because $\wredbc$ is deterministic. 
\end{proof}

\begin{thmm}[Uniform normalization]
	\label{prop:uniform}\hfill
	\begin{enumerate}
		\item\label{p:uniform-weak} Reduction $\wredx{\bs}$ 
		is  uniformly normalizing.
		\item\label{p:uniform-surface} Reduction $\sredx{\bs}$ 
		is  uniformly normalizing.
	\end{enumerate}
Moreover,   all maximal  $\wredx{\bs}$-sequences (resp. all maximal  $\sredx{\bs}$-sequences) from the same term  $M$  have  the same number of $\betac$ steps.
\end{thmm}
\begin{proof}We write $\red$ (resp.  $ \wred,\sred $) for $\redx{\bs}$ (resp.  $ \wredx{\bs},\sredx{\bs} $).

	%
	%
	%
	%

	\paragraph{Claim 1.}

	
	

	Let  $M \wred^* N$ where $N$ is $\wred$-normal, and so, in particular $\wredbc$-normal. 
	By \Cref{lem:weak},
	$M\wredbc^{\!k} L \wredxStar{\sigma} N$ where $M\wredbc^{\!k} L$ is the (unique) maximal $ \wredbc $-sequence from $M$.
	We prove  that no $\wred$-sequence from	$M$ may have more than $k$ $\betac$ steps. 
Indeed, every  sequence 
  	$\s\colon M\wred^*N'$  can be factorized 
	(\Cref{prop:weak_fact}.\ref{p:weak_fact-from-weak})  as   $M\wredxStar{\betac} L' \wredxStar{\sigma} N' $ with the same number of $\betac$ steps  as $\s$, and $M\wredxStar{\betac} L' $ is a prefix of the maximal $\wredbc$-sequence $M\wredbc^k L$ from $M$ (since $\wredbc$ is deterministic). 	
	
	We deduce that	  no infinite   $\wred$-sequence from $M$  is possible (by \Cref{fact:b_steps}).

%
%
	
	\begin{figure}
	\begin{minipage}[c]{.40\textwidth}

		\begin{diagram}[tight,height=1.4em,width=2em]
			&                       &       &         & N\\
			&                       &       & \ruOnto^{\weak}_{\sigma}   & \\
			M & \rOnto^{\weak}_{\betac} & \cdot\\
			&                       &       & \rdOnto^{\weak}_{\sigma}   & \\
			&                       &       &         & N'\\
		\end{diagram}
		
		\caption{Weak reduction}\label{fig:weak}
	\end{minipage}
	\hspace{5mm}%
	\begin{minipage}[c]{.55\textwidth}
	
		\begin{diagram}[tight,width=2.8em,height=2em]
			&                        &       &                            & S_1                           \\
			&                        &       & \ruOnto^{\weak}_{\sigma}   & \dOnto_{\sigma} & \rdOnto      \\
			M   & \rOnto^{\weak}_{\betac}& L &                            & S_3               & \rOnto      & N\\
			&                        &       & \rdOnto^{\weak}_{\sigma}   & \uOnto_{\sigma} & \ruOnto      \\
			&                        &       &                            & S_2                            \\
		\end{diagram}
		
		\caption{Surface reduction }\label{fig:surface}
	\end{minipage}

	\end{figure}


	\paragraph{Claim 2.}
	Assume that 
	$M \sred^* N$ with $N$ $\sred$-normal.
	Recall that $\sred$ is confluent (\Cref{prop:surface-properties}.\ref{p:surface-properties-betasigma}),  so $N$ is the  unique $\sred$-normal form of $M$.

	First, by induction on $N$, we prove  that   given a term $M$, 
	\begin{center}
		\begin{tabular}{l}
			$(\#)$		all sequences $M\sred^* N$  have the same number of $\betac$ steps.
		\end{tabular}
	\end{center} 
	Let $\s_1,\s_2$ be two such sequences.  \Cref{fig:surface} illustrates the argument.
	By weak factorization (\Cref{prop:weak_fact}.\ref{p:weak_fact-from-surface}), there is a   sequence $M\wred^* S_1 \nwredxStar{} N$ (resp. $M\wred^* S_2\nwredxStar{} N$)  with  the same number of $\betac$ steps as 
	$\s_1$ (resp. $\s_2$), and whose steps are all surface. Note that $S_1$ and $S_2$ are $\wred$-normal (by \Cref{lem:surface_nf}, because $N$ is in particular $\wred$-normal), and so in particular $\wredbc$-normal. 
	By \Cref{lem:weak},  $M\wred^* S_1$ ,  $M\wred^* S_2$ have the same number $k$ of $\betac$ steps, and so do  the sequences  
  $\s_1':M \wredx{\betac}^k\ L\wredxStar{\sigma} S_1$ and $\s_2': M\wredx{\betac}^k\ L \wredxStar{\sigma} S_2$.

	To prove $(\#)$, we show that  the sequences  $\s_1''\colon S_1 \nwred^* N$ and $\s_2''\colon S_2 \nwred^* N$ have the same number of $\betac$ steps.
	
	By confluence of $\sreds$ (\Cref{prop:surface-properties}.\ref{p:surface-properties-separate}), {$S_1\  \sredxStar{\sigma}  \  S_3\  \xrevredxStar{\surf}{\sigma}\  S_2$}, for some $S_3$, and  (by confluence of $\sred$, \Cref{prop:surface-properties}.\ref{p:surface-properties-betasigma})
	there is a sequence $\mathfrak{t} \colon S_3\sredxStar{} N$. 
	By \Cref{lem:wbeta_nf}, since $S_1,S_2$ are $\wred$-normal, terms  in these sequences are $\wred$-normal, and so all steps are not only surface, but also  $\nwred$ steps. 
	That is, $S_1\nwredxStar{\sigma} S_3$, $S_2\nwredxStar{\sigma} S_3$ and $\mathfrak{t}\colon S_3\nwredxStar{} N$. 
	Hence $S_1, S_2, S_3, N$ have the same shape by \Cref{fact:shape}.  
	
We examine  the shape of $N$, and  prove  claim $(\#)$ by showing  that  $\s_1''$ and $\s_2''$ have the same number of $\betac$ steps as $\mathfrak{t}$ (note that the sequences $S_1\nwredxStar{\sigma} S_3$ and $S_2\nwredxStar{\sigma} S_3$ have no $\betac$ steps).

	\begin{itemize}
		\item $N= \oc V$. In this case, $N=S_1=S_2$, and the claim (\#) is immediate.
		\item $N=(\lam x.P)Q$, and  $S_i= (\lam x.P_i)Q_i$ (for $i\in \{1,2,3\}$). We have $P_i\sred^* P$ and $Q_i\sred^*Q$.
		Since $P$ and $Q$ are $\sred$-normal,   by   \ih we have:
		\begin{itemize}
			\item  the two sequences    $P_1 \sredxStar{} P$ and    $P_1 \sredxStar{\sigma} P_3 \sredxStar{} P$ have the same number of $\betac$ steps, and similarly $Q_1 \sredxStar{} Q$ and $Q_1 \sredxStar{\sigma} Q_3 \sredxStar{} Q$. 
			Hence $\s_1''$ and 
			$\mathfrak{t}$ have the same number of $\betac$ steps.
			\item Similarly, $\s_2''$ and 
			$\mathfrak{t}$ have the same number of $\betac$ steps.
		\end{itemize}
	\end{itemize}
	
	This completes the proof of $(\#)$.
	We now can conclude that $\sred$ is \emph{uniformly normalizing}.
	If the term $M$ has a sequence $\s \colon M\sred^* N$ where $N$ is $\sred$-normal, then  no $\sred$-sequence  can have more $\betac$ steps than $\s$, because  given any sequence  $M\sred^* T$ then (by confluence of $\sred$) $T\sred^* N$, and (by $ \# $) $M\sred^* T\sred^* N$ has the same number of  $\betac$ steps as $\s$. Hence, all $\sred$-sequences from $M$ are finite.
\end{proof}



\subsubsection{Normalizing strategies }\label{sec:normalizing_strategy}
We are ready to define and deal with normalizing strategies for $\lc$.  Our definition is inspired, and generalizes,  the stratified  strategy proposed in \cite{Guerrieri15,GuerrieriPaoliniRonchi17}, which iterates  weak  reduction (there called head reduction) according to a more strict discipline.

%

\paragraph*{Iterated $\ex$-Reduction.} We define a family of normalizing strategies, parametrically on the reduction to iterate, which can be surface or weak.
Let $\ex\in \{\weak, \surf\}$.
Reduction   $\lered$ is defined 
as follows, by iterating $\ered$ in the left-to-right order (\Cref{thm:normalizing_strategy} then shows that $\lered$ is a normalizing strategy). 
\begin{enumerate}
	\item If $M$  is \emph{not} $\ered$-normal: 
	\begin{equation*} \infer{M \,\lered\, M'}{M \,\ered\, M'} \end{equation*}

	\item If $M$ \emph{is} $\ered$-normal (below, ``$V \, \bs$-normal'' means $V \!\in\! \Var$ or $V \!=\! \lambda x.L$ with $L$ $\bs$-normal):
	\begin{equation*}
	 \infer{M \defeq \oc (\lam x.N) \,\lered\, \oc (\lam x.N')}{N \,\lered\, N'} \qquad
	\infer{M \defeq (\lam x.N) L \,\lered\, (\lam x.N') L }{N \,\lered\, N'} \qquad
	\infer{M \defeq VN \,\lered \, VN'}{V \text{ $\bs$-normal} &  N \,\lered\, N'}  
	\end{equation*}  
\end{enumerate}

\SLV{}{
\begin{remark} Here we choose to iterate $\ex$-reduction in  left-to-right order. This choice is arbitrary, and only for the sake of  presentation. 
Assume $M$ is  $\ered$-normal. $M$ has shape $V_1(...V_k(!U))$. For simplicity, we arbitrarily choose to reduce to normal form first $V_1$, then $V_2$, and so on till $!U$. However,  the choice of the order here is irrelevant: we could proceed right-to-left, or reduce $V_1, V_2,...!U$ in parallel . This because any step $M \redx{\bs} M'$ is    a $\nered$ step, and so it  preserves  the shape of $M$, and the absence of $\ered$ redexes. 
\end{remark}
}

\begin{thmm}[Normalization for $\bs$]\label{thm:normalizing_strategy}
	 Assume $M \RedStar_{\bs} N$ where $N$ is   $\redx{\bs}$-normal. 
	 Let $\ex\in \{\weak, \surf\}$. 
	 Then, every maximal $\leredx{\bs}$-sequence from $M$ ends in $N$.
\end{thmm}
\begin{proof}
	By induction on the term $N$. 	
	Let $\s=M, M_1, M_2, \dots$ be a maximal $\lered$-sequence from $M$. 
		
	We write   $\red$, $\ered$, $\nered$ and $\lered$ for $\red_{\bs}$, $\eredx{\bs}$, $\neredx{\bs}$ and $\leredx{\bs}$, respectively.	
	 We observe that 
	 \begin{center}
	 		(**) every maximal $\ered$-sequence from $M$ is finite.
	 \end{center}
	   Indeed, from  
	 $M\red^*N$, by $\ex$-factorization, we have that 
		$M\ered\, L   \nered^* N$. Since $N$ is $\ered$-normal, so is $L$   (by \Cref{lem:surface_nf}) and  (**) follows by
		 uniform normalization of $\ered$ (\Cref{prop:uniform}).

 	Let $\s' \sqsubseteq \s$ be the maximal prefix of  $\s$ such that $M_i\ered M_{i+1}$.  Since it is  finite,  $\s'$ is $ M, \dots, M_{k}$, where   $M_k$ is $\ex$-normal. Let $\s'' = M_k, M_{k+1} \dots$  be the sequence such that  $\s=\s'\s''$.
	
	Note that  all terms in $\s''$ are $\ex$-normal (by repeatedly using  \Cref{lem:surface_nf} from $M_k$), hence $M_k\nered M_{k+1} \nered \dots $, and (by  shape preservation, \Cref{fact:shape}) all terms in $\s''$ have the same shape as $M_k$.
	
 	By confluence of $\red$ (\Cref{prop:conf}.\ref{p:conf-betasigma}), $M_k \red^* N$.  Again, all terms in this sequence are $\ex$-normal, by  repeatedly using  \Cref{lem:surface_nf} from $M_k$. 
	So, $M_k \nered^* N $, and (by shape preservation, \Cref{fact:shape}) $M_k$ and $N$ have the same shape.\\

 	We have established that $M_k$ and all terms in $\s'':M_k,  M_{k+1}, \dots$ have the same shape as $N$.
Now we examine the possible cases for  $N$.
	\begin{itemize}
		\item $N= \oc x$, and   $M_k= \oc x$.  Trivially $M\lered^* M_k =  N$.
		
		\item $N= \oc(\lam x.N_P)$ and   $M_{k}= \oc (\lam x. P)$ with  $P \red^*  N_P$.
		 Since   $ N_P$ is $\bs$-normal, by \ih  every maximal 
		$\lered$-sequence from $P$ terminates in $ N_P$, and so every maximal 	$\lered$-sequence from $!(\lam x. P)$ terminates in $!(\lam x.N_P)=N$.
		 Since the  sequence $\s''= M_k,  M_{k+1}, \dots $ is a maximal $\lered$-sequence,  we have that  $\s=\s'\s''$ is as follows
		\[		M\lered^* M_k= \oc(\lam x. P)\lered^* \oc (\lam x.N_P)=N.\]
		
		\item $N=(\lam x. N_P)N_Q$ and 
		 $M_k=(\lam x.P)Q$, with   {$P\red^* N_P$} and $Q\red^* N_Q$. 	
		Since  $N_P$ and $N_Q$ are both  $\bs$-normal, by \ih:
		\begin{itemize}
			\item every maximal $\lered$-sequence from $P$ ends in $N_P$. So every  $\lered$-sequence from $(\lam x.P)Q$ eventually reaches $(\lam x.N_P)Q$;  
			\item every maximal $\lered$-sequence from $Q$ ends in $N_Q$. So every $\lered$-sequence from $(\lam x.N_P)Q$  eventually reaches $(\lam x.N_P)N_Q=N$.
		\end{itemize}
		Therefore  $\s$ is as follows
		\[M\lered^* M_k=(\lam x.P)Q \lered^* (\lam x.N_P)Q \lered^* (\lam x.N_P)N_Q=N.\]	
		
		\item $N=x N_Q$ and $M_k=xQ$. Similar to the previous one.
		\qedhere
	\end{itemize}
\end{proof}

From a normalizing strategy for ${\bs}$ (\Cref{thm:normalizing_strategy}), we derive a normalizing strategy in~$\lc$.

\begin{coroll}[Normalization for $\lc$]
	 Let $\ex\in \{\weak, \surf\}$. If $M$ is $\lcc$-normalizing, then any maximal $\leredx{\bs}$-sequence from $M$ followed by any maximal $\redi$-sequence ends in the $\lcc$-normal form of $M$.
\end{coroll}

\begin{proof}
	By \Cref{thm:normalizing_strategy}, every maximal $\leredx{\bs}$-sequence from $M$ ends in a $\redx{\bs}$-normal form $L$. 
	Since $\redi \subseteq \redx{\sigma\id}$ is strongly normalizing (\Cref{lemma:strong-normalizing-idsigma}), every maximal $\redi$-sequence from $L$ ends in a $\redi$-normal form $N$, which is also $\redx{\bs}$-normal by \Cref{lem:redi}.
	Therefore, $N$ is $\lcc$-normal and this is the unique $\lcc$-normal form of $M$ since $\redlc$ is confluent (\Cref{prop:conf}.\ref{p:conf-betasigmaid}).
\end{proof}

%% file: Discussion_and_Related.tex

%
%
%
\section{Conclusions and Related Work}
\label{sec:related}

\subsection{Discussion: Reduction and Evaluation}

In computational calculi it is standard practice to define evaluation as \emph{weak reduction}, aka \emph{sequencing} 
\cite{Filinski:phd1996,JonesSLT98,LevyPT03, LagoGL17}.
Despite  the prominent role that weak reduction has in the literature, in particular for calculi with effects, what one discovers when analyzing the rewriting properties is somehow unexpected.
As  we  observe in  \Cref{sec:operational},  where we consider both the computational core  $\lc$ \cite{deLiguoroTreglia20},  and  a widely recognized reference   such as the calculus  $\lambda_{ml^*}$ by Sabry and Wadler \cite{SabryWadler97} (in turn inspired by Moggi \cite{Moggi'88,Moggi'89,Moggi'91}),
while full reduction 
is confluent, the closure of the rules under \emph{evaluation contexts}  turns out to be \emph{non-deterministic}, \emph{non-confluent}, and its \emph{normal forms} are \emph{not unique}. 
The issues   come from the monadic rules of  \emph{identity} and \emph{associativity}, hence they  are common to  \emph{all} computational calculi. 


\paragraph{A Bridge between Evaluation  and  Reduction.}

On the one hand, computational $\lambda$-calculi 
 have an unrestricted \emph{non-deterministic reduction} that generates the equational theory of the calculus, studied  for foundational and semantic purposes.
On the other hand, \emph{weak reduction} models evaluation in an  ideal programming language. 
It is then natural to wonder what is the relation between reduction and evaluation.
This is the first contribution of this paper. We establish a bridge between evaluation  and  reduction via a factorization theorem stating that every reduction can be rearranged so as to bring forward  weak reduction steps.

We focused on the rewriting theory of a specific computational calculus, namely the computational core $\lc$ \cite{deLiguoroTreglia20}.
We expect that our results and approach can be adapted also to other computational calculi such as $\lambda_{\mlsym}$. This demands further investigations.   Transferring the results is not immediate because the correspondence between the two calculi is not direct with respect to the \emph{rewriting} (see \Cref{rem:intricacies}).

\subsection{Technical Contributions} We studied the 
rewriting theory of the  computational core $\lc$ introduced in \cite{deLiguoroTreglia20}, a variant of Moggi's $\lambda_c$-calculus \cite{Moggi'88}, focusing on two questions: 
	\begin{itemize}
		\item how to reach   \emph{values}? 
		
		\item how to reach  \emph{normal forms}?
	\end{itemize}
For the first point, we show that weak $\betac$-reduction is enough (\Cref{sec:values}). For the second question, we define a family of normalizing strategies  (\Cref{sec:normalization}).

We have  faced  the issues caused by identity and associativity rules (which internalize the monadic rules in the syntax), and dealt with them  by means of factorization techniques. 


We have  investigated  in depth  the structure of normalizing~reductions, and we assessed 
the role of the  $\sigma$-rule (aka  associativity) as computational and not merely structural.
We found out that it plays at least three distinct,  independent roles in $\lc$:
\begin{itemize}
	\item $\sigma$ unblocks ``premature'' $\betac$-normal forms so as to guarantee that there are not $\lcc$-normalizing 
	terms whose semantics is the same as diverging terms, as we have seen in \Cref{subsect:betasigma};
	\item it internalizes the associativity of Kleisli composition into the calculus, as a syntactic reduction rule, as explained in \Cref{sect:intro} after \Cref{eq:sigma};
	\item it ``simulates'' the contextual closure of the $\betac$-rule for terms that reduce to a value, as we have seen in \Cref{thm:return-value}.
\end{itemize}

\subsection{Related Work}

\paragraph{Relation with Moggi's Calculus.}

Since our focus is on operational properties and reduction theory, we chose the computational core $\lc$ \cite{deLiguoroTreglia20} among the different variants of computational calculi in the literature inspired by Moggi's seminal work \cite{Moggi'88,Moggi'89,Moggi'91}. 
Indeed, the computational core $\lc$ has a ``minimal'' syntax that internalizes Moggi’s original idea of deriving a calculus from the categorical model consisting of the Kleisli category of a (strong) monad. For instance, $\lc$ does not have to consider both a pure and a (potentially) effectful functional application. So, $\lc$ has less syntactic constructors and less reductions rules with respect to other computational calculi, and this simplifies our operational study. 

Let us discuss the difference between $\lc$ and Moggi's $\lambda_c$. As observed in \Cref{sect:intro,sec:compcal}, the first formulation of $\lambda_c$ and of its reduction relation
was introduced in \cite{Moggi'88}, where it is formalized by using {\em let}-constructor. 
Indeed, this operator is not just a syntactical sugar for the application of $\lambda$-abstraction. In fact, it represents the extension to computations of functions from values to computations, therefore interpreting Kleisli composition. 
Combining {\em let} with ordinary abstraction and application is at the origin of the complexity of the reduction rules in \cite{Moggi'88}. On the other hand, this allows extensionality to be internalized.
Adding the $\eta$-rule to $\lc$ breaks confluence, as shown in \cite{deLiguoroTreglia20}. 

Besides using {\em let} or not, a major difference of $\lc$ with respect to $\lambda_c$ is the neat distinction among the two syntactical sorts of terms, restricting the
combination of values and non-values since the very definition of the grammar of the language. 
In spite of these differences, in \cite[\S 9]{dLT-19} it has been proved that there exists an interpretation of
$\lambda_c$ into $\lc$ that preserves the reduction, while there is a reverse translation that preserves convertibility, only.

\paragraph{Other Related Work.} 
Sabry and Wadler \cite{SabryWadler97} is the first work on the computational calculus to 
put on center stage  the reduction. Still the focus of the paper are the properties of the translation between that and the monadic metalanguage---the reduction theory itself is not investigated.

In \cite{HZ09} a different  refinement  of $\lambda_c$ has been proposed. 
Its reduction rules are  divided into a  purely operational, a  structural and an  observational system.
It is proved that the purely operational system suffices to reduce any closed term to a value. This result is similar  to our  \Cref{thm:values}, with 
weak $\beta_c$ steps corresponding  to head reduction in \cite{HZ09}. Interestingly, the analogous of our rule $\sigma$ is part of the structural system, while the rule corresponding
to our $\id$ is generalized and considered as an observational rule. Unlike our work, normalization is not studied in 
 \cite{HZ09}.  

Surface reduction  is a generalization of weak reduction that comes from linear logic. 
We inherit surface factorization from the linear $\lambda$-calculus in \cite{Simpson05}.
Such a reduction has been recently studied in several variants of 
the $\lambda$-calculus, especially for semantic purposes \cite{AccattoliPaolini12,CarraroGuerrieri14,AccattoliGuerrieri16,EhrhardGuerrieri16,GuerrieriManzonetto18,Guerrieri18,BucciarelliKRV20,GuerrieriO21}.

Regarding the $\sigma$-rule, 
in \cite{CarraroGuerrieri14} two commutation rules are added to Plotkin's CbV $\lambda$-calculus in order to remove meaningless normal forms---the resulting calculus is called \emph{shuffling}.  The commutative rule there called $\sigma_3$ is literally the same as $\sigma$ here. 
In the   setting of the \emph{shuffling calculus},   properties such as  the fact that   all maximal  surface $\betav\sigma$-reduction sequences  from the same term $M$   have the same number of $\betav$ steps, and so such a reduction is uniformly normalizing,  were  known via semantic tools  \cite{CarraroGuerrieri14, Guerrieri18}, namely non-idempotent intersection types. 
In this paper we   give the first syntactic proof of such a  result.

A relation between the  computational calculus, \cite{Simpson05} and other linear calculi are well-known in the literature, see for example \cite{EggerMS09, SabryWadler97, MaraistOTW99}.

%

In \cite{deLiguoroTreglia20}, Theorem 8.4 states  that any closed term returns a value if and only if it is convergent according to a big-step operational semantics. 
That proof is incomplete and needs a more complex argument via  factorization, as  we do here to prove \Cref{thm:values} (from which  that  statement in  \cite{deLiguoroTreglia20}  easily follows).

 \paragraph{Acknowledgements.}
We are in debt with Vincent van Oostrom, for the  technical issues he pinpointed, and for his many insightful technical suggestions.
We also thank the referees for their valuable comments.

This work was  partially supported  by the ANR project PPS: ANR-19-CE48-0014

%% file: Appendix.tex

\section{ General properties of the contextual closure }

\paragraph{Shape Preservation.}
We start by recalling a basic but key  property of contextual closure. If a step $\redc$ is obtained by closure under \emph{non-empty context} of a rule $\rredc$, then it preserves the shape of the term.
{We say that $T$ and $T'$ have \emph{the same shape} if  both terms are an 
	application (resp. an abstraction, an variable, a term of shape $!P$).}

\shape*

Note that a root  step  $\Root{}$ is both a  \emph{weak} and a \emph{surface} step.

The implication in the previous lemma cannot be reversed as the following example shows:

\[M\eq V(\II P)\redx{\ii} VP\eq N\]
$M$ is a $\sigma$-redex, but $N$ is not.

\paragraph{Substitutivity.}
A relation $\hookrightarrow$ 
on terms is \emph{substitutive} if 
\begin{equation}\tag{\textbf{substitutive}}
	R \hookrightarrow   R' 
	\text{ implies } R \subs x Q \hookrightarrow R'\subs x Q.
\end{equation}
An obvious induction  on the shape of terms shows the   following (\cite{Barendregt'84} p. 54).
	\begin{property}[Substitutive]\label{fact:subs} Let $\redc$ be  the contextual closure of $\rredc$.
		\begin{enumerate}
			\item\label{fact:subs-function} If $\rredc $ is substitutive then $\redc$ is 
			substitutive: ~ $T\redc T'$ implies $T \subs{x}{Q} \redc T' \subs{x}{Q}$.
			\item\label{fact:subs-argument} If $Q\redc Q'$ then $T\subs{x}{Q} \redc^* 
			T\subs{x}{Q'} $.
		\end{enumerate}
\end{property}

\section{Properties of the syntax $\Lambda^!$}\label{app:bang}

In this section, we  consider the set of terms $\Lambda^!$ (the same syntax as the \emph{full} bang calculus, as defined in \Cref{subsect:translations-comp-bang}), endowed with a generic reduction $\redx{\Rule}$ (from a generic rule $\mapsto_\Rule$).
We study some properties that hold in general in $(\Lambda^!, \redx{\Rule})$. 

\emph{Terms}    are generated by the grammar:
\begin{align*}
	T,S,R&::=  x  \mid 	ST  \mid 	\lambda x.T    \mid \oc T & (\textbf{terms } \Bang)
\end{align*}

\emph{Contexts} ($\cc$), \emph{surface contexts} ($\ss$) and  \emph{weak contexts} ($\ww$)  are generated  by the grammars:

\begin{align*}
	\cc & ::= \hole{}  \mid T\cc \mid \cc T   \mid \lambda x.\cc    \mid  \oc\cc  & \qquad(\textbf{contexts})\\
	\ss & ::= \hole{} \mid T\ss \mid \ss T   \mid \lambda x.\ss &\qquad (\textbf{surface contexts})\\
	\ww & ::= \hole{} \mid T\ww \mid \ww T   \mid \oc \ww &\qquad (\textbf{weak contexts})
\end{align*}

If  $\Root{\Rule}$ is a rule, the reduction $\redx{\Rule}$ its the  closure under context
 $\cc$.
\emph{Surface   reduction}  $\sredx{\Rule}$ (resp.~\emph{weak  reduction} $\wredx{\Rule}$)
  is the closure of $\Root{\Rule}$ under surface contexts $\ss$ (resp.~weak contexts $\ww$).
 \emph{Non-surface reduction}  $\nsredx{\Rule}$ (resp.~\emph{non-weak reduction} $\nwredx{\Rule}$)
 is the closure of $\Root{\Rule}$ under contexts $\cc$ that are not surface (resp.~not weak).

\subsection{Shape preservation for internal steps in $\Bang$.}\label{sec:preservation}
\Cref{fact:shape} (p.~\pageref{fact:shape}) implies  that $\nsredx{\Rule}$ and $ \nwredx{\Rule}$  steps always preserve the shape of terms.
We recall that we  write $\Root{\Rule}$ to indicate the step $\redx{\Rule}$ obtained by \emph{empty contextual closure}.
The following property immediately follows from \Cref{fact:shape}.

%
%
%
%
%
%
%

\begin{fact}[Internal Steps]\label{fact:isteps}
	Let $\Root{\Rule}$ be a rule and $\redx{\Rule}$ be its contextual closure.
	The following hold for $\iredx{\,\Rule} \in \{\nsredx{\Rule},\nwredx{\Rule}\}$.
	\begin{enumerate}
		\item\label{p:preserve-shape} Reduction $\iredx{\Rule}$ preserves the shapes of terms. 
		\item\label{p:variable} There is no $T$ such that $T \iredx{\,\Rule} x$, for any variable $x$.

		\item\label{p:bang}   $T\iredx{\,\Rule} \, \oc  U_1$ implies $T = \oc T_1$ and $T_1\redx{\Rule} U_1$. 
		
		\item\label{p:abstraction} $T\iredx{\,\Rule} \lam x. U_1$ implies $T = \lam x. T_1$ and $T_1\redx{\,\Rule} U_1$.
		\item\label{p:application} $T\iredx{\,\Rule} U_1U_2$ implies  $T = T_1T_2$, with either (i) $T_1\iredx{\,\Rule} U_1$ (and $T_2=U_2$), 
		or (ii) $T_2\iredx{\,\Rule} U_2$ (and $T_1=U_1$).  Moreover, $T_1$ and $U_1$ have the same shape, and so   $T_2$ and $U_2$.
		
	\end{enumerate}
\end{fact}

\begin{coroll}\label{cor:redex}
	Let $\Root{\Rule}$ be a rule and $\redx{\Rule}$ be its contextual closure.
	Assume $T\nsredx{\Rule}\, S$ or $T\nwredx{\Rule}\, S$.
	\begin{itemize}
		\item $T$ is a $\bbeta$-redex if and only if $S$ is.
		\item $T$ is a $\sigma$-redex if and only if $S$ is. 
	\end{itemize}
\end{coroll}

\begin{proof}
	The left-to-right direction follows from \Cref{fact:isteps}.\ref{p:preserve-shape}.
	The right-to-left direction is obtained by repetitively applying \Cref{fact:isteps}.\ref{p:bang}--\ref{p:application}.
\end{proof}

%

\subsection{Surface Factorization, Modularly.}\label{app:factorization}

%
%
%

\label{sec:modular}

In an abstract setting, let us consider a   rewrite system $(A,\red)$ where $\red \eq\reda \cup \redc$. 
Under which condition $\red$ admits factorization, assuming that both $\reda$ and $\redc$ do?
That is, if $\reda \eq \ereda \cup \ireda$ and $\redc \eq \eredc \cup \iredc$ $\ex$-factorize (\ie $\reda^* \subseteq \ereda^* \cup \ireda^*$ and $\reda^* \subseteq \ereda^* \cup \ireda^*$), is it the case that $\red^* \subseteq \ered^* \cup \ired^*$ (where $\ered \defeq \ereda \cup \eredc$ and $\ired \defeq \ireda \cup \iredc$)?
To deal with this question,  a technique   for proving factorization  for \emph{compound systems}  in a 
\emph{modular} way  has been introduced in \cite{AccattoliFaggianGuerrieri21}.
The approach  can be seen as an  analog---for factorization---of the classical 
technique for confluence based on Hindley--Rosen lemma: 
if   $\reda,\redc$ are $\ex$-factorizing reductions, their union  $\reda \cup \redc$ also is, 
provided that two \textit{local} conditions of commutation hold.

\begin{thmm}[Modular factorization, abstractly \cite{AccattoliFaggianGuerrieri21}]\label{thm:modular}
	Let  $\reda \eq (\ereda \cup \ireda)$ and $\redc \eq (\eredc \cup \iredc)$ be  $\ex$-factorizing reductions.
	Let $\ered  \deff \ereda \cup\eredc$, and $\ired  \deff  \ireda \cup\iredc$. 
	The  union 	$\reda\cup \redc$  satisfies factorization 
	$\F{\ered}{\ired}$  if the following  swaps hold
	\begin{equation}\label{eq:LS}\tag{\textbf{Linear Swaps}}
		\ireda \cdot   \eredc  \ \subseteq\  \eredc \cdot \reda^*     \quad\text{ and }\quad \iredc \cdot   \ereda  \ \subseteq\  \ereda \cdot \redc^* 
	\end{equation}
	
\end{thmm}


\paragraph{Extensions of the bang calculus.}
Following \cite{FaggianGuerrieri21}, we now consider a calculus $(\Bang, \red)$, where $\red \eq \redbb\cup \redc$ and $\redc$ is the contextual closure of a new rule $\Root{\gamma}$. 
\Cref{thm:modular}
{states} that the  compound system $\redbb\cup \redc$ satisfies surface factorization 
if $\F\sredbb\nsredbb$, $\F\sredc\nsredc$, 
and  the two linear swaps hold. We  know that  $\F\sredbb\nsredbb$ always hold. We now show that verifying the linear swaps  
reduces to a single simple test, leading to \Cref{prop:test}.

First, we observe that each  linear swap condition can be tested by considering for the surface step  only $\mapsto$, that is, only the closure of $\mapsto$ under \emph{empty} context.
This is expressed in the following lemma, where we  include also a useful variant.


\begin{lem}[Root linear swaps]\label{l:surf_swaps} In $\Bang$, let  $\reda, \redc$ be the contextual closure of rules $\rreda,\rredc$.
	
	\begin{enumerate}
		\item 	$\nsreda \cdot \rredc \subseteq  {\sredc} \cdot \reda^* $ implies 
		$\nsreda \cdot \sredc \subseteq  {\sredc} \cdot \reda^* $. 
		
		\item Similarly,  $\nsreda \cdot \rredc \subseteq  {\sredc} \cdot \reda^= $ implies 
		$\nsreda \cdot \sredc \subseteq  {\sredc} \cdot \reda^= $. 
	\end{enumerate}
\end{lem}

\begin{proof}Assume $M \nsreda U \sredc N$. 
	If $U$ is the redex, the claim holds by assumption. Otherwise,
	we prove  $M{\sredc} \cdot \reda^* N $, by induction on the structure of $ U $.  Observe that  both $M$ and $N$ have the same shape as $U$ (by Property \ref{fact:shape} ).
	\begin{itemize}
		\item $U=U_1U_2$ (hence $M=M_1M_2$ and $N=N_1N_2$). We   have  two cases.   
		\begin{enumerate}
			\item Case  $U_1 \sredc N_1$.  By \Cref{fact:isteps}, either $M_1\reda U_1$ or $M_2\reda U_2$.
			
			\begin{enumerate}
				\item 	 Assume $M \defeq M_1M_2 \nsreda U_1 M_2\sredc N_1 M_2 \eqdef N$. 
				
				We have $M_1 \nsreda U_1 \sredc N_1 $,
				and we conclude by \ih.
				
				\item  Assume $M \defeq U_1M_2 \nsreda U_1 U_2\sredc N_1 U_2 \eqdef N$. 
				
				 Then $U_1M_2 \sredc N_1M_2\reda N_1U_2$.
			\end{enumerate}
			\item Case $U_2 \sredc N_2$. Similar to the above.	 
		\end{enumerate}
		\item $U=\lam x.U_0$ (hence $M=\lam x. M_0$ and $N=\lam x. N_0$). We conclude by  \ih.
	\end{itemize}
	Cases $U= \oc U_0$ or $U=x$ do not apply.
\end{proof}

As we    study $ \redbb\cup \redc$, one of the linear swap is 	$\nsredx{\gamma} \cdot \sredbb \subseteq  {\sredbb} \cdot \redc^*$. We show that \emph{any} $\redc$ linearly swaps after $\sredbb$ as soon as $\rredc$  is \emph{substitutive}. 
	
	\begin{lem}[Swap with $\sredbb$] \label{l:swap_after_b} 
	%
		If  	  $\rredc$ is   substitutive, then 
		$\nsredx{\gamma} \cdot \sredbb \subseteq  {\sredbb} \cdot \redc^* $  always holds. 
	\end{lem}
	\begin{proof} We prove $ \nsredx{\gamma} \cdot  \mapsto_{\bbeta} ~\subseteq  ~ {\sredbb} \cdot \redc^* $, and conclude by 
		\Cref{l:surf_swaps}.
		
		Assume  $M \nsredx{\gamma} (\lam x.P) !Q \rredbb P \subs x Q$. We want to prove   $  M{\sredbb} \cdot \redc^* P \subs x Q$.  By \Cref{fact:isteps},  $M=M_1M_2$ and {either $M_1= \lam x.P _0\redc  \lam x.P$ or $M_2= \oc Q_0 \redc \oc Q$.}
		\begin{itemize}
			\item 	In the first case,  $M=(\lam x.P_0)!Q$, with $P_0 \redc P$. 
			So, $M=(\lam x.P_0)!Q \rredbb P_0 \subs x Q $ and we conclude by substitutivity of $\redc$ (\Cref{fact:subs}.\ref{fact:subs-function}).
			\item In the second case, $M=(\lam x.P)!Q_0$ with $Q_0 \redc Q$. 
			Therefore $M=(\lam x.P)!Q_0 \rredbb P \subs x {Q_0} $, and we conclude by \Cref{fact:subs}.\ref{fact:subs-argument}.
			\qedhere
		\end{itemize}	
\end{proof}

Summing up, since surface factorization for $ \bbeta$ is known, we obtain the following compact  test for surface  factorization in extensions of $\redbb$.

\begin{prop}[A modular test for surface factorization]\label{prop:test-bang} 	
	Let  	$\redbb$  be $\betab$-reduction and $\redc$ be the contextual closure of a rule $\rredc $.  	The reduction	$\redbb \cup \redc$  satisfies surface factorization if:
	\begin{enumerate}
		\item \emph{Surface factorization of $\redc$}: ~~
		$\redc^* \, \, \subseteq ~\sredxStar{\gamma} \cdot \nsredxStar{\gamma}$
		\item  $\rredc$ is  \emph{substitutive}:  ~~ 	$R \rredc R' 
		\text{ implies } R \subs x Q \rredc R'\subs x Q.$
		\item \emph{Root linear swap}: ~~ $\nsredbb \cdot \rredc  \subseteq \ \rredc \cdot\redbb^* $.
	\end{enumerate}
\end{prop}

\subsection{Restriction to computations.}
\label{subsect:restriction}

In $\Com$, let $\Root{\Rule}$ be a rule and $\redx{\Rule}$ be its contextual closure.
The restriction of reduction to computations preserves $\redx{\Rule},  \sredx{\Rule}, \nsredx{\Rule}, \wredx{\Rule}, \nwredx{\Rule}$ steps.
Thus, all properties that hold for $(\Bang, \redx{\Rule})$ (e.g.~\Cref{fact:isteps,cor:redex}) also hold for $(\Com, \redx{\Rule})$.  

In particular, \Cref{prop:test} is immediate consequence of \Cref{prop:test-bang}.


%

%

%

\section{Properties of reduction in  $\lc$ }

%
%
%
%

We now consider $\lc$, that  is $ (\Com, \redlc)$. 
As we have just seen above, the  properties  we have studied in \Cref{app:bang}  also hold when  restricting reduction to computations.
Moreover, $\lc$ satisfies also specific  properties that do not hold in general, as the following.

\begin{lem}
	\label{lemma:id-redex}
	Let $M\in \ComTerm$ and  $M\nsredx{\lcc} \, L$: 	$M$ is a $\id$-redex (resp. a $\ii$-redex) if and only if $L$ is.
\end{lem}

\begin{proof} 
	If $M$ is a $\id$-redex, this means that $M=(\lambda z.\oc z)P\nsredx{\gamma}  (\lambda z.\oc z)N=L$ where $P\nsredx{\gamma}N$, hence $L$ is a $\id$-redex. Moreover, if $M$ is a $\ii$-redex, then $P\neq \oc V$, hence by \cref{fact:isteps} $L\neq \oc V'$ for any $V'$. Thus $L$ is a $\ii$-redex.
	
	 Let us prove that if $L$ is a $\id$-redex, so is $M$.
	Since  $L\eq  (\lam z.!z) N$, by  \cref{fact:isteps},  $M$ is  an application; we have the following cases:
	\begin{enumerate}[(i)]
		\item either $M\eq (\lam z. P)N \nsredx{\gamma}  (\lam z.\oc z) N  $ where $P\nsredx{\gamma} \oc z$\,;
		\item or $M\eq  (\lam z.\oc z) P \nsredx{\gamma}  (\lam z.\oc z) N  $ where $P\nsredx{\gamma} N$.
	\end{enumerate}
	Case \textit{(i.)} is impossible because otherwise $P\eq \oc V$ for some value $V$, by \Cref{fact:isteps}, such that $V\red_{\gamma} z$, but such a $V$ does not exist. Therefore we are necessarily in  
	case \textit{(ii.)}, \ie $M$ is a $\id$-redex. 
	Moreover, if $L$ is a $\ii$-redex, then $N\neq \oc V$, hence $N$ is an application, and so is $P$ by \Cref{fact:isteps}.
\end{proof}

Note that \Cref{lemma:id-redex} is false if we replace the hypothesis $M\nsredx{\lcc} \, L$ with $M\nwredx{\lcc} \, L$.
Indeed, consider $M = (\lam x. (\lam y. \oc y)\oc x)N \nwredx{\lcc} (\lam x. \oc x) N = L$: $L$ is a $\id$-redex but $M$ is not.

\begin{lem}\label{lem:iotaToVar}
	There is no $M \in \ComTerm$ such that $M\red_{\ii} \oc x$.
\end{lem}

\begin{proof}
By induction on $M \in \ComTerm$, proving that for every $M$ such that $M\red_{\ii}N$, $N \neq \oc x$.
\end{proof}

\subsection{Postponement of $\iota$, Technical Lemmas.}


\newcommand{\reduno}{\red_{\beta_1}}
\newcommand{\reddue}{\red_{\beta_2}}

\newcommand{\circled}[1]{\raisebox{.5pt}{\textcircled{\raisebox{-.9pt} {#1}}}}
\newcommand{\parreduno}{\Rightarrow_{\normalfont\scalebox{.5}{\circled{1}}}}
\newcommand{\parreddue}{\Rightarrow_{\normalfont\scalebox{.5}{\circled{2}}}}
\newcommand{\parredtre}{\Rightarrow_{\normalfont\scalebox{.5}{\circled{3}}}}


\begin{lem}[$\iota$ vs. $\beta_1$]\label{lem:iotaVSbeta1}
	$ M \redi L \reduno N \mbox{ implies } M \reduno^* 	\cdot  \redi^= N  $
\end{lem}

\begin{proof}We set the notation  $\parreduno \ := \ \reduno^* \cdot \redi^=$. 
	
The proof is by induction on $L$. 
	Cases:
	\begin{itemize}
		\item  $L = (\lambda x. \oc x)\oc V \mapsto_{\beta_1} \oc V = N$. 
		Then, there are two possibilities. 
		
		Either  $M = (\lambda z. \oc z)L \mapsto_\ii   L$  then
		$$
		M = (\lambda z. \oc z)((\lambda x. \oc x)\oc V) \reduno (\lambda z. \oc z)\oc V \reduno \oc V = N.
		$$
		
		Or $M = (\lambda x. \oc x)\oc W $ with $\oc W \redi \oc V$, and then
		$$M =(\lambda x. \oc x)\oc W  \reduno\oc W \redi \oc V = N$$
		
		The case $M = (\lambda x.P)\oc V $ with $P \red_{\ii} \oc x$ is impossible by \Cref{lem:iotaToVar}.
		
		\item $L= \oc \lam x. P \reduno \oc \lam x. P'= N$ where $P\reduno P'$. 
		In this case note that necessarily $M = \oc \lam x. Q$ where $Q\redi P$. 
		Otherwise it should have been $M = (\lam z.\oc z) L\mapsto_{\ii} L$ but the $\ii$ step is impossible because $L= \oc \lam x. P $. 
		By \ih, since $Q\redi P\reduno P'$, we have $Q\parreduno P$; hence $M \parreduno N$.
		
		\item $L=VP\reduno V'P'=N$ where $\reduno$ is not root steps, that is:
		\begin{enumerate}[(a)]
			\item\label{a:beta1}  either $V\reduno V'$ and $P = P'$;
			\item\label{b:beta1}  or $V = V' $ and $P \reduno P' $. 
		\end{enumerate}
		By \Cref{fact:shape}, $M$, $L$, $N$ are applications. So, $M$ has the following shape:
		\begin{enumerate}
			\item\label{1:beta1} $M = VQ$ with $Q\redi P$
			\item\label{2:beta1}$M = WP$ with $W\redi V$
			\item\label{3:beta1} $M=(\lam x.\oc x)(VP)\mapsto_{\ii} VP=L $
		\end{enumerate}
		
		We distinguish six sub-cases:
		
		\begin{description}
			\item[Case \ref{a:beta1}\ref{1:beta1}] We have $M= VQ \reduno V'Q \redi V'P=N$, switching  the steps $\redi$ and $\reduno$, directly.
			\item[Case \ref{b:beta1}\ref{1:beta1}] $Q \redi P \reduno P'$, then the thesis follows by i.h.: $Q\parreduno P'$ and then $M=VQ \parreduno VP' = N$.
			\item[Case \ref{a:beta1}\ref{2:beta1}] $W \redi V \reduno V'$, then the thesis follows by i.h.: $W\parreduno V'$ and then $M=WP \parreduno V'P = N$.
			\item[Case \ref{b:beta1}\ref{2:beta1}] We have $M= WP \reduno WP' \redi VP'=N$, switching  the steps $\redi$ and $\reduno$, directly.
			\item[Case \ref{a:beta1}\ref{3:beta1}] $M=(\lam x. \oc x)(VP)\reduno (\lam x. \oc x)(V'P)\mapsto_{\ii} V'P=N$.
			\item[Case \ref{b:beta1}\ref{3:beta1}]$M=(\lam x. \oc x)(VP)\reduno (\lam x. \oc x)(VP')\mapsto_{\ii} VP'=N$.
			\qedhere
		\end{description}
	\end{itemize}
\end{proof}	


\begin{lem}[$\iota$ vs. $\beta_2$]\label{lem:iotaVSbeta2}
	$ M\redi L \reddue N \mbox{ implies } \reduno^* \cdot \reddue^= \cdot \reduno^* \cdot \redi^* N$
\end{lem}

\begin{proof} We set the notation $\parreddue ~\Coloneqq ~\reduno^* \cdot \reddue^= \cdot \reduno^* \cdot \redi^*$.
	The proof is by  induction on $L$. Note that if $L=\oc x$ there is no $\beta_2$ reduction from it, so this case is not in the scope of the induction.
	Cases:
	\begin{itemize}
		\item  $L = (\lambda x. P')\oc V' \mapsto_{\beta_2} P'\Subst{V'}{x} = N$. 
		Then, there are three possibilities.
		\begin{enumerate}[(i)]
			\item\label{i:beta2} $M=(\lam x.P) \oc V'$ with $P\redi P'$
			\item\label{ii:beta2} $M=(\lam x.P') \oc V$ with $V\redi V'$
			\item\label{iii:beta2} $M=(\lam x.\oc x) L \mapsto_{\ii} L$
		\end{enumerate} 
		So by analyzing each of the three cases above, we can postpone the $\redi$ step as follows:
		\begin{description}
			\item[Case \ref{i:beta2}] $M=(\lam x.P)\oc V\mapsto_{\beta_2} P\Subst{V}{x}\redi P'\Subst{V}{x}$ where the last reduction step is possible by \Cref{fact:subs}.\ref{fact:subs-function}. Note that $P\neq \oc x$ otherwise would not possible $P\redi P'$, as assumed.
			\item[Case \ref{ii:beta2}] $M=(\lam x.P)\oc V\mapsto_{\beta_2} P\Subst{V}{x}\redi^* P\Subst{V'}{x}$ where the last reduction step is possible by \Cref{fact:subs}.\ref{fact:subs-argument}.
			\item[Case \ref{iii:beta2}]  $M=(\lam x.\oc x) L\mapsto_{\beta_2}(\lam x.\oc x) N \mapsto_{\ii} N$
		\end{description}
		
		\item $L= \oc \lam x. P \reddue \oc \lam x. P'= N$ where $P\reddue P'$. In this case note that $M$ has necessary the shape $\oc \lam x. Q$ where $Q\redi P$. Otherwise $M$ should have been $(\lam z.\oc z) L\mapsto_{\ii} L$ but it is impossible by definition of $ \mapsto_{\ii}$ since $L= \oc \lam x. P $. 
		The thesis follows by induction, since we have $Q\redi P\reddue P'$, $Q\parreddue P$.
		
		\item $L=VP\reddue V'P'=N$ where $\reddue$ is not root steps, that is:
		\begin{enumerate}[(a)]
			\item\label{a:beta2}  either $V\reddue V'$ and $P = P'$;
			\item\label{b:beta2}  or $V = V' $ and $P \reddue P' $. 
		\end{enumerate}
		By \Cref{fact:shape}, $M$, $L$, $N$ are applications. So, $M$ has the following shape:
		\begin{enumerate}
			\item\label{1:beta2} $M = VQ$ with $Q\redi P$
			\item\label{2:beta2}$M = WP$ with $W\redi V$
			\item\label{3:beta2} $M=(\lam x.\oc x)(VP)\mapsto_{\ii} VP=L $
		\end{enumerate}
		
		We distinguish six sub-cases:
		
		\begin{description}
			\item[Case \ref{a:beta2}\ref{1:beta2}] We have $M= VQ \reddue V'Q \redi V'P=N$, switching  the steps $\redi$ and $\reddue$, directly.
			\item[Case \ref{b:beta2}\ref{1:beta2}] $Q \redi P \reddue P'$, then the thesis follows by i.h., that is: $Q\parreddue P'$ and then $M=VQ \parreddue VP' = N$.
			\item[Case \ref{a:beta2}\ref{2:beta2}] $W \redi V \reddue V'$, then the thesis follows by i.h., that is: $W\parreddue V'$ and then $M=WP \parreddue V'P = N$.
			\item[Case \ref{b:beta2}\ref{2:beta2}] We have $M= WP \reddue WP' \redi VP'=N$, switching  the steps $\redi$ and $\reddue$, directly.
			\item[Case \ref{a:beta2}\ref{3:beta2}] $M=(\lam x. \oc x)(VP)\reddue (\lam x. \oc x)(V'P)\mapsto_{\ii} V'P=N$
			\item[Case \ref{b:beta2}\ref{3:beta2}]$M=(\lam x. \oc x)(VP)\reddue (\lam x. \oc x)(VP')\mapsto_{\ii} VP'=N$
		\end{description}
	\end{itemize}
\end{proof}

\begin{lem}[$\iota$ vs. $\sigma$]\label{lem:iotaVSsigma}
	$ M \redi L \reds N \mbox{ implies } M \red_{\sigma}^* 	\cdot  \redi^= N  $
\end{lem}

\begin{proof} We set the notation  $\parredtre ~ \Coloneqq ~(\reds \cup\reduno)^* \cdot \redi^=$.
	
	The proof is by induction on $L$. We distinguishing if the last $L \reds N$ is a root step or not.
	
	If $L\mapsto_{\sigma} N$ then $L=V((\lam x. P)Q)$ and $N= (\lam x. VP)Q $. Thus, there are seven cases for $M$:
	
	\begin{enumerate}[(i)]
		\item\label{iv:sigma} $M=(\lam z.\oc z)(V((\lam x. P)Q)) $;
		\item\label{v:sigma} $M=V((\lam z.\oc z)((\lam x. P)Q)) $;
		\item\label{vi:sigma} $M=V((\lam x. (\lam z.\oc z)P)Q) $;
		\item\label{vii:sigma} $M=(V((\lam x. P)((\lam z.\oc z)Q))) $;
		\item\label{i:sigma} $M=W((\lam x. P)Q)$ with $W\redi V$ and $\redi$ is not a root step;
		\item\label{ii:sigma} $M=V((\lam x. R)Q)$ with $R\redi P$ and $\redi$ is not a root step;
		\item\label{iii:sigma} $M=V((\lam x. P)R)$ with $R\redi Q$ and $\redi$ is not a root step.
	\end{enumerate} 
	So by analyzing each of the seven cases above, we can postpone the $\redi$ step as follows:
	\begin{description}
		\item[Case \ref{iv:sigma}] $M=(\lam z.\oc z)(V((\lam x. P)Q)) \reds (\lam z.\oc z)((\lam x. VP)Q)  \mapsto_{\ii} (\lam x. VP)Q =N$.
		\item[Case \ref{v:sigma}] $
		M=V((\lam z.\oc z)((\lam x. P)Q))  \reds V((\lambda x. (\lam z.\oc z) P)Q) \reds (\lambda x. V((\lam z.\oc z)P))Q \red_\gamma (\lambda x. VP)Q = N$
		where in the last step $\gamma$ is $\iota$ or $\beta_1$ depending on whether $P$ is of the form $\oc W$ or not.
		\item[Case \ref{vi:sigma}] $M = V((\lambda x. (\lam z.\oc z) P)Q) \reds (\lambda x. V((\lam z.\oc z) P))Q \red_\gamma (\lambda x. VP)Q = N$
		where in the last step $\gamma$ is $\iota$ or $\beta_1$ depending on whether $P$ is of the form $\oc W$ or not.
		\item[Case \ref{vii:sigma}] $	M = V((\lambda x. P)((\lam z.\oc z) Q)) \sreds (\lambda x. VP)((\lam z.\oc z) Q) \red_\gamma (\lambda x. VP)Q = N	$
		where in the last step $\gamma$ is $\iota$ or $\beta_1$ depending on whether $Q$ is of the form $\oc W$ or not.
		\item[Case \ref{i:sigma}] $M=W((\lam x. P)Q)\mapsto_{\sigma} (\lam x. WP)Q \redi (\lam x. VP)Q =N$.
		\item[Case \ref{ii:sigma}] $M=V((\lam x. R)Q)\mapsto_{\sigma} (\lam x. VR)Q \redi (\lam x. VP)Q =N$.
		\item[Case \ref{iii:sigma}]  $M=V((\lam x. P)R)\mapsto_{\sigma} (\lam x. VP)R \redi (\lam x. VP)Q =N$.
	\end{description}
	
	Consider the case $L= \oc \lam x. P \reds\oc \lam x. P'= N$ with $P\reds P'$. 
	So, note that $M$ has necessary the shape $\oc \lam x. Q$ where $Q\redi P$. Otherwise $M$ should have been $(\lam z.\oc z) L\mapsto_{\ii} L$ but it is impossible by definition of $ \mapsto_{\ii}$ since $L= \oc \lam x. P $. 
	The thesis follows by \ih, since $Q\redi P\reds P'$, $Q\parredtre P$.
	
	The last case to consider is $L=VP\reds V'P'=N$ where $\reds$ is not root steps, that is:
	\begin{enumerate}[(a)]
		\item\label{a:sigma}  either $V\reds V'$ and $P = P'$;
		\item\label{b:sigma}  or $V = V' $ and $P \reds P' $. 
	\end{enumerate}
	By \Cref{fact:shape}, $M$, $L$, $N$ are applications. So, $M$ has one of the following shapes:
	\begin{enumerate}
		\item\label{1:sigma} $M=(\lam z.\oc z)L \mapsto_{\ii} VP=L $;
		\item\label{2:sigma}$M = WP$ with $W\redi V$;
		\item\label{3:sigma}$M = VQ$ with $Q\redi P$.
	\end{enumerate}
	
	Hence, combining \Cref{a:sigma,b:sigma} with \Cref{1:sigma,2:sigma,3:sigma}, we distinguish six sub-cases:
	
	\begin{description}
		\item[Case \ref{a:sigma}\ref{1:sigma}] $M=(\lam x. \oc x)(VP)\reds (\lam x. \oc x)(V'P)\mapsto_{\ii} V'P=N$.
		\item[Case \ref{b:sigma}\ref{1:sigma}] $M=(\lam x. \oc x)(VP)\reds (\lam x. \oc x)(VP')\mapsto_{\ii} VP'=N$.
		\item[Case \ref{a:sigma}\ref{2:sigma}] $W \redi V \reds V'$, then the thesis follows by i.h: $W\parredtre V'$ and then $M=WP \parredtre V'P = N$.
		\item[Case \ref{b:beta2}\ref{2:beta2}] We have $M= WP \reds WP' \redi VP'=N$, switching  the steps $\redi$ and $\reds$, directly.
		\item[Case \ref{a:beta2}\ref{3:beta2}] We have $M= VQ \reds V'Q \redi V'P=N$, switching  the steps $\redi$ and $\reds$, directly.
		\item[Case \ref{b:beta2}\ref{3:beta2}] $Q \redi P \reds P'$, then the thesis follows by i.h.: $Q\parredtre P'$ and then $M=VQ \parredtre VP' = N$.
		\qedhere
	\end{description}
	
\end{proof}


%% file: App_Normalization.tex
\section{Normalization of $\lc$}\label{app:normalization}

\lemsurfacenf*
\begin{proof}By easy induction on the shape of $M$. Observe that $M$ and $N$ have the same shape, because the step $M\nered N$ is not a root step.
	\begin{itemize}
		\item   $M= \oc V$, and $N= \oc V'$: the claim is trivial.
		
		\item $M=VP$ and $N=V'P'$.  Either $V\nered V'$ (and $P'=P$) or $P\nered P'$ (and $V=V'$). Assume $M=VP$ is $\ex$-normal. Since $V$ and $P$ are $\ex$-normal, by  \ih so are  $V'$ and $P'$. Moreover, $N$ is not a redex, by \Cref{cor:redex}, so $N$ is normal. Assuming $N=V'P'$ normal is similar.
		\qedhere
	\end{itemize}
\end{proof}

\begin{fact}\label{fact:id_diamond}
	The reduction $\ii$ is quasi-diamond. Therefore, if $S\redi^k N$ where $N$ is $\ii$-normal, then any maximal $\ii$-sequence from $S$ ends in $N$, in $k$ steps.
\end{fact}

\lemredi*

\begin{proof}
	Easy to prove by induction on the structure of terms.	
\end{proof}

\begin{fact}[Shape preservation of $\ii$-sequences]\label{fact:redi_shape}
	If $S$ is not an $\ii$-redex, and $S\redi^k N$ then no term in the sequence is an $\ii$-redex, and so $N$ has the same shape as $S$:  
	\begin{enumerate}
		\item 	$S=!(\lam x. Q)$ implies  $N=!(\lam x. N_Q)$. Moreover, $Q\redi^k N_Q$. 
		\item  $S=xP$ implies $N=xN_P$. Moreover, $P\redi^k N_P$.
		\item  $S=(\lam x. Q)P$ implies $N=(\lam x. N_Q)N_P$. Moreover, $Q\redi^{k_1} N_Q$,  $P\redi^{k_2} N_P$ and $k=k_1+k_2$. 
	\end{enumerate}
\end{fact}


\lemNorm*
\begin{proof}The proof is by induction on $M$.

	Assume $M$ is a is a $\sigma$-redex, \ie $M=V((\lam x.P)L)$ where $V$ is an abstraction.
	\begin{itemize}
		\item If  $M$ is also an $\ii$-redex, then $V= \II$, and:
		\begin{enumerate}
			\item If $P = \oc U$, then $M = \II ( (\lam x. \oc U) L) \redi (\lam x. \oc U) L \redi^{k-1} N$ and  $M\reds (\lam x. \II \oc U)L \redb (\lam x. !U) L$.
			\item  If $P \neq \oc U$, then  $M = \II((\lam x.P)L) \redi (\lam x.P)L \redi^{k-1} N$ and  $M\reds (\lam x. \II P)L \redi (\lam x.P)L $.
		\end{enumerate}
		\item Otherwise, if  $M= V(\II L)$, then 
		$M\redi VL\redi^{k-1} N$ and $M \reds (\lam z. V!z)L \redb VL$.   
		
		\item No other case is possible, because  $M = V((\lam x.P)L)$ not an $\ii$-redex implies (by \Cref{fact:redi_shape}) that $N=N_1N_2$, with 
		$N_1\in \abs$. If $(\lam x.P)\not =\II$, then  $N$ would be a $\sigma$-redex, because  again $N_2=N_2'N_2''$ with $N_2'\in \abs$  (by \Cref{fact:redi_shape}).
	\end{itemize}
	Assume  $M$ is not a $\sigma$-redex. 
	We examine the shape of $S$ and use \Cref{fact:redi_shape}.
	\begin{itemize}
		\item $M = \oc\lam x.Q$. We have  $N= \oc\lam x. N_Q$, $Q\redi^k N_Q$, where $Q$ is not $\sigma$-normal, and $N_Q$ is $\sigma\ii$-normal. We conclude by \ih.
		
		\item $M=\II P$.  We have  $\II P\redi P \redi^{k-1} N$. Since  $P$ is not $\sigma$-normal, we use the   \ih on $P\redi^{k-1} N$,
		obtaining that $P\redi N' \redi^{k-2} N$ and $P\reds P'\redx{\betac\ii} N'$. Therefore also $\II P\reds \II P'\redx{\betac\ii} \II N'\ired N' \ired^{k-1} N$.
		
		\item   $M=(\lam x. Q)P$ ($M$ is not an $\ii$-redex). We have  $N=(\lam x. N_Q)N_P$, where $N_Q$  and $N_P$ are $\sigma\ii$-normal. We distinguish two cases.
		\begin{itemize}
			\item If $Q$ is not $\sigma$-normal, we note that $Q\redi^{k_1} N_Q$, and conclude by \ih Indeed, by \ih, 
			we obtain  that  $Q\redi N_Q'\redi^{k_1-1} N_Q $,  and $Q\reds Q' \redx{\betac\ii} N_Q'$.   So,
			$(\lam x.Q)P\redi ( \lam x. N_Q')P\redi^{k_1-1} (\lam x. N_Q) P \redi^{k_2}(\lam x. N_Q) N_P$, and 
			$(\lam x. Q)P\reds (\lam x. Q')P \redx{\betac\ii}   (\lam x. N_Q')P$.
			\item If $P$ is not  $\sigma$-normal, we note that $P\redi^{k_2} N_P$, and conclude by \ih
			\qedhere
		\end{itemize}
		
	\end{itemize}
\end{proof}

\lemmastrongnormalizingidsigma*
\begin{proof}
	We define two sizes $\size{M}$ and $\sizeaux{M}$ for any term $M$.
	\begin{align*}
		\size{x} &= 1
		&
		\sizeaux{x} &= 1
		\\
		\size{\la{x}M} &= \size{M} + 1
		&
		\sizeaux{\la{x}M} &= \sizeaux{M} + \size{M}
		\\
		\size{VM} &= \size{V} + \size{M} 
		&
		\sizeaux{V{M}} &= \sizeaux{V} + \sizeaux{M} + 2\size{V}\size{M} 
		\\
		\size{\oc M} &= \size{M}
		&
		\sizeaux{\oc {M}} &= \sizeaux{M}
	\end{align*}
	
	Note that $\size{M} > 0$ and $\sizeaux{M} > 0$ for any term $M$.
	It easy to check that if $M \redid \cup \reds N$, then $(\size{N}, \sizeaux{N}) <_\textup{lex} (\size{M}, \sizeaux{M})$, where $<_\textup{lex}$ is the strict lexicographical order on $\Nat^2$.
	Indeed, if $M \redid N$, then $\size{M} > \size{N}$; and if $M \reds N$ then $\size{M} = \size{N}$ and $\sizeaux{M} > \sizeaux{N}$.
	The proof is by straightforward induction on $M$. 
	We show only the root-cases, the other cases follow from the \ih immediately.
	\begin{itemize}
		\item If $(\la{x}\oc x)M \Root{\id} M$ then $\size{(\la{x}\oc x)M} = \size{\la{x}\oc x} + \size{M} > \size{M}$.
		\item If $(\la{x}M)((\la{y}N)L) \Root{\sigma} (\la{y}(\la{x}M)N)L$ then clearly \\
		$\size{(\la{x}M)((\la{y}N)L)} = \size{(\la{y}(\la{x}M)N)L}$ and
	\end{itemize}
	\begin{small}
		\begin{align*}
			&\sizeaux{(\la{x}M)((\la{y}N)L)} 
			\\
			&= \sizeaux{\la{x}M} + \sizeaux{(\la{y}N)L} + 2 \size{\la{x}M} \size{(\la{y}N)L}  
			\\
			&= \sizeaux{\la{x}M} + \sizeaux{\la{y}N} + \sizeaux{L} + 2\size{\la{y}N}\size{L} + 2 \size{\la{x}M} \size{(\la{y}N)L} 
			\\
			&= \sizeaux{\la{x}M} + \sizeaux{N} + \size{N} + \sizeaux{L} + 2\size{N}\size{L} + 2\size{L}  + 2 \size{\la{x}M} \size{\la{y}N} + 2 \size{\la{x}M} \size{L} 
			\\
			&= \sizeaux{\la{x}M} + \sizeaux{N} + \size{N} + \sizeaux{L} + 2\size{N}\size{L} + 2\size{L} ´ + \underline{2 \size{\la{x}M}} + 2 \size{\la{x}M} \size{N} + 2 \size{\la{x}M} \size{L}  
			\\
			&> \sizeaux{\la{x}M} + \sizeaux{N} + 2\size{\la{x}M}\size{N}  + \underline{\size{\la{x}M}} + \size{N} + \sizeaux{L} + 2\size{\la{x}M}\size{L} + 2\size{N}\size{L} + 2\size{L} 
			\\
			&= \sizeaux{\la{x}M} + \sizeaux{N} + 2\size{\la{x}M}\size{N}  + \size{\la{x}M} + \size{N} + \sizeaux{L} + 2\size{(\la{x}M)N}\size{L} + 2\size{L} 
			\\
			&= \sizeaux{(\la{x}M)N} + \size{(\la{x}M)N} + \sizeaux{L} + 2\size{\la{y}(\la{x}M)N}\size{L} 
			\\
			&= \sizeaux{\la{y}(\la{x}M)N} + \sizeaux{L} + 2\size{\la{y}(\la{x}M)N}\size{L} 
			\\
			&= \sizeaux{(\la{y}(\la{x}M)N)L} 
			\qquad
			\qedhere
		\end{align*}
	\end{small}
\end{proof}

\section{Notational Equivalence between $\lc$ and  \cite{deLiguoroTreglia20}}\label{app:crbiso}

\newcommand{\lstar}{\lambda_\Bind}

Here we recall the calculus introduced in \cite{deLiguoroTreglia20} (see also our \Cref{sect:intro}), henceforth denoted by $\lstar$, and formalize that $\lstar$ is \emph{isomorphic} to the computational core $\lc$. 
In other words, $\lc$  (as defined in \Cref{sec:compcal}) is nothing but another presentation of $\lstar$, with just a different notation.

First, we recall the syntax of $\lstar$, with $\Unit$ and $\Bind$ operators:
\[\begin{array}{r@{\hspace{0.7cm}}rll@{\hspace{1cm}}l}
	\ValTerm^\Bind: & V, W & \Coloneqq & x \mid \lambda x.M & \mbox{(values)} \\ [1mm]
	\ComTerm^\Bind: & M,N & \Coloneqq & \Unit V \mid M \Bind V & \mbox{(computations)}
\end{array}\]

We set $\Term^\Bind \defeq \ValTerm^\Bind \cup \Com^\Bind$.
Contexts are defined in \Cref{sect:intro}.
Reductions in $\lstar$ are the contextual closures of the rules \eqref{eq:beta-c}, \eqref{eq:id}, and \eqref{eq:sigma} on p.~\pageref{eq:beta-c}, oriented left-to-right, giving rise to reductions $\redbc$, $\redid$, and $\reds$, respectively.
We set $\redlc \, = \, \redbc \cup \redid \cup \reds$.

Consider the translation $\trb{\cdot}$ from $\Term^\Bind$ to $\Term$, and conversely, the translation $\trc{\cdot}$ from $\Term$  to $\Term^\Bind$:

\begin{table}[h!]
	\begin{tabular}{l|ll}
		~                & $\trb{\cdot}:\Term^\Bind\to \Term$                                         & $\trc{\cdot}:\Term\to\Term^\Bind$                                                 \\\hline\hline   
		variables        & $ \trb x =x $                                                      & $ \trc x =x  $                                                                \\
		abstraction      & $ \trb {\lam x.P} = \lam x. \trb P $ & $ \trc {\lam x.M} = \lam x. \trc M    $                      \\
		returned values  & $ \trb {\Unit V} = \oc (\trbp V)  $        & $ \trc{\oc V} =  {\Unit \! (\trcp V)} $  \\
		bind/application & $ \trb {P\Bind V} = \trbp V  \trbp P  $             & $ \trc {V M}  = \trcp{M} \Bind \trcp{V}   $               \\
	\end{tabular}
\caption{Translations between $\lc$ and $\lstar$}
\end{table}

Essentially, translating terms of $\lstar$ in terms of the computational core $\lc$ rewrites $\Unit$ as $\oc$, and reverts the order in $\Bind$.

\begin{prop}\label{prop:ctbiso}
	The following holds:
	\begin{enumerate}
		\item $\trb{\trcp{M}} = M$ for every term $M$ in $\Term$;
		\item $\trc{\trbp{P}} = P$ for every term $P$ in $\Term^\Bind$;
		\item for any $\gamma \in \{\betac, \id, \sigma, \CCsym\}$, if $M \redc N$ in $\lc$, then $\trcp{M}\redc \trcp{N}$ in $\lstar$;
		\item for any $\gamma \in \{\betac, \id, \sigma, \CCsym\}$, if $P\redc Q$ in $\lstar$, then $\trbp{P}\redc \trbp{Q}$ in $\lc$.
	\end{enumerate}
\end{prop}
\begin{proof}
	Immediate by definition unfolding.
\end{proof}

\section{Computational versus Call-by-Value}
\label{subsect:translations-comp-cbv}
\SLV{}{\CF{Qui sembra  inutile cambiare la definizione di $\redbv$ cambiando la definizione di contesto. 
		Ma in ogni caso:\\ perche' questa def?	
		\begin{align*}
			\cc&\Coloneqq \hole{} \mid V \cc \mid \cc M \mid \lam x.\cc  &\qquad \text{Contexts}
		\end{align*}
		{non e' vero che se nel buco si mette una computazione si ottiene una computazione}
		{ Per esempio $\hole{}(xx)$ produce \RED{ $(xx)(xx)$ } \\}
		Mi aspetterei:
		\begin{align*}
			\cc&\Coloneqq \hole{} \mid V \cc \mid (\lam x.\cc) M \mid \lam x.\cc  &\qquad \text{Contexts}
		\end{align*}
	}
	\medskip
}
Here we formalize the relation between a fragment of the computational core $\lc$ and Plotkin's call-by-value (CbV, for short) $\lambda$-calculus \cite{Plotkin'75}.

The fragment of $\lc$ that includes just $\betac$ as only reduction rule, i.e. $(\ComTerm, \red_{\betac})$, is also isomorphic to the \emph{kernel} of Plotkin's CbV $\lambda$-calculus, 
which is the  restriction of $\redbv$ (see \Cref{sec:CbV}) to 
the set of terms  $ \ComTerm^{\textit{v}}\subseteq \Lambda $ defined as follows (note that $\ValTerm^{\,\textit{v}} \subseteq \ComTerm^{\textit{v}}$).
\[\begin{array}{r@{\hspace{0.7cm}}rll@{\hspace{1cm}}l}
	\ValTerm^{\,v}: & V,W &\Coloneqq & x\mid \lam x.M 
	\\
	\ComTerm^{\textit{v}}: & M,N,L &\Coloneqq & V \mid VM
\end{array}\]

\SLV{}{and the reduction $\redbv$ defined as the closure of the rule 
	\begin{align*}
		(\lambda x.M) V \Root{\betav} M\Subst{V}{x}
	\end{align*}
	under the contexts defined as 
	\begin{align*}
		\cc& \Coloneqq \hole{} \mid V \cc \mid \cc M \mid \lam x.\cc  &\qquad \text{Contexts}
	\end{align*}
}


To establish such an isomorphism,
consider the translation $\trb{\cdot}$ from $\Term$ to $\ComTerm^{\textit{v}}$, and conversely, the translation $\trc{\cdot}$ from $\ComTerm^{\textit{v}}$ to $\Term$. 

\begin{table}[h!]
	\begin{tabular}{l|ll}
		~                & $\trb{\cdot} \colon\Term\to \ComTerm^{\textit{v}}$                                         & $\trc{\cdot} \colon \ComTerm^{\textit{v}}\to\Term$                                                 \\\hline\hline  
		variables        & $ \trb x =x $                                                      & $ \trc x = \oc x  $                                                                \\
		abstraction      & $ \trb {\lam x.P} = \lam x. \trb P $ & $ \trc {\lam x.M} = \oc\lam x. \trc M    $                      \\
		returned values  & $ \trb {\oc V} =  \trbp V  $  &   \\
		bind/application & $ \trb {VP} = \trbp V  \trbp P  $             & $ \trc {V P}  =\begin{cases}
			x\trcp{P} \quad\quad\quad\quad \mbox{ if }V\eq x\\
			(\lam x.\trcp{Q})\trcp{P} \quad\mbox{ if }V\eq \lam x.Q
		\end{cases}   $               \\
	\end{tabular}
\caption{Translations between the computational core $\lc$ and the kernel of the Call-by-Value $\lam$-calculus}
\end{table}

Essentially, the translation $\trb{\cdot}$ simply forgets the operator $\oc$, and dually the translation $\trc{\cdot}$ adds a $\oc$ in front of each value that is not in the functional position of an application.
These two translations form an isomorphism.

\begin{prop}\label{prop:ctviso}\hfill
	\begin{enumerate}
		
		\item $\trb{\trcp{M}} = M$, for every term $M \in\ComTerm^{\textit{v}}$;
		\item $\trc{\trbp{P}}  = P$,  for every term $P \in \Term$;
		\item $M \redbv N$ implies $\trcp{M}\redbc \trcp{N}$, for every terms $M, N \in\ComTerm^{\textit{v}}$;
		\item $P \redbc Q$ implies $\trbp{P}\redbv \trbp{Q}$,  for every terms $P, Q \in \Term$.
	\end{enumerate}
\end{prop}

\begin{proof}
	Immediate by definition unfolding.
\end{proof}

Observe also that  the restriction of \emph{weak context}  to $ \ComTerm^{\textit{v}} $ give exactly the grammar defined in \Cref{sec:compcal}, and this  for all three weak contexts ($\leftc, \rightc, \ww$), which all collapse in the same shape.
Thus, \Cref{prop:ctviso} also holds when $\redbv$ and $\redbc$ are replaced by $\wredbv$ and $\wredbc$, respectively.
As a consequence, since $\wredbv$ is deterministic in $ \ComTerm^{\textit{v}}$ and $\redbv$ weak factorizes (\Cref{thm:factorization_CbV}.\ref{p:factorization_CbV-left}),
\begin{property}[Properties of $\betac$ and its  \emph{weak}  restriction]
	\label{fact:factorization-betac}
	 In $\lc$:
	\begin{itemize}
	 \item reduction $\wredbc$ is deterministic;
	
	 \item reduction $\redbc$ satisfies weak factorization: ~~ $\redbc^*  \  \subseteq  \ \wredbc^* \cdot \nwredx{\betac}^*$.
	\end{itemize}
\end{property}

\paragraph{Call-by-Value versus  its Kernel.}
\label{subsect:translations-cbv-kernel}

The CbV kernel---and so $(\ComTerm,\redbc)$---is as expressive as the CbV $\lambda$-calculus, as we discuss below. 
This result  was 
already shown by  Accattoli  \cite{Acc15}.\footnote{Precisely, Accattoli studies the relation between   the kernel calculus $\lambda_{vker}$ and the \emph{value substitution calculus} $\lambda_{vsub}$, \ie CbV and the kernel extended with explicit substitutions. 
	The syntax is slightly different, but not in an essential way.}

With respect  to its  \emph{kernel}, Plotkin's
CbV $\lambda$-calculus is more liberal in that application is unrestricted (left-hand side need not be a value).
The kernel has the same expressive power as CbV calculus, because 
the full syntax of Plotkin's CbV can be encoded into the restricted one the CbV kernel, and because the CbV kernel can simulate every reduction sequence of Plotkin's full CbV.

Formally, consider the translation $\trap{(\cdot)}$ from Plotkin's CbV $\lambda$-calculus to its kernel.
\begin{align*}
	\trap {(x)} &=x &  
	\trap {(\lam x.P)} &= \lam x. \trap P
	&
	\trap {(PQ)} &= 
	\begin{cases}
		\trap {P}  \trap {Q}  &\text{if $P$ is a value;} \\
		(\lambda x. x \trap{Q})\trap{P} &\text{otherwise}.
	\end{cases}
\end{align*}

\begin{prop}[Simulation of the CbV $\lambda$-calculus into its kernel]
	For every term $P$ in Plotkin's CbV  $\lambda$-calculus, if $P \redbv Q$ then $\trap{P} \redbv^+ \trap{Q}$ and $\trcp{{\trap{P}}} \redbc^+ \trcp{{\trap{Q}}}$.
\end{prop}

%
%
%
